\documentclass[12pt]{amsart}
\usepackage{amssymb}
\usepackage{latexsym}
\usepackage{epic, eepic}
\usepackage{graphicx}

\newcommand{\talpha}{\tilde \alpha}
\newcommand{\tbeta}{\tilde \beta}
\newcommand{\tA}{\widetilde A}

\newcommand{\tchi}{\tilde \chi}

\newcommand{\teta}{\tilde \eta}

\newcommand{\tg}{\tilde g}
\newcommand{\tM}{\widetilde M}
\newcommand{\Ps}{P^\sharp}
\newcommand{\tPs}{{\widetilde P}^\sharp}
\newcommand{\tP}{\widetilde P}
\newcommand{\tPi}{\widetilde \Pi}

\newcommand{\tR}{\widetilde R}

\newcommand{\ts}{\tilde s}

\newcommand{\tu}{\tilde u}
\newcommand{\Us}{U^\sharp}

\newcommand{\ty}{\tilde y}

\newcommand{\WFh}{\operatorname{WF}_{ h }}

\newcommand{\eps}{\epsilon}

\newcommand{\CC}{{\mathbb C}}

\newcommand{\NN}{{\mathbb N}}

\newcommand{\RR}{{\mathbb R}}
\newcommand{\IR}{{\mathbb R}}

\def\t2{{\mathbb T}^2}

\newcommand{\CI}{{\mathcal C}^\infty }
\newcommand{\CIc}{{\mathcal C}^\infty_{\rm{c}} }

\newcommand{\cE}{{\mathcal E}}
\newcommand{\cF}{{\mathcal F}}

\newcommand{\cA}{{\mathcal A}}
\newcommand{\cB}{{\mathcal B}}
\newcommand{\cC}{{\mathcal C}}

\newcommand{\cO}{{\mathcal O}}
\newcommand{\cP}{{\mathcal P}}

\newcommand{\Oo}{{\mathcal O}} 
\newcommand{\cS}{{\mathcal S}}

\newcommand{\cU}{{\mathcal U}}
\newcommand{\cV}{{\mathcal V}}

\newcommand{\vareps}{\varepsilon}

\newcommand{\supp}{\operatorname{supp}}

\newcommand{\loc}{\operatorname{loc}}

\newcommand{\Res}{\operatorname{Res}}
\newcommand{\Spec}{\operatorname{Spec}}

\newcommand{\rest}{|}

\renewcommand{\Re}{\mathop{\rm Re}\nolimits}
\renewcommand{\Im}{\mathop{\rm Im}\nolimits}

\newcommand{\ad}{\operatorname{ad}}

\newcommand{\Span}{\operatorname{span}}
\newcommand{\Op}{\operatorname{Op}}
\newcommand{\WF}{\operatorname{WF}}
\newcommand{\ra}{\rangle}
\newcommand{\la}{\langle}

\newcommand{\be}{\begin{equation}}
\newcommand{\ee}{\end{equation}}
\newcommand{\defi}{\stackrel{\rm{def}}{=}}
\newcommand{\defeq}{\stackrel{\rm{def}}{=}}
\def\hto0{\xrightarrow{h\to 0}}

\theoremstyle{plain}
\newtheorem{thm}{Theorem}
\newtheorem{prop}{Proposition}[section]

\newtheorem{lem}{Lemma}[section]

\theoremstyle{definition}

\newtheorem{rem}{Remark}[section]

\numberwithin{equation}{section}

\newcommand{\bequ}{\begin{equation}}

\newcommand{\norm}[1]{\Vert#1\Vert}
\newcommand{\set}[1]{\left\{\,#1\,\right\}}

\def\squarebox#1{\hbox to #1{\hfill\vbox to #1{\vfill}}} 
\newcommand{\stopthm}{\hfill\hfill\vbox{\hrule\hbox{\vrule\squarebox 
                 {.667em}\vrule}\hrule}\smallskip} 

\usepackage{epic, eepic}

\usepackage{amsxtra}

\ifx\pdfoutput\undefined
  \DeclareGraphicsExtensions{.pstex, .eps}
\else
  \ifx\pdfoutput\relax
    \DeclareGraphicsExtensions{.pstex, .eps}
  \else
    \ifnum\pdfoutput>0
      \DeclareGraphicsExtensions{.pdf}
    \else
      \DeclareGraphicsExtensions{.pstex, .eps}
    \fi
  \fi
\fi

\title
[Quantum decay rates in chaotic scattering]
{Quantum decay rates in chaotic scattering}
\author[S. Nonnenmacher]
{St\'ephane Nonnenmacher}
\author[M. Zworski]
{Maciej Zworski}
\address{Service de Physique Th\'eorique, 
CEA/DSM/PhT, Unit\'e de recherche associ\'e CNRS,
CEA/Saclay,\\
91191 Gif-sur-Yvette, France}
\email{nonnen@spht.saclay.cea.fr}
\address{Mathematics Department, University of California \\
Evans Hall, Berkeley, CA 94720, USA}
\email{zworski@math.berkeley.edu}

\begin{document}

\maketitle


\section{Statement of Results}
\label{in}

In this article we prove that for a large class of operators, including
Schr\"odinger operators, 
\begin{equation}
\label{eq:Ph}
P (h) =  - h^2 \Delta + V ( x ) \,, \ \ V \in \CIc ( X ) \,,  \ \ X = \RR^2 \,, 
\end{equation}
with hyperbolic classical flows, the smallness
of dimension of the trapped set implies that 
there is a gap between the resonances
and the real axis.
In other words, the quantum decay rates are bounded
from below if the classical repeller is sufficiently {\em filamentary}.
The higher dimensional statement is given in terms of the 
{\em topological pressure} and is presented in Theorem \ref{t:1g}.
Under the same assumptions, we also prove a useful
resolvent estimate:
\begin{equation}
\label{eq:Rh}
\| \chi ( P ( h ) - E )^{-1} \chi \|_{ L^2  \rightarrow L^2  } \leq 
C\, \frac {\log ( 1/h)} {h }
 \,, 
\end{equation}
for any compactly supported bounded function $ \chi $ -
see Theorem \ref{th:4}, and a remark following it for an 
example of applications.

We refer to \S \ref{ass} for the general assumptions on $ P( h) $,
keeping in mind that they apply to $ P ( h) $ of the form \eqref{eq:Ph}.
The resonances of $ P ( h ) $ are defined as poles of the meromorphic
continuation of the resolvent:
\[ 
R ( z, h ) \defeq ( P ( h ) - z )^{-1} \; : \; L^2 ( X ) 
\longrightarrow L^2 ( X ) \,, \ \ \Im z > 0 \,, 
\]
through the continuous spectrum $ [ 0 , \infty ) $. More precisely,
\[ 
R ( z , h ) \; : \; L^2_{\rm{comp}} ( X ) \longrightarrow 
L^2_{\rm{loc}} ( X ) \,,  \ \ z \in \CC \setminus (-\infty, 0] \,,
\]
is a meromorphic family of operators (here  $ L^2_{\rm{comp}} $ 
and $ L^2_{\rm{loc}}$ denote functions which are compactly supported and in 
$ L^2 $, and functions which are locally in $ L^2 $). 
The poles are called {\em resonances}
and their set is denoted by $ \Res ( P ( h ) )$ --- see \cite{BiZ,ZwN}
for introduction and references. Resonances are counted according to their
multiplicities (which is generically one \cite{KlZw}).

In the case of \eqref{eq:Ph} the 
classical flow is given by Newton's equations:
\begin{gather}
\label{eq:Newt}
\begin{gathered}
\Phi^t ( x, \xi ) \defeq ( x( t) , \xi ( t ) ) \,, \\  
x'( t ) =  \xi ( t) \,, \ \ \xi'( t ) = - dV ( x( t ) ) \,, \ \
x ( 0 ) = x \,, \ \ \xi( 0 ) = \xi \,.
\end{gathered}
\end{gather}
This flow preserves the classical Hamiltonian 
\[  p ( x , \xi ) \defeq |\xi|^2 + V ( x ) \,, \ \ 
( x , \xi ) \in T^* X \,, \ \ X = \RR^2 \,, \]
and the energy layers of $p$ are denoted as follows:
\be\label{e:shell}
\cE_E\defeq \set{\rho\in T^*X,\ p(\rho)=E},\qquad
\cE_E^\delta\defeq \bigcup_{|E'-E|<\delta}\cE_{E'}\,, \quad 
 \delta>0 \,. 
\ee
The incoming and outgoing sets at energy $E$ are defined as 
\be\label{eq:gaha}
 \Gamma^\pm_ E  \defi \set{  \rho  \in T^*X \; :
\; p(\rho) = E \,, \ \ \Phi^t ( \rho )
\not \rightarrow \infty \,, \ t \rightarrow \mp \infty}\subset \cE_E \,.
\ee
The trapped set at energy $E$, 
\begin{equation}
\label{eq:trapp}
 K_E \defi \Gamma^+_E \cap \Gamma^-_ E  
\end{equation}
is a compact, locally maximal invariant set, contained inside $T^*_{B(0,R_0)} 
X $, for 
some $ R_0 $. That is clear for \eqref{eq:Ph}
but also follows from the general assumptions of \S \ref{ass}.

\begin{quote} 
\begin{center}
We assume that the flow $\Phi^t$ 
is {\em hyperbolic} on $ K_E$.
\end{center}
\end{quote}

The definition of hyperbolicity is recalled in \eqref{eq:aa} -- see
\S \ref{ass} below. We recall that it is a structurally stable property, so that
the flow is then also hyperbolic 
on $ K_{E'} $, for $ E'$ near $ E$.
Classes of  potentials satisfying this assumption at a range of
non-zero energies are given in \cite{Mo}, \cite[Appendix c]{SjDuke},
\cite{Wojt}, see also Fig.\ref{Fblock}.
\begin{figure}[htbp]
\begin{center}
\includegraphics[angle=-90,width=5in]{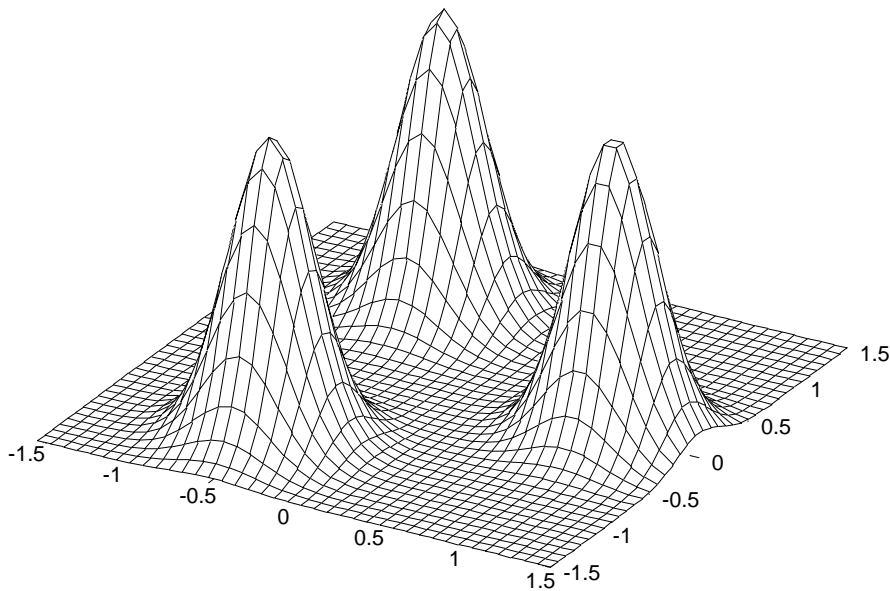}
\end{center}
\caption
{\label{Fblock} A three bump potential exhibiting a hyperbolic trapped set
for a range of energies. When the curve $\{ V= E \}$ is made of three
approximate circles of radii $ a $ and 
centers at equilateral distance $ R $, the partial dimension
$ d_H $ in \eqref{eq:2p} is approximately $\log 2 / \log (R/a)  $ when $ R \gg a $.}
\end{figure}
The dimension of the trapped set appears in 
the fractal upper bounds on the number of resonances. 
We recall the following result \cite{SjZw04} (see \cite{SjDuke} for the
first result of this type):
\begin{thm}\label{t:0}
Let $ P ( h ) $ be given by 
\eqref{eq:Ph} and suppose that the flow $ \Phi^t $ is hyperbolic on $ K_E $. 
Then in the semiclassical limit
\be\label{eq:2}
  | \Res (P ( h) ) \cap D ( E , C h ) | = \Oo ( h^{-d_H } )\,, 
\ee
where 
\be\label{eq:2p}  
2 d_H + 1 = \text{ Hausdorff dimension of $ K_E $.}
\ee
\end{thm}
We note that using \cite[Theorem 4.1]{PeSa}, and in dimension 
$ n = 2 $, we strengthened the formulation of the result in 
\cite{SjZw04} by replacing upper Minkowski (or box) dimension by 
the Hausdorff dimension. We refer to \cite[Theorem 3]{SjZw04} for the
slightly more cumbersome general case. 

In this article we address a different question which has been
present in the physics literature at least since the seminal
paper by Gaspard and Rice \cite{GaRi}. In the same setting of 
scattering by several convex obstacles,  it has also been considered
around the same time by Ikawa \cite{Ik} (see also the careful analysis by Burq \cite{NBu} and a recent paper by Petkov and Stoyanov \cite{Pest}).

\begin{quote}
{\bf Question:} What properties of the flow $ \Phi_t $, or 
of $ K_E $ alone, imply the existence of a {\em gap} $\gamma > 0 $ 
such that, for $h>0$ sufficiently small,
\[  
z \in \Spec ( P ( h)) \,, \quad  \Re z \sim E \; \Longrightarrow 
\; \Im z < - \gamma h \,? 
\]
In other words, what dynamical conditions guarantee a lower
bound on the quantum decay rate?
\end{quote} 

Numerical investigations in different settings of semiclassical three bump
potentials \cite{L,LZ}, three disk scattering
\cite{GaRi,LSZ,Wir}, Cantor-like Julia sets for 
$ z \mapsto z^2 + c $, $ c < - 2$ \cite{StZw}, and quantum maps 
\cite{NZ,SchTw},
all indicate that a trapped set $ K_E $ of low dimension (a ``filamentary'' fractal set) 
guarantees the existence of a {\em resonance gap} $ \gamma >0 $.
\begin{figure}[ht]
\begin{center}
\includegraphics[width=9cm]{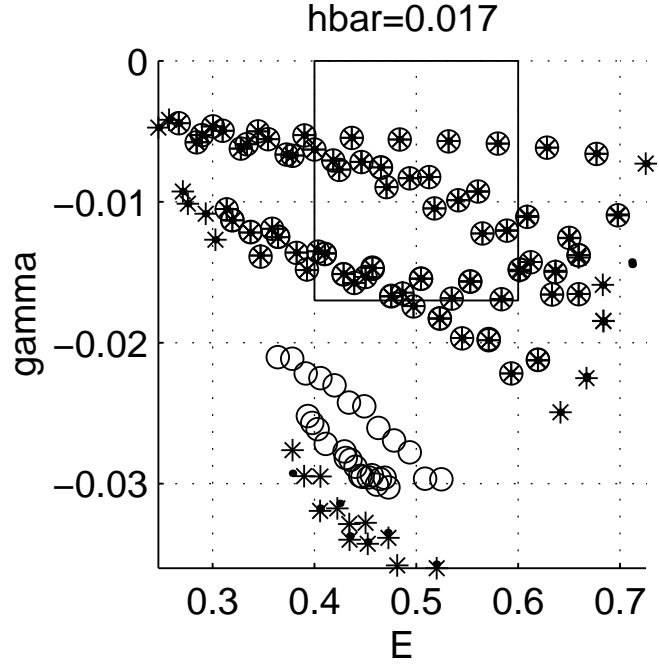}
\end{center}
\caption{A sample of numerical results of \cite{L}: the plot shows
resonances for the potential of Fig.~\ref{Fblock} ($ h = 0.017 $).
For the energies inside the box, the fractal dimension is approximately
$ d_H \simeq 0.288 < 0.5 $ 
(see \cite[Table 2]{L}), and resonances are separated from the 
real axis in agreement with Theorem \ref{t:1}.
}
\label{f:3}
\end{figure}
Some of these works also confirm the fractal Weyl law of
Theorem \ref{t:0}.
which, unlike Theorem \ref{t:1} below,
was first conjectured in the mathematical works on counting resonances.

Here we provide the following 
\begin{thm}
\label{t:1}
Suppose that the assumptions of Theorem \ref{t:0} hold and 
that the dimension $ d_H $ defined in \eqref{eq:2p} satisfies
\begin{equation}
\label{eq:t11}
d_H < \frac12 \,.
\end{equation}
Then there exists $ \delta,\ \gamma > 0 $, and $ h_{\delta,\gamma} > 0 $ such that
\be\label{eq:t12}
0 < h < h_{\delta,\gamma}\ \Longrightarrow \ 
\Res ( P ( h ) ) \cap \left( [E-\delta, E+\delta] - i [0, h \gamma] \right)
= \emptyset \,. 
\ee
\end{thm}

The statement of the theorem can be made more general and more 
precise using a more
sophisticated dynamical object, namely the {\em topological 
pressure} of the flow on $K_E$ with respect to the unstable Jacobian:
$$   
\cP_{E} ( s ) = \text{ pressure of the flow
on $ K_{E} $ with respect to the unstable Jacobian.} 
$$
We will give two equivalent definitions of the pressure below,
the simplest to formulate (but not to use), given in \eqref{eq:defpr}.

The main result of this paper is 
\begin{thm}\label{t:1g}
Suppose that $ P ( h ) $ satisfies the general assumptions of
\S \ref{ass} (in particular it can be of the form \eqref{eq:Ph} with 
$ X = \RR^n $), and that the flow $ \Phi^t $ is hyperbolic on the
trapped set $ K_{E} $. Suppose that the topological pressure satisfies
$$  
\cP_E ( 1/2 ) < 0 \,. 
$$
Then there exists $ \delta > 0 $ such that for any 
$ \gamma $ satisfying
\be\label{eq:press}
0 < \gamma < \min_{|E-E'|\leq\delta}(- \cP_{E'}( 1/2)) \,,  
\ee
there exits $ h_{\delta,\gamma}>0  $ such that 
\be\label{eq:t12n}
0 < h < h_{\delta,\gamma}\ \Longrightarrow \ 
\Res ( P ( h ) ) \cap \left( [E-\delta, E+\delta] - i [0, h \gamma] \right)
= \emptyset \,. 
\ee
\end{thm}
For $n=2$,   $d_H<1/2 $ is equivalent to $ \cP_E(1/2)<0$, 
which shows that Theorem \ref{t:1} follows from Theorem \ref{t:1g}.
The connection between ${\rm sgn}\cP_E(1/2)$ and a resonance gap
also holds in dimension $n\geq 3$; however, for $n\geq 3$ there is generally
no simple link between the sign of $\cP_E(1/2)$ and
the value of $d_H$ (except when the flow is ``conformal'' in the 
unstable, respectively stable directions \cite{PeSa}).

The optimality of Theorem \ref{t:1g} is not clear. Except in 
some very special cases (for instance when $ K_E $ consists of one
hyperbolic orbit) we do not expect the estimate on the width of the resonance
free region in terms of the pressure to be optimal. In fact, in the 
analogous case of scattering on convex co-compact hyperbolic 
surfaces the results of Naud (see \cite{Naud} and references given there)
show that the resonance free strip is wider at high energies than the strip
predicted by the pressure. That relies on delicate zeta function analysis
following the work of Dolgopyat: at zero energy there exists a 
Patterson-Sullivan resonance with the imaginary part (width) given by the pressure, 
but all other resonances have more negative imaginary parts. A similar
phenomenon occurs in the case of
Euclidean obstacle scattering
as has recently been shown by Petkov and Stoyanov \cite{Pest}.

\begin{figure}[ht]
\begin{center}
\includegraphics[width=16cm]{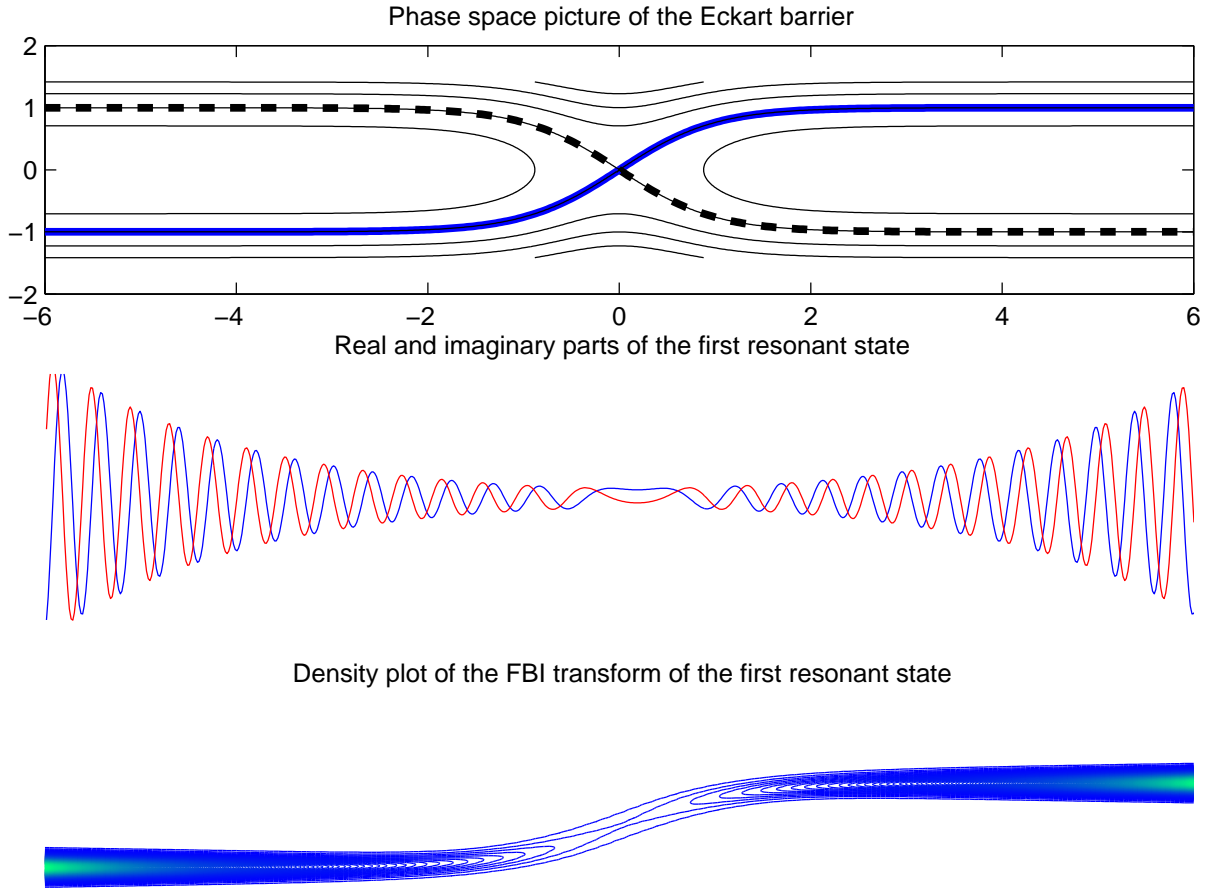}
\end{center}
\caption{The top figure shows the phase portrait for the Hamiltonian $ p ( x , \xi ) =
\xi^2 + \cosh^{-2} ( x ) $, with $ \Gamma_1^\pm $ highlighted. The middle
plot shows the resonant state corresponding to the resonance
closest to the real axis at $ h = 1/16 $, and the bottom plot shows
 the squared modulus of its FBI tranform.
The resonance states were computed by D.~Bindel
({\tt http://cims.nyu.edu/$\sim$dbindel/resonant1d}) and the FBI transform
was provided by L.~Demanet. The result of Theorem~\ref{t:2} is
visible in the mass of the FBI transform concentrated on $ \Gamma_1^+ $,
with the exponential growth in the outgoing direction.}
\label{fig1}
\end{figure}

The proof of Theorem \ref{t:1g} is based on the ideas developed in  the 
recent work of Anantharaman and the first author \cite{An,AnNo}
on semiclassical defect measures for eigenfuctions of the Laplacian on
manifolds with Anosov geodesic flows. Although we do not use semiclassical
defect measures in the proof of Theorem \ref{t:1g}, the following
result provides a connection:
\begin{thm}
\label{t:2}
Let $ P ( h) $ satisfy the general assumptions of \S \ref{ass}
(no hyperbolicity assumption here). 
Consider a sequence of values $h_k\to 0$ and a corresponding 
sequence of resonant states (see \eqref{eq:gath} in \S \ref{ass} 
below) satisfying
\begin{equation}
\label{eq:th2}  \| u ( h_k ) \|_{L^2 ( \pi ( K_E )  + B ( 0, \delta)) } = 1 \,,  \ \
 \Re z ( h_k ) = E+ o ( 1 ) \,, \ \ \Im z ( h ) \geq - C h
\end{equation}
where $ K_E $ is the trapped set at energy $ E $ \eqref{eq:trapp}
and $ \delta > 0$.
Suppose that a  semiclassical defect measure $ d \mu $ on
$ T^* X $ is
associated with the sequence $ (u ( h_k)) $:
\begin{gather}
\label{eq:t20}
\begin{gathered}
  \langle a^w ( x, h_k D) \chi u ( h_k ) , \chi u ( h_k)
\rangle \longrightarrow
\int_{ T^* X }  a ( \rho ) \, d\mu ( \rho ) \,, \ \ k \rightarrow \infty \,,
\\
a\in \CIc(T^*X)\,,\ \ \chi \in \CIc ( X ) \,, \ \  \pi^* \chi\rest_{ \supp a } \equiv 1 \,, \ \
\pi : T^* X \rightarrow X \,.
\end{gathered}
\end{gather}
Then
\be\label{eq:t21}
\supp \mu  \subset \Gamma^+ _E
\ee
and there exists $ \lambda > 0 $ such that
\be\label{eq:t22}
\lim_{ k \rightarrow \infty }  \Im z( h_k ) / h_k  = - \lambda/2 \,, \quad
 \text{ and} \quad
{\mathcal L}_{H_{p}} \mu =  \lambda \mu\,.
\ee
\end{thm}
See Fig.~\ref{fig1} for a numerical result illustrating the theorem. A similar
analysis of the phase space distribution for the resonant eigenstates of quantized open
chaotic maps (discrete-time models for scattering Hamiltonian flows)
has been recently performed in \cite{KNPS06,NoRu07}.
Connecting this theorem with Theorems \ref{t:1} and \ref{t:1g}, we see
that the semiclassical defect measures associated with sequences 
of resonant states have decay rates $\lambda$ bounded from below by $2\gamma > 0$,
once the dimension of the trapped set is small enough ($n=2$), or more
generally, the pressure at $\frac 12 $ is negative.

Our last result is the precise version of the resolvent
estimate \eqref{eq:Rh}:
\begin{thm}
\label{th:4} 
Suppose that $ P ( h ) $ satisfies the general assumptions of
\S \ref{ass} (in particular it can be of the form \eqref{eq:Ph} with 
$ X = \RR^n $), and that the flow $ \Phi^t $ is hyperbolic on
the trapped set $ K_{E} $. If  the pressure $  \cP_E ( 1/2 ) < 0 $ then for any 
$ \chi \in \CIc ( X ) $ 
we have 
\be\label{eq:t4}
\| \chi ( P ( h ) - E )^{-1} \chi \|_{ L^2 (  X) \rightarrow L^2 ( X) }
\leq  
C\, \frac{ \log ( 1/h ) } {h} \,, \quad 0 < h < h_0 \,.
\ee
\end{thm}
Notice that the upper bound $C\, \log ( 1/h )/h $ is the same as in the one
obtained in the case of one hyperbolic orbit by Christianson \cite{Chr}.
To see how results of this type imply dynamical estimates see
\cite{BZ,Chr}. In the context of Theorem \ref{th:4},
the applications are presented in \cite{Chr07}. Referring to that 
paper for details and pointers to the literatures we present one
application. 

Let $ P = -h^2 \Delta_g $ be the Laplace-Bertrami 
operator satisfying the assumptions below, for instance on 
a manifold Euclidean outside of a compact set with the standard 
metric there. 
The Schr\"odinger propagator, $ \exp ( - it \Delta_g ) $,
is unitary on any Sobolev space so regularity
is not improved in propagation. Remarkably, when $ K = \emptyset $,
that is, when the metric is nontrapping, 
the regularity improves when we integrate in
time and cut-off in space:
\[ \int_0^T \| \chi \exp ( - i t \Delta_g ) u \|_{H^{1/2} ( X) }^2 dt \leq
C \| u \|_{L^2 ( X ) }^2 \,, \ \ \ \chi \in C^\infty_{\rm{c}} ( X ) \,, \]
and this much exploited effect is known as {\em local smoothing}.
As was shown by Doi \cite{Doi} 
any trapping (for instance a presence
of closed geodesics or more generally $ K \neq \emptyset $) 
will destroy local smpoothing.
Theorem \ref{th:4} implies that under the assumptions that 
the geodesic flow is hyperbolic on the trapped set $ K \subset S^* X$, and
that the pressure is negative at $ 1/2 $ (or, when $ \dim X = 2 $, 
that the dimension 
of $ K \subset S^* X $ is less than $ 2 $)
local smoothing holds with $ H^{1/2} $ replaced by
$ H^{1/2-\epsilon} $ for any $ \epsilon > 0 $.

\medskip

\noindent
{\bf Notation.} In the paper $ C $ denotes a constant the value
of which may changes from line to line. The constants which matter and
have to be balanced against each other will always have a subscript
$ C_1, C_2 $ and alike. The notation $ u = {\mathcal O}_V ( f ) $
means that $ \| u \|_{V } = {\mathcal O} ( f ) $, and the notation
$ T = {\mathcal O}_{V \rightarrow W } ( f ) $ means that 
$ \| T u \|_W = {\mathcal O} (f ) \| u \|_V $.

\section{Outline of the proof}
\label{oops}
It this section we present the main ideas whith the precise definitions
and references to previous works given in the main body of the paper.
The operator to keep in mind is $ P = P(h) = - h^2 \Delta_g + V $, where 
$ V \in \CIc ( X ) $, $ X = \RR^n $, and the metric $ g $ is 
Euclidean outside a compact set.
The corresponding classical Hamiltonian is given by $ p = \xi^2 + V ( x) $. 
Weaker assumptions, which in particular 
do not force the compact support of the 
perturbation, are described in \S \ref{ass}.

First we outline the proof of Theorem \ref{t:1g} in the
simplified case in which resonances are replaced by the {\em eigenvalues}
of an operator modified by a {\em complex absorbing potential}:
$$ 
P_W = P_W ( h ) \defeq P - i W \,,
$$
where $ W \in \CI ( X ; [ 0 , 1] ) $, 
satisfies the following conditions:
$$  
W \geq 0 \,, \qquad  \supp W \subset X \setminus B ( 0 , R_1 ) \,, \qquad 
W\rest_{ X \setminus B ( 0 , R_1 + r_1 )  } = 1 \,,
$$
for $ R_1 , r_1$ sufficiently large. In particular, $ R_1 $ is 
large enough so that $ \pi ( K_E)  \subset B( 0 , R_1 ) $, where
$ K_E $ is the trapped set given by \eqref{eq:trapp}. The 
non-self-adjoint operator $ P_W $ has a discrete spectrum in 
$ \Im z > - 1/ C $ and the analogue of Theorem \ref{t:1g}
reads:

\medskip
\noindent
{\bf Theorem \ref{t:1g}$'$.}{\em Under the assumptions of 
Theorem \ref{t:1g}, for
\be
\label{eq:pressp}
0 < \gamma < \min_{|E-E'|\leq\delta}(- \cP_{E'}( 1/2)) \,,  
\ee
there exits $ h_0 = h_0 ( \gamma, \delta ) $ such that 
for $ 0 < h < h_0 $,
\be
\label{eq:t12p}
\Spec ( P_W ( h ) ) \cap \left( [E-\delta, E+\delta] - i [0, h \gamma] \right)
= \emptyset \,.
\ee}

\medskip
This means that the spectrum of $ P_W ( h ) $ near $ E $ is 
separated from the real axis by $ h \gamma $, where $ \gamma $ is
given in terms of the pressure of the square root of the 
unstable Jacobian, $ \cP_E ( 1/2 ) $.

This spectral gap is equivalent to the fact that the {\em decay rate} of 
any eigenstate is bounded from below:
\[  
P_W\, u = z\, u \,, \ \ z \in D ( E, 1/C ) \,, \ \ u \in L^2 
\ \Longrightarrow \ \| \exp ( - i t P_W /h ) u \| \leq e^{- \gamma t }\,\|u\| \,. 
\]
This is the physical meaning of the gap between the spectrum 
(or resonances) and the real axis --- a lower bound for the 
quantum decay rate --- and the departing point for the proof.
To show \eqref{eq:t12p} we will show that for functions, $ u $, which are
microlocally concentrated near the energy layer $ \cE_E=p^{-1} ( E ) $
(that is, $ u = \chi^w ( x, h D ) u + {\mathcal O} ( h^\infty ) $ 
for a $\chi$ supported near $ \cE_E $) we have 
\begin{gather}
\label{eq:PW}
\begin{gathered} 
 \| e^{ - i t P_W / h } u \| \leq C\, h^{-n/2}\, e^{- \lambda t }\, \| u \| \,,
 \qquad 0 < \lambda < \min_{|E-E'|\leq\delta}(- \cP_{E'}( 1/2)) \,, \\ 
0 \leq t \leq \tM \log ( 1/h ) \,, 
\end{gathered}
\end{gather}
for any $ \tM $. Taking $ \tM\gg n/2\lambda$ and applying the estimate
to an eigenstate $ u $ gives \eqref{eq:t12p}.

To prove \eqref{eq:PW} we decompose the propagator using an open cover $(W_a)_{a\in A}$
of the neighbourhood $\cE_E^\delta$ of the energy surface. That cover 
is adapted to the definition of the pressure (see \S\S 
\ref{s:selection},\ref{s:covers}) and it leads to a microlocal
partition of a neighbourhood of the energy surface:
$$ 
\sum_{ a\in A } \Pi_a = \chi^w ( x, h D ) + \cO ( h^\infty )
\,, \quad
\chi \equiv 1 \ \text{ on $\cE_E^{\delta/8}$},\quad  \text{ess-supp}\;\Pi_a \Subset W_{a}\,.
$$
The definition of the pressure in \S\ref{s:selection} also involves 
a time $ t_0>1$, independent of $h$, but depending on the classical cover.
Taking 
\be\label{e:M-N}
N \leq M \log ( 1/h )\,\qquad N\in\NN,\qquad  M>0\ \text{fixed but arbitrary large},
\ee
the propagator at time $t=N\,t_0$ 
acting of functions $u$ microlocalized inside $\cE_E^{\delta/8}$ can be written as
\be\label{eq:bigs}  
e^{- i N t_0 P_W/ h}\,u =  
\sum_{ \alpha  \in A^N } U_{\alpha_N } \circ \cdots \circ U_{
\alpha_1 }\,u + \cO (h^\infty)\,\|u\| \,, \qquad  
U_a \defeq e^{- i t_0 P_W/ h} 
\, \Pi_a \,.
\ee
Most terms in the sum appearing on the right hand side of \eqref{eq:bigs}
are negligible. The sequences $ \alpha=( \alpha_1 , \cdots , \alpha_N ) $
which are classically forbidden, that is, for which the 
corresponding sequences of 
neighbourhoods are {\em not} successively connected by classical 
propagation in time $ t_0 $, lead to negligible terms. So do the sequences
for which the propagation crosses the region where $ W = 1 $:
the operator $ \exp ( - it_0 P_W / h ) $ is negligible there, 
due to damping (or ``absorption'') by $ W$.

As a result, the only terms relevant in the sum on the right
hand side of \eqref{eq:bigs} come from $ \alpha \in A_1^N \cap 
\cA_N$ where $ A_1 $ indexes the element of the partition 
intersecting the trapped set $ K_E $, and $ \cA_N $ are
the classically allowed sequences --- see \eqref{e:A_N}. We then 
need the crucial {\em hyperbolic dispersion estimate} proved in 
\S\ref{s:estimate} after much preliminary work in \S\S 
\ref{s:iterating} and \ref{s:elf}: for $ N \leq M \log (1/h)$, $M>0$ arbitrary, we have
for any sequence $\alpha\in A_1^N \cap \cA_N$:
\be\label{eq:crucial}
\| U_{\alpha_N } \circ \cdots \circ U_{\alpha_1} \| \leq 
 h^{-n/2} ( 1 + \eps_0 )^N 
\prod_{j=1}^N \left( 
\inf_{\rho\in W_{\alpha_j} \cap
K^\delta_E }  \det\big(d\Phi^{t_0} ( \rho ) \rest_{ E^{+0}_ \rho  } \big) 
\right)^{-\frac12} 
\,.
\ee
The expression in parenthesis is the coarse-grained 
unstable Jacobian defined in \eqref{e:coarse-jac}, and
$\eps_0>0$ is a parameter depending on the cover $(W_a)$, 
which can be taken arbitrarily small
--- see \eqref{eq:press2}. From the definition of the pressure in \S\ref{s:selection},
summing \eqref{eq:crucial} over $ \alpha \in A_1^N \cap {\mathcal A}_N $ 
leads to \eqref{eq:PW}, with $\tM=Mt_0$.

In \S \ref{re} we show how to use \eqref{eq:PW} to obtain a
resolvent estimate for $ P_W $: at an energy $E$ for which the flow
is hyperbolic on $ K_E $ and $ \cP_E ( 1/2 ) < 0 $, we have
\begin{equation}
\label{eq:PWE}
\| ( P_W - E )^{-1} \|_{ L^2 ( X ) \rightarrow L^2 ( X ) } 
\leq C\, \frac{\log ( 1/h ) } h \,,\qquad 0 < h < h_0\,.
\end{equation}
To prove Theorem \ref{t:1g}, that is the gap between 
{\em resonances} and the real axis, we use the complex scaled
operator $ P_\theta $~: its eigenvalues near the real axis are resonances of $P$. 
If $ V $ is a decaying real analytic potential extending to a conic
neighbourhood of $ \RR^n $ (for instance a sum of three 
Gaussian bumps showed in Fig.~\ref{Fblock}), then we can 
take $ P_\theta = -h^2 e^{-2i\theta} \Delta + V ( e^{i \theta} x ) $,
though in this paper we will always use exterior complex scaling
reviewed in \S\ref{defcs}, with $ \theta \simeq M_1 \log(1/h)/h $,
where $ M_1 $ is chosen depending on $ M $ in \eqref{e:M-N}.

To use the same strategy of estimating $ \exp ( - i t P_\theta/ h ) $
we need to further modify the operator by conjugation with 
microlocal exponential weights. That procedure is described in \S\ref{qd}.
The methods developed there are also used in the proof of 
Theorem \ref{t:2} and in showing how the estimate \eqref{eq:PWE}
implies Theorem \ref{th:4}.

Since we concentrate on the more complicated, and scientifically 
relevant, case of resonances, the additional needed facts about the
study of $ P_W $ and its propagator are presented in the Appendix.

\section{Preliminaries and Assumptions}
\label{pr}
In this section we recall basic concepts of semiclassical 
analysis, state the general assumptions on operators
to which the theorems above apply, define hyperbolicity and 
topological pressure. We also define resonances using 
{\em complex scaling} which is the standard tool in the
study of their distribution. Finally, we will review
some results about semiclassical Lagrangian 
states and Fourier integral operators.

\subsection{Semiclassical analysis}\label{s:semiclass}
Let $ X $ be a $\CI$ manifold which agrees with $ \RR^n $ outside a
compact set, or more generally
\[  X = X_0 \sqcup (\RR^n \setminus B( 0 , R_0 ) ) \sqcup \cdots \sqcup
(\RR^n \setminus B( 0 , R_0 ) )  \,, \ \ X_0 \Subset X \,. \]
A weight function on $ T^* X $ is of the form 
\[ m : T^*X \longrightarrow ( 0 , \infty ) \,, \ \
m ( x , \xi ) \leq C ( 1 + d ( x , y ) + | \xi - \eta | )^N m ( y , \eta )
\,, \]
where $ d ( x , y ) $ is a distance function on $ X $, and any 
uniform choice of distance in the fibers is allowed --- 
the usual Euclidean distance is taken outside $ X_0 $. The typical 
choice is $ m ( x, \xi ) = 1 + |\xi|^2_g $, for a metric $ g $.

The class of 
symbols associated to the weight $ m$ is defined as 
$$
S^k_\delta ( T^* X , m ) = \set{ a \in \CI( T^* X \times (0, 1]  ) :
|\partial_x ^{ \alpha } \partial _\xi^\beta a ( x, \xi ;h ) | \leq
C_\alpha m ( x , \xi ) h^{-k-\delta ( | \alpha| + |\beta |). } } \,.
$$
Most of the time we will use the class with $ \delta = 0 $
in which case that we drop the subscript. When $ m \equiv 1 $ and $ k =0$,
we simply write $ S ( T^*X ) $ or $ S $ for the class of symbols. 

We denote by
$ \Psi_{h,\delta}^{k} ( X, m ) $  or $\Psi_{h}^{k} ( X, m )$
the corresponding  class of pseudodifferential operators.
We have surjective quantization and symbol maps:
\[ 
 \Op \; : \; S^{ k } ( T^* X, m  ) \ \longrightarrow 
  \Psi^{k}_h ( X, m ) \,, \quad
   \sigma_h \; : \; \Psi_h^{k} ( X, m  ) \ \longrightarrow 
  S^{ k } ( T^* X , m ) / S^{ k-1} ( T^* X , m ) \,.
\]
Multiplication of symbols corresponds to composition of 
operators, to leading order:
\[
 \sigma_h ( A \circ B ) = \sigma_h ( A)\sigma _h ( B ) \,, \]
and
\[ \sigma_h \circ \Op : S^{k} ( T^* X ,m ) 
\ \longrightarrow  S^{ k } ( T^* X, m  ) / S^{ k-1} ( T^* X, m  ) \,,
\]
is the natural projection map. A finer filtration can be obtained
by combining semiclassical calculus with the standard calculus
(or in the yet more general framework of the Weyl calculus) --- see
for instance \cite[\S 3]{SjZw02}.

The class of operators and the quantization map are defined locally using the
definition on $ \RR^n $:
\be\label{eq:weyl}
  \Op (a) u ( x') = a^w(x,hD) u(x') = \frac1{ ( 2 \pi h )^n } 
  \int \int  a \big( \frac{x' + x }{2}  , \xi \big) 
e^{ i \la x' -  x, \xi \ra / h } u ( x ) dx d \xi \,, 
\ee
and we refer to \cite[Chapter 7]{DiSj} for a detailed discussion, and 
to \cite[Appendix D.2]{EZ} for the semiclassical calculus on manifolds.

The semiclassical Sobolev spaces, $ H_h^s ( X ) $  are defined by 
choosing a globally elliptic, self-adjoint operator, 
$ A \in \Psi_h ( X, \langle \xi \rangle ) $ (that is an operator satisfying 
$ \sigma ( A) \geq \langle \xi \rangle / C $ everywhere) and
putting 
$$  
\| u \|_{ H_h^s } = \| A^{s} u \|_{L^2 ( X ) } \,.
$$
When $  X = \RR^n $, 
$$
  \| u \|^2_{ H_h^s } \sim  \int_{\RR^n} 
\langle \xi \rangle^{2s} |{\mathcal F}_h u ( \xi ) |^2 d \xi \,, \ \ 
\ {\mathcal F}_h u ( \xi ) \stackrel{\rm{def}}{=} 
\frac{1}{ ( 2 \pi h )^{n/2} }\int_{\RR^n } u ( x ) e^{ - i \langle x , \xi \rangle/h } d x \,.
$$
Unless otherwise stated all norms in this paper, $ \| \bullet \| $, 
are $ L^2 $ norms.

For $ a \in S ( T^* X ) $ we follow \cite{SjZw02} and 
say that the {\em essential support} is equal to a given compact
set $ K \Subset T^*X $, 
$$
  {\text{ess-supp}}_h\; a = K \Subset T^*X\,, 
$$
if and only if   
$$
 \forall \, \chi \in S ( T^*X ) \,, \ \supp \chi \subset 
\complement K \ \Longrightarrow
 \ \chi \, a \in h^\infty {\mathcal S} ( T^* X) \,.
$$
Here $\cS $ denotes the 
Schwartz class which makes sense since $ X $ is Euclidean outside a compact.
In this article we are only concerned with a purely semiclassical 
theory and deal only with {\em compact} subsets of $ T^* X $. 

For $ A \in \Psi_h ( X) $,  $  A = \Op ( a ) $, we put
$$
\WFh ( A) =  \text{ess-supp}_h\; a \,,
$$
noting that the definition does not depend on the choice of $ \Op $.

We introduce the following condition
\be\label{eq:tempu}  
u \in \CI ( ( 0 , 1]_h ; {\mathcal D}' ( X) ) 
\,, \ \ 
\ \exists \; P\,, h_0 \,, \   \   \| \langle x \rangle^{-P } u \|_{ L^2
 ( X ) } 
\leq h^{-P} \,, \ h < h_0 \,, 
\ee
and call families, $ u = u (h) $, satisfying \eqref{eq:tempu} $ h$-tempered.
What we need is that for $ u ( h ) $, $h$-tempered, $ \chi^w ( x,h D) u ( h ) 
\in h^\infty {\mathcal S} ( X )$ for $ \chi \in h^\infty {\mathcal S}
( T^* X ) $. That is, applying an operator in the residual 
class produces a negligible contribution.

For such $h$-tempered families we define the semiclassical $ L^2$-wave 
front set~:
\be\label{eq:defWF}
\WFh ( u ) =   \complement \big\{  ( x, \xi ) 
\; : \; \exists \, a \in S ( T^* X ) \,, \  \ a ( x, \xi ) \neq 0 
\,, \ \| a^w ( x , h D )  u \|_{L^2} = {\mathcal O}( h^\infty) \big\}  
 \,.
\ee
The last condition in the definition can be equivalently replaced
with 
\[ a^w ( x , h D ) u \in h^\infty  \CI ( ( 0 , 1]_h ; \CI ( X) )  \,,\]
since we may always take $ a \in {\mathcal S} ( T^* X ) $.

Equipped with the notion of semiclassical wave front set, 
it is useful and natural to 
consider the operators and their properties {\em microlocally}. For that 
we consider the classe of {\em tempered} operators, 
$ T=T(h)  \; : \; \cS ( X )  \rightarrow \cS'  ( X) $, 
defined by the condition
$$
\exists \, P , h_0 \,,  \qquad
\| \langle x \rangle^{-P} T u \|_{ H_h^{-P} ( X ) }  
\leq h^{-P} 
\| \langle x \rangle^{P} u \|_{ H_h^P ( X ) } \,, \ \ 0 < h < h_0 \,.
$$
For open sets, $ V \subset \overline V \Subset T^* X$, $ U 
\subset \overline U \Subset T^* X $, the operators 
{\em defined microlocally} near  $ V \times U $ 
are given by the following equivalence classes of tempered operators:
\begin{gather}
\label{eq:2.7}
\begin{gathered}
T \sim T'\ \text{if and only if there exist open sets} \\ 
\widetilde U,\,\widetilde V\Subset T^* X,\ \ \overline  U \Subset  \widetilde U,\ \
\overline  V \Subset  \widetilde V,\ \ \text{such that}\\
A ( T - T' ) B = \cO_{{\cS}' \to {\cS }} ( h^\infty )\,, \\
\text{for any }A, B \in \Psi_h ( X ) \ \ \text{with }\quad
\WF_h ( A ) \subset \widetilde V \,, \ \ \WF_h ( B ) \subset \widetilde U\,.\\
\end{gathered}
\end{gather}
For two such operators $T,T'$ we say that $ T = T' $ {\em microlocally} 
near $ V \times U$. 
If we assumed that, say
$ A = a^w ( x, h D ) $, where $ a \in \CIc ( T^* X ) $
then $ \cO_{\cS'\to \cS} ( h^\infty ), $  could be replaced by 
$ \cO_{ L^2 \to L^2 } ( h^\infty ) $ in the condition.
We should stress that ``microlocally'' is always meant in this 
semi-classical sense in our paper.

The operators in $ \Psi_h ( X ) $ are bounded on $ L^2 $ uniformly 
in $ h $. 
For future reference we also recall the sharp G{\aa}rding 
inequality (see for instance \cite[Theorem 7.12]{DiSj}): 
\begin{equation}
\label{eq:sharpg}
  a \in S ( T^* X )\,, \quad a \geq 0 \ \Longrightarrow \ 
\langle a^w ( x , h D ) u , u \rangle \geq - C h \| u \|_{L^2}^2 \,, \quad
u \in L^2 ( X ) \,,
\end{equation}
and Beals's characterization of pseudodifferential operators on $ X $
(see \cite[Chapter 8]{DiSj} and \cite[Lemma 3.5]{SjZw04} for the 
$ S_\delta $ case)~:
\begin{equation}
\label{eq:beals} A \in \Psi_{h,\delta} ( X  ) 
\; \Longleftrightarrow \; \left\{ \begin{array}{l} \| \ad_{W_N} 
\cdots \ad_{ W_1  } A \|_{ L^2 \rightarrow L^2 } = {\mathcal O} 
( h^{(1-\delta) N } ) 
 \\
\ \\
 \ \forall \, W_j \in {\rm{Diff}}^1(X) \,, \quad  j =1, \cdots, N \,,\\
\ \\
\text{$ W_j = \langle a , hD_x \rangle + \langle b , x \rangle $, $ a, b \in \RR^n$,
outside $ X_0 $.}
\end{array} \right. \end{equation}
Here $ \ad_B C = [ B , C ] $.

\subsection{Assumptions on $ P ( h ) $}
\label{ass}

We now state the general assumptions on the operator $ P = P(h) $,
stressing that the simplest case to keep in mind is 
\[ 
P = - h^2 \Delta + V( x )  \,, \ \ V \in \CIc ( \RR^n ) \,.
\]
In general we consider
\[ 
P ( h ) \in \Psi_h ( X, \langle \xi \rangle^2  ) \,, 
\ \ P ( h ) = P ( h)^* \,, 
\]
and an energy level $ E>0$, for which
\begin{gather}
\label{eq:gac}
\begin{gathered}
 P(h) = p^w ( x, hD) + h p_1^w ( x , h D; h ) \,, \ \ 
p_1 \in S ( T^* X, \langle \xi \rangle^2 ) 
\,, \\
 | \xi | \geq C \ \Longrightarrow \ 
p ( x , \xi ) \geq \langle \xi \rangle^2 / C \,, 
\ \ \ p = E \ \Longrightarrow dp \neq 0 \,,
\\
\exists \; R_0,  \ \forall \;  
u \in \CI ( X \setminus B ( 0 , R_0) )\,, \ \  P ( h ) u ( x ) = 
Q ( h ) u ( x ) \,. 
\end{gathered}
\end{gather}
Here the operator near infinity takes the following form on each ``infinite branch'' 
$\RR\setminus B(0,R_0)$ of $X$:
$$
 Q(h)  =\sum_{|\alpha|\leq2} a_{\alpha}(x;h){(hD_x)}^{\alpha}    \,,
$$
with 
$a_\alpha(x;h)=a_\alpha(x)$ independent of $h$ for $|\alpha|=2$,
$a_{\alpha}(x;h)\in C_b^\infty(\RR^n)$ uniformly bounded with respect to
$h$ (here $C_b^\infty(\RR^n)$
denotes the space of $C^\infty$ functions with bounded derivatives 
of all orders), and 
\begin{gather}
\label{eq:valid}
\begin{gathered}
  \sum_{|\alpha|=2} a_\alpha(x) {\xi}^\alpha\geq {(1/c)}{|\xi|}^2, \;\; 
{\forall\xi\in\RR^n}\,, \text{ for some constant $c>0$, } \\
\sum_{|\alpha|\leq2}a_\alpha(x;h){\xi}^\alpha \, \longrightarrow\, 
{\xi}^2 \,,\quad  \text{as $ |x|\to \infty$, uniformly with respect to $h$.}
\end{gathered}
\end{gather}
We also need the following analyticity
assumption in a neighbourhood of infinity:
there exist $\theta_0\in {[0,\pi)},$ $\epsilon>0$ such that the
coefficients $a_\alpha(x;h)$ of $Q(h)$ extend holomorphically in $x$ to
$$\set{ r\omega: \omega\in {{\mathbb C}^n}\,, \ \ 
 \text{dist}(\omega,{\bf{S}}^n)<\epsilon\,, \ 
r\in {{\mathbb C}}\,, \   |r|>R_0\,, \ 
 \text{arg}\,r\in [-\epsilon,\theta_0+\epsilon) }\,, $$
with \eqref{eq:valid} valid also in this larger set of $x$'s.
Here for convenience we chose the same $ R_0 $ as the one 
appearing in \eqref{eq:gac}, but that is clearly irrelevant.

We note that the analyticity assumption in a conic neighbourhood
near infinity automatically 
strengthens \eqref{eq:valid} through an application of Cauchy 
inequalities:
\be
\label{eq:validC}
\partial_x^\beta 
\left( \sum_{|\alpha|\leq2}a_\alpha(x;h){\xi}^\alpha - 
{\xi}^2 \right) \leq  | x|^{-|\beta|}\, f_{|\beta|}(|x|)\, \langle \xi \rangle^2
\,, \ \ x \longrightarrow \infty 
\,,
\ee
where for any $j\in\NN$ the function $f_j(r)\searrow 0$ when $r\to\infty$.

\subsection{Definitions of hyperbolicity and topological pressure}
\label{defhyp}

We use the notation
\[ 
\Phi^t ( \rho ) = \exp ( t H_{p} ) ( \rho ) \,, \ \ 
\rho = ( x , \xi ) \in T^* X \,,
\]
where $ H_{p} $ is the Hamilton vector field of $ p $, 
\[ 
H_p \defeq \sum_{i=1}^n \frac{\partial p }{\partial \xi_i} 
\frac{\partial}{\partial x_i} - \frac{\partial p}{\partial x_i} 
\frac{\partial}{\partial \xi_i}=\{p,\cdot\} \,, 
\]
in local coordinates in $ T^*X $. The last expression is the Poisson bracket
relative to the symplectic form $\omega = \sum_{i=1}^n d\xi_i\wedge dx_i$.

We assume $ p = p ( x , \xi ) $ and $E>0$ satisfy the assumptions 
\eqref{eq:gac} and \eqref{eq:valid} of \S\ref{ass}, 
and study the flow $\Phi^t$ generated by $p$
on $\cE_E$.
The incoming and outgoing sets, $ \Gamma_E^\pm $, and the trapped
set, $ K_E $, are given by \eqref{eq:gaha} and \eqref{eq:trapp}
respectively.

We say that the flow $\Phi^t$ 
is {\em hyperbolic on $ K_E $}, if
for any $\rho\in K_{E} $,
the tangent space to $\cE_{E}$ at $\rho$ splits into flow, unstable and stable 
subspaces \cite[Def.~17.4.1]{KaHa}: 
\be\label{eq:aa}
\begin{split}
i)\ & T_\rho (\cE_{E}) = \RR H_p (\rho ) \oplus E^+ _ \rho \oplus 
E^- _ \rho \,, \quad  \dim E^\pm _ \rho  = 1 \\
ii)\ &  d \Phi^t_\rho ( E^\pm_\rho ) = E^\pm _{ \Phi^t ( \rho ) }\,, 
\quad \forall t\in\IR\\
iii)\ & \exists \; \lambda > 0\,,\quad  \| d\Phi^t_\rho ( v ) \| \leq C e^{-\lambda |t| } \| v \| \,,
\quad \text{for all $ v \in E^\mp_ \rho  $, $ \pm t \geq 0 $.}\\
\end{split}
\ee
$K_E$ is a {\em locally maximal hyperbolic set} for the flow
$\Phi^t\rest_{\cE_E}$. The following properties are then satisfied:
\be
\label{eq:aaa} 
\begin{split}
iv)\ & K_{E}  \ni \rho \longmapsto E^{\pm}_ \rho  \subset
T_\rho (\cE_{E}) \
\text{ is H\"older-continuous} \\
v)\ &\text{any }\rho\in K_E\ 
\text{admits local (un)stable manifolds } W^{\pm}_{\loc}(\rho)\text{ tangent to }
E^\pm_\rho\\
vi)\ & \text{There exists an ``adapted'' metric $g_{\rm{ad}}$ near $K_E$
such that one can take $C=1$ in}\ iii).
\end{split}
\ee
The adapted metric $g_{\rm{ad}}$ can be extended to 
the whole energy layer, such as to coincide with the standard Euclidean metric
outside $T^*_{B(0,R_0)}X$.
We call 
\be\label{eq:weaks}
 E^{+0}_\rho \defeq E^+_ \rho \oplus\RR H_p(\rho) \,,  \qquad
E^{-0}_\rho \defeq E^-_\rho \oplus\RR H_p(\rho)\,,
\ee 
the weak unstable and weak stable 
subspaces at the point $\rho$ respectively. Similarly, we denote by
 $W^{+0}(\rho)$
(respectively $W^{-0}(\rho)$) the weak unstable (respectively stable) manifold. 
The ensemble of all the (un)stable
manifolds $W^{\pm}(\rho)$ forms the (un)stable {\em lamination} on $K_E$, and one
has 
$$
\Gamma_E^{\pm} = \cup_{\rho\in K_E} W^{\pm}(\rho)\,.
$$


If periodic orbits
are dense in $K_E$, then the flow is said to be {\em Axiom A} on $K_E$ \cite{BowRue}.

Such a hyperbolic set is {\em structurally stable} \cite[Theorem~18.2.3]{KaHa}, so that
\be\label{e:struct-stab}
\exists\delta>0,\ \forall E'\in [E-\delta,E+\delta],\quad
\text{$K_{E'}$ is a hyperbolic set for $\Phi^t\rest_{\cE_{E'}}$.}
\ee
Since the topological pressure plays
a crucial r\^{o}le in the statement and proof of 
Theorem \ref{t:1g}, we recall its definition in our
context (see \cite[Definition 20.2.1]{KaHa} or 
\cite[Appendix A]{PeSa}). 

Let $d$ be the distance function associated with the adapted metric.
We say that a set $ \cS \subset K_E $ is 
{\em $(\epsilon,t)$-separated} if for $ \rho_1, \rho_2 \in \cS $, 
$ \rho_1 \neq \rho_2 $, we have $ d ( \Phi^{t'} ( \rho_1 ) , 
\Phi^{t'} ( \rho_2 ) )>\eps $ for some $ 0 \leq t' \leq t $. Obviously,
such a set must be finite, but its cardinal may grow exponentially with $t$.
The metric $g_{\rm{ad}}$ induces a volume form $\Omega$ on any $n$-dimensional 
subspace of $T(T^*\RR^n)$. 
Using this volume form, we now define the
unstable Jacobian on $K_E$. For any $\rho\in K_E$, the determinant map
$$
\wedge^{n}\,d\Phi^t (\rho) \rest_{E^{+0}_ \rho}:\wedge^{n}E^{+0}_ \rho\longrightarrow 
\wedge^{n}E^{+0}_{\Phi^t(\rho)}
$$
can be identified with the real number
\be\label{e:determ}
\det\big(d\Phi^t ( \rho ) \rest_{ E^{+0}_\rho} \big)\defeq 
\frac{\Omega_{\Phi^t(\rho)}(d\Phi^t v_1\wedge d\Phi^t v_2\wedge\ldots\wedge d\Phi^t v_{n})}
{\Omega_{\rho}( v_1\wedge  v_2\wedge\ldots\wedge  v_{n})}\,,
\ee
where $(v_1,\ldots,v_{n})$ can be any basis of $E^{+0}_ \rho$.
This number defines the unstable Jacobian:
\be\label{e:unst-jac}
\exp \lambda^+_t(\rho) \defeq 
 \det\big(d\Phi^t ( \rho ) \rest_{ E^{+0}_ \rho  } \big)\,.
\ee
From there we take
\be\label{e:partitionf}
Z_t ( \eps , s) \defeq \sup_{ \cS} 
\sum_{ \rho \in \cS } \exp \left(- s\, \lambda^+_t ( \rho) \right) \,, 
\ee
where the supremum is taken over all $(\epsilon,t)$-separated sets. 
The pressure is then defined as
\be\label{eq:defpr}
 P_E ( s ) \defeq \lim_{\eps \to 0 } 
\limsup_{ t \to \infty }\; \frac 1 t \log Z_t ( \eps , s ) \,.
\ee
This quantity is actually independent of the volume form $\Omega$:
after taking logarithms, a change in $\Omega$
produces a term $ \cO ( 1 ) / t $ which is irrelevant in the 
$ t \to \infty $ limit. 

We remark that the standard definition
of the unstable Jacobian consists in restricting $\Phi^t ( \rho )$ on the 
{\em strong} unstable subspace $E^{+}_ \rho $; yet, 
including the flow direction in the definition 
\eqref{e:unst-jac} does not alter the pressure, and is better
suited for the applications in this article.
In \S\ref{s:pressure} we will give a different equivalent definition of the 
topological pressure, more adapted to our aims.

\subsection{Definition of resonances through complex scaling}\label{defcs}

We briefly
recall the complex scaling method -- see \cite{Sj} and references
given there.
Suppose that $ P = P ( h ) $ satisfies the assumptions of 
\S \ref{ass}. Here we can consider $ h $ as a fixed parameter
which plays no r\^ole in the definition of resonances. 

For any $\theta\in[0,\theta_0]$, let $ \Gamma_\theta
\subset \CC^n $ be a totally real contour with the following properties:
\begin{gather}
\label{eq:gpr}
\begin{gathered}
\Gamma_\theta \cap B_{\CC^n } ( 0 , R_0 ) = B_{\RR^n } ( 0 , R_0 ) \,, \\
\Gamma_\theta \cap \CC^n \setminus B_{\CC^n } ( 0 , 2 R_0 ) =
e^{ i \theta } \RR^n \cap \CC^n \setminus B_{\CC^n } ( 0 , 2 R_0 ) \,,\\
\Gamma_\theta = \{ x + i F_\theta ( x) \; : \; x \in \RR^n \} \,,  \  \
\partial_x^{\alpha}  F_\theta ( x ) =
\Oo_{\alpha} ( \theta )  \,.
\end{gathered}
\end{gather}
Notice that $F_\theta(x)=(\tan\theta) x$ for $|x|>2R_0$. By gluing $\Gamma_\theta\setminus B(0,R_0)$
to the compact piece $X_0$ in place of each infinite branch $\RR^n\setminus B(0,R_0)$, we obtain
a deformation of the manifold $ X $, which we denote by $ X_\theta $.

The operator $ P$ then defines a dilated operator:
\[ 
P_\theta \defeq P^\sharp\rest_{ X_\theta } \,,
\ \  P_\theta u  = P^\sharp ( u^\sharp ) \rest_{X_\theta } \,,
\]
where $ P^\sharp $ is the holomorphic continuation of the
operator $ P $, and $ u^\sharp$ is an almost analytic extension of
$ u \in \CIc ( X_\theta ) $ 

For $ \theta $ fixed and $E>0$, the scaled operator $ P_\theta - E $
is uniformly elliptic in $ \Psi_h ( X_\theta, \langle \xi \rangle^2 ) $, outside a 
compact set,
hence the resolvent, 
$ ( P_\theta - z )^{-1} $, is meromorphic for  $ z \in D ( E , 1 /C ) $.
We can also take $ \theta $ to be $ h $ dependent 
and the same statement holds for $ z \in D ( E , \theta / C ) $. 
The spectrum of $ P_\theta $ with $ z \in D ( E , \theta / C ) $ 
is independent of $ \theta $ and 
consists of {\em quantum resonances} of $ P$. The latter are generally defined as the
poles of the meromorphic continuation of
$$
( P - z )^{-1} \; : \; \CIc ( X ) \; \longrightarrow \; \CI ( X )
$$
from $D ( E , \theta /C )\cap \{\Im z>0\}$ to $D ( E , \theta /C )\cap\{\Im z <0\}$.
The {\em resonant states} associated with a resonance $ z $, $ \Re z \sim E >0$, 
$ |\Im z| < \theta / C $, are solutions to $ ( P - z ) u = 0 $ satisfying
\begin{gather}
\label{eq:gath}
\begin{gathered}
\ \exists \, U \in \CI \left( \Omega_\theta \right) \,, \ \
\Omega_\theta \defeq \bigcup_{ -\epsilon < \theta' < \theta + \epsilon } X_{\theta'} \\ 
  u = U\rest_X \,, \ \ u_{\theta'} = U\rest_{X_{\theta'} } \,, \ \
( P_{\theta'}  - z ) u_\theta' =0 \,, \ \ 0 < \theta' < 
\theta \,, \ \ u_\theta \in L^2 ( X_\theta ) \,. 
\end{gathered}
\end{gather}
If the multiplicity of the pole is higher there is a possibility of
more complicated states but here, and in Theorem \ref{t:2}, we 
consider only resonant states satisfying $ ( P - z ) u = 0 $.
At any pole of the meromorphically continued resolvent, 
such states satisfying \eqref{eq:gath}
always exist. We shall also call a nontrivial $ u_\theta $ 
satisfying $ ( P_\theta - z ) u_\theta = 0$, $ u_\theta \in L^2 ( X_\theta ) $,
a resonant state.

If $ \theta $ is small, as we shall always assume, we identify 
$ X $ with $ X_\theta $ using the map, $ R : X_\theta \rightarrow X $, 
\be\label{eq:idt}
X_\theta \ni x \longmapsto \Re x \in X \,,
\ee
and using this identification, consider $ P_\theta $ as an operator on $ X $,
defined by $ (R^{-1})^* P_\theta R^* $. 
We note that in the identification of $ L^2 ( X ) $ with $ L^2 ( X_\theta ) $
using $ x \mapsto \Re x $, 
\[ 
C^{-1}\,\| u ( h ) \|_{ L^2 ( X) } 
\leq \| u ( h ) \|_{ L^2 ( X_\theta ) } \leq C\, \| u ( h ) \|_{L^2 (X ) }\,,
\]
with $ C $ independent of $ \theta $ if $ 0 \leq \theta \leq 1/C_1 $.

For later use we conclude by describing the principal symbol of $ P_\theta $, 
as an operator on $ L^2 ( X )$ using the identification above:
\begin{equation}
\label{eq:pth}
p_\theta ( x, \xi ) = p ( x + i F_\theta ( x) , [( 1 + i dF_\theta (x)^t)]^{-1}
\xi ) \,, 
\end{equation}
where the complex arguments are allowed due to the analyticity of 
$ p ( x, \xi ) $ outside of a compact set --- see \S \ref{ass}.
In this paper we will always take $ \theta = \Oo (\log ( 1/h) h) $
so that $ p_\theta ( x, \xi ) - p ( x , \xi ) = 
\Oo ( \log( 1/h) h ) \langle \xi \rangle^2 $. More precisely,
\begin{equation}
\label{eq:pth1}
\begin{split}
& \Re p_\theta ( x , \xi) = p ( x, \xi ) + {\mathcal O} ( \theta^2) 
\langle \xi \rangle^2 
 \,, \\ 
& \Im p_\theta ( x, \xi ) = -  d_\xi p ( x, \xi ) [ dF_\theta ( x )^t \xi ]
+ d_x p (x , \xi ) [ F_\theta ( x ) ] + \Oo ( \theta^2) 
\langle \xi \rangle^2  \,. 
\end{split}
\end{equation}
In view of (\ref{eq:valid}) and (\ref{eq:validC}), 
we obtain the following estimate 
when $ |x | \geq R_0 $: 
\be\label{eq:pth2} 
\Im p_\theta ( x, \xi ) = - 2 \langle dF_\theta ( x ) \xi , \xi \rangle
+ \Oo\big(\theta\, (f_0(|x|) + f_{1}(|x|)) + 
\theta^2 \big) \langle \xi\rangle^2 \,,
\ee
where 
$f_j(r) 
\rightarrow 0$ as $r\to\infty$. 
In particular, if $R_0$ is taken large enough,
\be\label{eq:pth3}
(x,\xi)\in \cE_E^{\delta},\ \ 
|x | \geq 2\,R_0 \Longrightarrow  \Im p_\theta ( x, \xi) \leq - C \theta \,.
\ee

\section{Semiclassical Fourier integral operators and their iteration}\label{fio}
The crucial step in our argument is the analysis of compositions of a
large number --- of order $ \log ( 1/h ) $ --- of local Fourier integral operators.
This section is devoted to general aspects of that procedure, which will 
then be applied in \S\ref{s:estimate}.

\subsection{Definition of local Fourier integral operators}
We will review here the {\em local} theory of these operators
in the semiclassical setting. Let $ \kappa : T^* \RR^n \rightarrow
T^* \RR^n $ be a local diffeomorphism defined near $ ( 0 , 0 ) $, and satisfying
\be\label{e:fixed}
\kappa ( 0 , 0 ) = ( 0 , 0 ) \,, \quad \kappa^* \omega = \omega \,. 
\ee
(Here $\omega$ is the standard symplectic form on $T^* \RR^n$).
Let us also assume that the following projection from the graph of $ \kappa$, 
\be\label{eq:surj}
  T^* \RR^n \times T^* \RR^n \ni 
( x^1,\xi^1 ; x^0 , \xi^0 ) \longmapsto ( x^1, \xi^0 ) 
\in \RR^n \times \RR^n \,, \quad (x^1,\xi^1)=\kappa ( x^0 , \xi^0 )\,,
\ee
is a diffeomorphism near the origin. It then follows that there
exists (up to an additive constant) a unique function
$ \psi \in \CI ( \RR^n \times \RR^n ) $, 
such that for $ ( x^1, \xi^0) $
near $ ( 0 , 0 ) $, 
$$
\kappa ( \psi'_\xi ( x^1, \xi^0 ) , \xi^0 ) = ( x^1, \psi'_x ( x^1, \xi^0 ) ) \,, \ \ 
\det \psi''_{ x \xi } \neq 0 \,, \ \ \psi ( 0, 0 ) = 0 \,. 
$$
The function $\psi$ is said to {\em generate} the transformation $\kappa$ near $(0,0)$.
The existence of such a function $\psi$ in a small neighbourhood of $(0,0)$ is equivalent with
the following property: the $n\times n$ block $\big(\frac{\partial x^1}{\partial x^0}\big)$ 
in the tangent map $d\kappa(0,0)$ is invertible.

A {\em local semiclassical quantization of $ \kappa $} is an operator $T=T(h)$ acting
as follows~:
\be\label{eq:lfio}
T\, u ( x^1 ) \defeq \frac{1}{ (2 \pi h)^n } \int \int e^{ i ( \psi ( x^1, \xi^0 ) -
\langle x^0 , \xi^0 \rangle )/h } \alpha ( x^1, \xi^0 ; h ) u ( x^0 ) d x^0 d \xi^0 \,.
\ee
Here the amplitude $ \alpha $ is of the form
$$ 
\alpha ( x, \xi; h ) = \sum_{ j=0}^{L-1} h^j\,\alpha_j ( x, \xi )  + 
h^L\, \tilde\alpha_L ( x, \xi ; h ) \,, \quad \forall L\in\NN \,, 
$$
with all the terms, $ \alpha_j \,, \tilde\alpha_L \in S ( 1 ) $, 
supported in a fixed neighbourhood of $ ( 0 , 0 ) $. Such an operator $T$ is
a local Fourier integral operator associated with $\kappa$. 

We list here several basic properties of $ T $ -- see for instance
\cite[\S 3]{SjZw02} and \cite[Chapter 10]{EZ}:
\begin{itemize}
\item We have $ T^* T = A^w ( x, h D ) $, $ A \in S ( T^* \RR^n ) $, 
\be\label{eq:unit}  A ( \psi'_\xi (x^1, \xi^0 ) , \xi^0 ) = 
\frac{| \alpha_0 ( x^1 , \xi^0 ) |^2}{ |\det \psi_{x \xi }'' ( x^1 , \xi^0 ) 
|}   + {\mathcal O}_{S(1)} ( h ) \,.
\ee
In particular, $ T $ is bounded on $ L^2 $ uniformly with 
respect to $ h $.

If $ T^* T = I $ microlocally near $U\supset ( 0 , 0 ) $, then
\be\label{e:unitarity}
|\alpha_0 ( x^1 , \xi^0 ) | = |\det \psi_{x \xi}'' ( x^1 , \xi^0 ) |^{\frac12} \quad
\text{for $(x^0,\xi^0)$ near}\ \ U.
\ee 

\medskip

\item If $ \alpha ( 0 , 0 ) \neq 0 $ then $ T $ is microlocally 
invertible near $ ( 0 , 0) $: there exists an operator $S$
of the form \eqref{eq:lfio} quantizing $ \kappa^{-1} $, such that
$ST = I$ and $TS = I$ microlocally near $(0,0)$.

\medskip

\item 
For $ b \in S ( 1 ) $, 
$$
T\, b^w ( x, h D ) = c^w ( x , h D)\, T + \cO_{L^2 \to
L^2 } ( h ) \,, \quad  \kappa^* c \defeq c\circ \kappa
= b\,.
$$
Moreover, 
if $ \alpha ( 0 , 0 ) \neq 0 $, 
then for any $ b \in S ( 1 ) $ supported in a sufficiently small 
neighbourhood of $ ( 0 , 0 ) $, 
\be
\label{eq:egor}
 T\, b^w ( x, h D ) = c^w ( x , h D) \,T \,, \quad  \kappa^* c
= b + \cO_{S( 1) } ( h ) \,. 
\ee
The converse is also true: if $ \kappa $
satisfies the projection properties \eqref{eq:surj} and 
$ T $ satisfies \eqref{eq:egor} 
for all $ b\in S ( 1 ) $ with support near $ ( 0 , 0 ) $, then 
$ T $ is equal to an operator of the form \eqref{eq:lfio} microlocally near $(0,0)$.
The relation \eqref{eq:egor} is a version of {\em Egorov's theorem}
and we will frequently use it below.

\medskip

\item For $ b\in S(1)$ we have $b^w(x,hD)\,T = \tilde T +\cO_{L^2\to L^2}(h^\infty)$,
where $\tilde T$ is of the form \eqref{eq:lfio} with the same phase $\psi(x^1,\xi^0)$,
but with a different symbol $\beta( x^1, \xi^0 ; h )\in S ( 1 )$.
Its principal symbol reads $\beta_0(x^1, \xi^0)=b(x^1,\psi'_x(x^1,\xi^0))\, \alpha_0( x^1, \xi^0)$,
and the full symbol $\beta$ is supported in $\supp \alpha$.
\end{itemize}
The proofs of these statements are similar to the proof 
of the next lemma, which is an application of the stationary phase method
and a very special case of the composition formula for 
Fourier integral operators.

\begin{figure}[ht]
\begin{center}
\includegraphics[width=16cm]{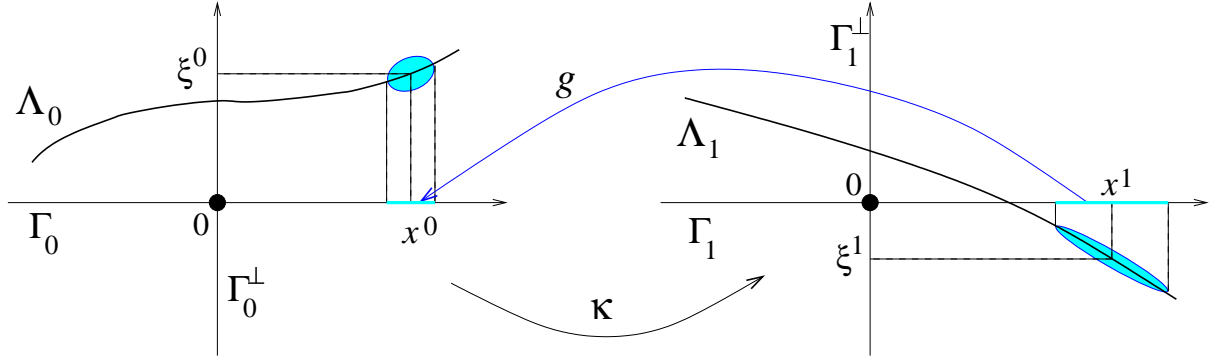}
\end{center}
\caption{A schematic illustration of the objects appearing in 
Lemma \ref{l:0}. We labelled the $x^j$ and $\xi^j$ axes respectively by $\Gamma_j$ and $\Gamma^\perp_j$,
in order to represent also the more general case of \eqref{eq:surj1}.}
\label{f:4}
\end{figure}

\begin{lem}\label{l:0}
We consider
a Lagrangian $ \Lambda_0 = \set{ ( x , \varphi_0' ( x ) ),\ x\in \Omega_0 } $,
$\varphi_0\in C^\infty_b(\Omega_0)$, contained
in a small neighbourhood $V\subset T^*\RR^n$, such that $\kappa$ is generated by
$\psi$ near $V$. We assume that
\be\label{eq:kapl}
 \kappa ( \Lambda_0 ) = \Lambda_1 = \{ ( x ,  \varphi_1' ( x ) )\,,\ x\in \Omega_1 \}\,,
\quad \varphi_1\in C^\infty_b(\Omega_1) \,.
\ee
Then, for any symbol $a\in \cC^\infty_c(\Omega_0)$, the application 
of $ T $ to the {\em Lagrangian state} 
$a(x)\, e^{i \varphi_0(x)/ h }$ associated with $\Lambda_0$ satisfies
\be\label{eq:l120}
T \big( a \, e^{i \varphi_0/ h }  \big) (x)
= e^{ i \varphi_1 ( x ) / h }\, \big( \sum_{j=0}^{L-1} b_j ( x ) h^j + 
h^L r_L ( x,h ) \big)\,,
\ee
where the coefficients $ b_j $ are described as follows. 
Consider the map:
\be
 \Omega_1 \ni x \mapsto g(x)\defeq \pi \circ \kappa^{-1} \big( x , \varphi_1' ( x )\big)\in\Omega_0 \,.
\ee
Any point $x^1\in \Omega_1$, is mapped by $g$ to 
the unique point 
$x^0$ satisfying
$$
\kappa(x^0,\varphi_0'(x^0))= (x^1,\varphi_1'(x^1))\,.
$$
The principal symbol $b_0$ is then given by
\be\label{eq:l1300}
  b_0 ( x^1) = e^{ i \beta_0 / h }\; \frac{\alpha_0 ( x^1 , \xi^0) }
{|\det \psi_{x \xi}'' ( x^1 , \xi^0  ) |^{\frac12} }\;
|\det  d g (x^1)|^{\frac12}\; a \circ g(x^1) \,, \quad
 \beta_0 \in \RR ,\quad \xi^0=\varphi'_0\circ g(x^1)\,,
\ee
and it vanishes outside $\Omega_1$. 
Furthermore, we have for any $\ell\geq\NN$:
\be\label{eq:l130}
\begin{split}
\norm{ b_j }_{ C^\ell ( \Omega_1 ) } &\leq C_{ \ell,j } \norm{ a }_{ C^{\ell+2j}( \Omega_0)  } \,, 
\qquad 0\leq j\leq L-1 \,, \\
\norm{ r_L ( \bullet , h ) }_{ C^\ell ( \Omega_1 ) } &\leq C_{ \ell , L  } 
\norm{ a }_{ C^{ \ell + 2 L + n} ( \Omega_0 ) } \,.
\end{split}
\ee
The constants $ C_{\ell , j } $ depend only on $ \kappa $,
$ \alpha $, and 
$  \sup_{\Omega_0} |\partial^\beta \varphi_0 | $, for 
$ 0 < |\beta| \leq 2 \ell + j $. 
\end{lem}
\begin{proof}
The stationary points of the phase 
in the integral defining $ T ( a\, e^{ i \varphi_0/h })(x^1)$ are obtained by solving:
\[
d_{ x^0, \xi^0} \big( \psi ( x^1 , \xi^0 ) - \langle x^0 , \xi^0 \rangle + 
\varphi_0 ( x^0 ) \big) = 0 \; \Longleftrightarrow \; 
\left\{ \begin{array}{l} \xi^0 = \varphi_0' ( x^0 ) \,, \\
x^0 = \psi_\xi' ( x^1 , \xi^0 ) \,.  \end{array} \right. 
\]
The assumption \eqref{eq:kapl} implies, for $x^1\in\Omega^1$, the existence of a unique solution
$ x^0 = g(x^1) $, $ \xi^0 = \varphi_0'\circ g(x^1) $, 
and the nondegeneracy of the Hessian of the phase.
One also checks that, after inserting the dependence $x^0(x^1)$, $\xi^0(x^1)$ in
the critical phase, the derivative of the latter satisfies
\[ 
d_{x^1} \big( \psi ( x^1 , \xi^0(x^1)) - \langle x^0(x^1) , \xi^0(x^1) \rangle 
+ \varphi_0 ( x^0(x^1) )\big) = \varphi_1' ( x^1 ) \,.
\]
This shows that the critical phase is equal to $ \varphi_1 ( x^1 ) $,
up to an additive a constant.

The stationary phase theorem (see for instance \cite[Theorem 7.7.6]{Hor})
now shows that \eqref{eq:l120} holds with 
\begin{align}
\label{eq:b_0}
b_0 ( x^1 ) &= e^{ i \beta_0 / h }
|  \det ( I - \psi''_{\xi\xi} ( x^1 , \xi^0 ) \circ
 \varphi_0'' ( x^0 ) ) | ^{-\frac12}\,\alpha ( x^1, \xi^0)\, a ( x^0 ) \,, \\
b_j ( x^1 ) &= \sum_{j'=0}^j L_{j'} ( x^1 , D_{x,\xi} ) \big( \alpha_{j-j'} ( x^1 , \xi ) a ( x ) \big) 
\rest_{ \xi = \xi^0, x = x^0 } \,. \label{eq:b_j}
\end{align}
Each $ L_j ( x, D_{x,\xi} ) $ is a differential operator of order
$ 2 j $, with coefficients of the form 
\[ 
\frac{ P_{j\gamma} ( x^1 )} {  \det ( I - \psi''_{\xi\xi} 
( x^1 , \xi^0) \circ \varphi_0'' ( x^0 ) )^{3j} } \,, 
\]
where $ P_{ j \gamma } $ is a polynomial of degree $\leq 2j $
in the derivatives of $ \psi$ and $ \varphi_0 $, of 
order at most $ 2 j + 2 $ (the right hand side of \eqref{eq:b_0} can also be 
written as $L_0 (\alpha\, a)$). 
The remainder $ r_L ( x^1; h ) $ 
is bounded by a constant (depending on $ M $ and $ n $) times
\[  
\frac{
\big( \sum_{ | \alpha | \leq 2 L } \sup_{ x, \xi } | \partial^\alpha
_{x, \xi } \big( \psi( x^1, \xi ) - \la x , \xi \ra + \varphi_0(x)\big) | \big)^{2L} 
\big( \sum_{ | \alpha | \leq 2 L + n} \sup_{ x, \xi } | \partial^\alpha
_{x, \xi } \big( \alpha ( x^1, \xi ) a ( x ) \big) | \big)}
{\inf_{x,\xi}   |  \det \big( I - \psi''_{\xi\xi} ( x^1 , \xi ) \circ 
\varphi_0'' ( x ) \big) |^{ 3L} }  \,,
\]
with similar estimates for the derivatives $\partial^\ell r_L(\bullet;h)$.
The bounds \eqref{eq:l130} follow from the structure of the operators $L_j$,
and the above estimate on the remainder.

It remains to identify the determinant appearing in \eqref{eq:b_0}
with the more invariant formulation in \eqref{eq:l1300}.
The differential, $d\kappa (x^0,\xi^0 ) $, is the map
$(\delta {x^0},\delta {\xi^0} )\mapsto (\delta {x^1},\delta {\xi^1}
)$, where
\[
\begin{split} & \delta {x^0}=\psi ''_{\xi x}\delta {x^1} + 
\psi''_{\xi \xi} \delta {\xi^0} \\
&  \delta {\xi^1} =
\psi''_{x \xi }\delta {\xi^0} +\psi ''_{xx}\delta {x^1} \,,
\end{split} 
\]
and the $\psi''$ are evaluated at $ ( x^1, \xi^0 ) $.
By expressing  $\delta {x^1},\,\delta {\xi^1} $ in
terms of $\delta {x^0},\,\delta {\xi^0} $ we get
\be\label{e:dkappa}
d \kappa ( x^0 , \xi^0) = 
\left( \begin{array}{ll} (\psi ''_{\xi x})^{-1}
& -(\psi _{\xi x}'') ^{-1} \psi _{\xi \xi }'' \\
\psi ''_{xx}(\psi ''_{\xi x})^{-1} & \psi ''_{x\xi }-
\psi _{xx}''(\psi ''_{\xi x})^{-1}\psi''_{\xi \xi }  
\end{array} \right)  \,.   
\ee
The upper left block in this matrix is indeed invertible, as explained at the beginning
of the section.
From \eqref{e:dkappa} we also see that the restriction of $ d\kappa $ to 
$ \Lambda_0 $ followed by the projection $ \pi $ is 
given by 
\[ 
\delta{x^0} \longmapsto \delta x^1=( \psi_{ \xi x }'' )^{-1} ( I - \psi_{\xi \xi}''
\circ \varphi_0 '' ) ( \delta{x^0} ) \,. 
\]
Hence, noting that 
$g=\pi \circ \kappa^{-1} \circ ( \pi \rest_{\Lambda_1 })^{-1}=
(\pi \circ \kappa \circ ( \pi\rest_{\Lambda_0 })^{-1})^{-1}$, we get
\[
 \det (d g (x^1)) = 
\frac{\det \psi_{\xi x } ''(x^1,\xi^0)} 
{ \det( I - \psi_{\xi \xi} ''(x^1,\xi^0) \circ \varphi_0 ''(x^0) ) } \,, 
\]
which completes the proof of \eqref{eq:l1300}.
\end{proof}

We want to generalize the above considerations by relaxing the structure of 
$ \kappa $: we only assume that $\kappa$ is locally a canonical diffeomorphism such that
$\kappa(0,0)=(0,0)$.
Without loss of generality, we can find linear Lagrangian subspaces,
$ \Gamma_j,\,\Gamma_j^\perp \subset T^*\RR^n$ ($j=0,1$), with the following 
properties:
\begin{itemize}
\item $\Gamma_j^\perp$
is transversal to $ \Gamma_j$ 
(that is, $\Gamma_j^{\perp}\cap \Gamma_j=\set{0}$)\footnote{Here $\Gamma^\perp$ is 
{\em not} the symplectic annihilator
of $\Gamma$ -- see for instance \cite[Sect.21.2]{Hor}}
\item if $ \pi_j $ (respectively $ \pi_j^\perp $) is
the projection $ T^* \RR^n \rightarrow \Gamma_j $ along $ \Gamma_j^\perp$,
(respectively the projection $ T^* \RR^n \rightarrow \Gamma_j^\perp $ along $ \Gamma_j $), 
then, for some neighbourhood $U$ of the origin, the map
\be\label{eq:surj1} 
\kappa(U)\times U  \ni ( \kappa ( \rho ) , \rho ) 
\longmapsto \pi_1( \kappa (\rho) ) \times \pi_0^\perp (\rho ) 
\in \Gamma_1 \times \Gamma_0^\perp 
\ee
is a local diffeomorphism. If we write the
tangent map $d\kappa(\rho)$ as a matrix from $\Gamma_0\oplus \Gamma_0^\perp$ to
$\Gamma_1\oplus\Gamma_1^\perp$, then the upper-left block is invertible.
\end{itemize}

Let $ A_j $ be linear symplectic transformations with the properties
$$  
A_j ( \Gamma_j ) = \{ ( x , 0 ) \} \subset T^* \RR^n \,,
\quad  
 A_j ( \Gamma_j^\perp ) = \{ ( 0 , \xi ) \} \subset T^* \RR^n \,, 
$$
and let $ M_j $ be {\em metaplectic quantizations} of $ A_j $'s (see 
\cite[Appendix to Chapter 7]{DiSj} for a self-contained presentation
in the semiclassical spirit). Then the rotated diffeomorphism
\be\label{e:conjug}
 \tilde\kappa \defeq A_1 \circ \kappa \circ A_0^{-1} 
\ee
has the properties of the map $\kappa$ in Lemma \ref{l:0}. Let $ \tilde T $ 
be a quantization of $\tilde\kappa$ as in \eqref{eq:lfio}. Then 
\be\label{eq:Tdef}  
T \defeq M_1^{-1} \circ \tilde T \circ M_0
\ee
is a quantization of $\kappa$. 

By transposing Lemma~\ref{l:0} to this framework, we may apply $T$ 
to Lagrangian states supported on a Lagrangian $\Lambda_0$, $\kappa(\Lambda_0)=\Lambda_1$,
such that $ \pi_j : \Lambda_j \rightarrow \Gamma_j $ is 
locally surjective,  $j=0,1$.
The action of $\kappa^{-1}$ on $\Lambda_1$ can now represented by
the function
\be
g=\pi_0 \circ \kappa^{-1} \circ 
( \pi_1 \rest_{\Lambda_1} )^{-1}: \Gamma_1\to \Gamma_0\,.
\ee

Finally, performing phase space translations, we may relax the condition
$\kappa(0,0)=(0,0)$. 

\subsection{The Schr\"odinger propagator as a Fourier integral operator}\label{s:Schrod=FIO}

Using local coordinates on the manifold $X$, the above formalism applies to
propagators acting on $L^2(X)$.
\begin{lem}\label{l:lprop}
Suppose that $ P ( h )  $ satisfies the assumptions of \S \ref{ass}, 
$$ 
V_0 \Subset \cE_E^\delta \defeq p^{-1} ( ( E-\delta, E+\delta))  \,, 
\ \ 
 \chi \in S ( 1) \,, \ \ 
 \chi \rest_{\cE_E^\delta} \equiv 1 \,, \ \ 
 V_1 \subset \Phi^t ( V_0 ) \,. $$
 For a fixed time $t>0$, let 
\be\label{e:U_chi}
U_\chi ( t ) \defeq \exp ( - it \chi^w ( x, h D )\,  P ( h ) \,
\chi^w ( x, h D ) / h ) \,, 
\ee
be a modified unitary propagator of $ P $, acting on $L^2(X)$.

Take some $ \rho_0 \in V_0 \cap\cE_E $, and call $\rho_1=\Phi^t(\rho_0)\in V_1$.
Let $ f_0 : \pi ( V_0 ) \rightarrow \RR^n $, 
$ f_1 : \pi ( V ) \rightarrow \RR^n $ be local coordinates such 
that $ f_0 ( \pi ( \rho_0 )) = f_1 ( \pi ( \rho_1 ) ) = 0 \in 
\RR^n $. They induce on $V_0,\,V_1$ the symplectic coordinates
\be\label{e:F_U}
F_j ( x , \xi ) \defeq ( f_j ( x ) , (d f_j ( x ) ^t )^{-1} 
\xi  - \xi^{(j)}) \,, \quad j=0,1\, ,
\ee
where $\xi^{(j)}\in\RR^n$ is fixed by the condition $F_j(\rho_j)=(0,0)$.
Then the operator on $L^2(\IR^n)$,
\be
\label{eq:lpr}
T^\sharp(t) \defeq   e^{ - i \langle x, \xi^{(1)} \rangle/h }\, (f_1^{-1})^*\, U_\chi ( t )\, (f_0)^* 
e^{ i \langle x, \xi^{(0)} \rangle / h } \,,
\ee
is of the form \eqref{eq:Tdef} for some choices of $ A_j$'s,
microlocally near $ ( 0 , 0 ) $.
\end{lem}
Although complicated to write, the lemma simply states that
the propagator is a Fourier integral operator in the 
sense of this section. 

\medskip

\noindent{\em Sketch of the proof of Lemma \ref{l:lprop}:} 
The first step is to prove that for $ a \in S ( 1 ) $ with 
support in $ \chi \equiv 1 $ we have 
\be\label{eq:Uchi}
  U_\chi ( t)^{-1} a^w ( x, h D) U_\chi ( t ) = a_t^w ( x, h D) \,, \ \ 
a_t = (\Phi^t)^* a + {\mathcal O}_{S ( 1 ) } ( h ) \,. 
\ee
This can be see from differentiation with respect to $ t $~:
$$
\partial_t a_t^w = \frac{i}{h} [ \chi^w P \chi^w , a_t^w ] = 
\frac{i}{h} [ P , a_t^w ] + {\mathcal O} ( h^\infty ) \,, \ \ 
a_0^w = a^w \,. 
$$
Since  $ (i/h) [ P, a_t^w ] = ( H_p a_t )^w + \cO ( h ) $
we conclude that $ a_t^w = [( \Phi_t )^* a]^w + \cO_{L^2 
\rightarrow L^2 } ( h ) $. An iteration of this argument shows \eqref{eq:Uchi}
(see \cite[Chapter 9]{EZ} and the proof of Lemma \ref{l:new1}
below). The converse to Egorov's 
theorem (see \cite[Lemma 3.4]{SjZw02} or \cite[Theorem 10.7]{EZ}) 
implies that \eqref{e:U_chi} is a quantization of $\Phi^t$, microlocally
near $\rho_0 \times \rho_1$.

On the classical level, the symplectic coordinates $F_0$, $F_1$ of \eqref{e:F_U} are such
that the symplectic map
$$
\kappa \defeq F_1 \circ \Phi^t \circ F_0^{-1} \quad\text{satisfies}\quad \kappa(0,0)=(0,0)\,.
$$
Hence the operator $T^\sharp(t)$ is a quantization of $\kappa$, and
can be put in the form \eqref{eq:Tdef} for some choice of symplectic rotations
$A_j$, microlocally near $(0,0)$. A possible choice of these rotations
is given in Lemma~\ref{l:Gam01} below.
\stopthm

\medskip

We will now describe a particular choice of coordinate chart in the neighbourhood 
$U_\rho$ of an arbitrary
point $\rho\in K_E$. Using the notation of the previous lemma, $U_\rho$ 
may be identified through a
symplectic map $F_\rho$
with a neighbourhood of $(0,0)\in T^*\RR^{n}$. This way, Lagrangian (respectively isotropic) subspaces 
in $T_\rho(T^*X)$ are identified with Lagrangian (respectively isotropic) subspaces in $T_0(T^*\RR^{n})$.

We now recall that the weak stable and unstable subspaces $E^{\pm 0}_\rho$ 
defined by \eqref{eq:weaks} are Lagrangian. 
The proof of that well known fact is simple: for any two vectors $v,w\in E^{+}_\rho$, we have 
$$
\forall t\in\RR,\qquad\omega(v,w)=\Phi^{t*}\omega(v,w)=\omega(\Phi^t_*v,\Phi^t_*w)\,.
$$
By assumption, the vectors on the right hand side converge to zero when $t\to -\infty$, which proves that
strong unstable subspaces are isotropic. The same method 
shows that $\omega(v,H_p)=0$, so weak unstable subspaces are Lagrangian. The same results
apply to stable subspaces. Besides,
the isotropic subspace $E^{-}_\rho$ is transversal to the Lagrangian $E^{+0}_\rho$, so the tangent
space to the energy layer $ \cE_E $ at $\rho$ is decomposed into 
$T_\rho\cE_E = E^{+0}_\rho\oplus E^{-}_\rho$.
\begin{lem}\label{l:coord}
Take any point $\rho\in K_E$. As above, we may identify
a neighbourhood $U_\rho\subset T^*X$ of $\rho$ with a neighbourhood of $(0,0)\in T^*\RR^{n}$. 
The tangent
space $T_\rho(T^*X)$ is then identified with $T_0(T^*\RR^{n})\equiv T^*\RR^{n}$.

The space $T^*\RR^{n}$ can be equipped with a {\em symplectic} basis
$(e_1,\ldots,e_n;f_1,\ldots,f_n)$ such that $e_1=H_p(\rho)$, $E^+_\rho=\Span\{e_2,\ldots,e_n\}$ and
$E^-_\rho = \Span\{f_2,\ldots,f_n\}$. We also require that $\Omega_\rho(e_1\wedge\cdots\wedge e_n)=1$,
where $\Omega$ is the volume form on $E^{+0}_\rho$ induced by the adapted metric $g_{\rm{ad}}$ 
(see $vi)$ in \eqref{eq:aaa}).
The two Lagrangian subspaces
$$
\Gamma\defi E^{+0}_\rho\quad\text{and}\quad \Gamma^\perp \defi E^{-}_{\rho}\oplus 
\RR f_1
$$
are transversal.
Let us call $(\ty_1,\ldots,\ty_n,\teta_1,\ldots,\teta_n)$ the linear symplectic coordinates 
on $T^*\RR^{n}$ dual to the basis $(e_1,\ldots,e_n;f_1,\ldots,f_n)$.

There exists a symplectic coordinate chart $(y,\eta)$ near $\rho\equiv (0,0)$, such that
\begin{gather}
\label{e:nonlin}\eta_1 = p-E\,, \quad \ 
\frac{\partial}{\partial y_i}(0,0)=e_i\quad\text{and}\quad 
\frac{\partial}{\partial \eta_i}(0,0)=f_i,\quad i=1,\ldots,n \, .
\end{gather}
Such a chart will be called {\em adapted} to the dynamics.
$(y,\eta)$ is mapped to $(\ty,\teta)$ through a local symplectic diffeomorphism
fixing the origin, and tangent to the identity at the origin.
\end{lem}
\begin{proof}
Once we select the Lagrangian $\Gamma = E^{+0}_\rho$, 
with the isotropic $E^{-}_\rho$ plane transversal 
to $\Gamma$, it is always possible to 
complete $E^-_\rho$ into a Lagrangian $\Gamma^\perp$ transversal to $\Gamma$, by adjoining
to $E^-_\rho$ a certain subspace $\RR v$. Since $\Gamma\oplus \Gamma^\perp$ spans 
the full space $T^*\RR^n$, the vector $v$ must be transversal to 
the energy hyperplane $T_\rho \cE_E$.

Since we took $e_1=H_p(\rho)$ we can 
equip $E^+_\rho$ with a basis $\{e_2,\ldots,e_n\}$ satisfying
$\Omega_\rho(e_1\wedge e_2\wedge\cdots\wedge e_n)=1$.
There is a unique choice of vectors $\{f_1,\ldots,f_n\}$ 
such that these vectors generate $\Gamma^\perp$
and satisfy $\omega(f_i,e_j)=\delta_{ij}$ for all $i,j=1,\ldots,n$.
The property $\omega(f_j,e_1)=0$ for
$j>1$ implies that $f_j$ is in the energy hyperplane, while $\omega(f_1,e_1)=dp(f_1)=1$ shows 
that 
$p((\ty,\teta))= E+\teta_1+\cO(\teta_1^2)$ when $\teta_1\to 0$.

From Darboux's theorem, there exists a (nonlinear) symplectic 
chart $(y^\flat,\eta^\flat)$ near the origin such that $\eta_1^\flat=p-E$.
There also exists a linear symplectic transformation $A$ such
that the coordinates $(y,\eta)=A(y^\flat,\eta^\flat)$ satisfy
$\eta_1=\eta_1^\flat$
as well as the properties \eqref{e:nonlin} on $T_0(T^*\RR^{n})$. The last statement 
concerning the mapping $(\ty,\teta)\mapsto (y,\eta)$ comes from the fact that the vectors
$\partial/\partial\ty_i$, $\partial/\partial\teta_i$ satisfy \eqref{e:nonlin} as well.
\end{proof}
\begin{lem}\label{l:Gam01}
Suppose that $ P $ satisfies the assumptions of \S \ref{ass} and
the hyperbolicity assumption \eqref{eq:aa}. Fixing $t>0$ and using the notation of 
Lemmas~\ref{l:lprop} and \ref{l:coord},
we consider the symplectic frames 
$\Gamma_0\oplus\Gamma^\perp_0$ and $\Gamma_1\oplus\Gamma_1^\perp$,
constructed respectively near $\rho_0$ and $\rho_1=\Phi^t(\rho_0)$.

Then, the graph of $\Phi^t$ near $\rho_1\times \rho_0$ projects 
surjectively to 
$\Gamma_1\times \Gamma_0^\perp$ (see \eqref{eq:surj1}). 
This implies that
the operator \eqref{eq:lpr} can be written in the form \eqref{eq:Tdef}, where
the metaplectic operators $M_j$ quantize the coordinate changes 
$F_j(x,\xi)\mapsto(\ty^j,\teta^j)$,
while $\tilde T(t)$ quantizes $\Phi^t$ written in 
the coordinates $(\ty^0,\teta^0)\mapsto (\ty^1,\teta^1)$.

The symplectic coordinate changes $(\ty^j,\teta^j)\mapsto (y^j,\eta^j)$
can be quantized by Fourier integral operators 
$T_0,\,T_1$ of the form \eqref{eq:lfio} and microlocally unitary. 
If we set $\cU_j\defeq T_j \circ M_j$, $j=0,1$,
the operator \eqref{eq:lpr} can then be written as
\be\label{e:Tsharp}
T^\sharp(t)=\cU_1^{*}\circ T(t) \circ \cU_0
\ee
microlocally near $(0,0)$, 
where $T(t)$ is a Fourier integral operator of the form \eqref{eq:lfio} which quantizes
the map $\Phi^t$, when written in the
adapted coordinates $(y^0,\eta^0)\mapsto (y^1,\eta^1)$.

\end{lem}
\begin{proof}
We may express the map $\Phi^t$ from $V_0$ to $V_1$ using the coordinate charts
$(y^0,\eta^0)$ on $V_0$, respectively $(y^1,\eta^1)$ on $V_1$.
The tangent map $d\Phi^t(\rho_0)$ is then given by a matrix of the form
\be
d\Phi^t (\rho_0)\equiv \begin{pmatrix}
1 & 0 & * & 0 \\
0 & A & * & 0 \\
0 & 0 & 1 & 0 \\
0 & 0 & * & \,^t\!A^{-1}
\end{pmatrix}.
\ee
Since the full matrix is symplectic, 
the block, 
$$
\begin{pmatrix}1 & 0\\0 & A\end{pmatrix}\,, 
$$
is necessarily invertible: this
implies that the graph of
$\Phi^t$ projects surjectively to 
 $\Gamma_1\times \Gamma_0^\perp$ in some neighbourhood of
$\rho_1\times \rho_0$. 
Equivalently, if we represent $\Phi^t$ near $\rho_1\times \rho_0$ as a map $\tilde\kappa$ in
the ``linear'' coordinates $(\ty^0,\teta^0)$ and $(\ty^1,\teta^1)$,
the graph of $\tilde\kappa$ projects surjectively to
$(\ty^1,\teta^0)$, so that the operator $\tilde T(t) = M_1\circ T^\sharp(t)\circ M_0^{-1}$ quantizing
$\tilde\kappa$ can put in the form \eqref{eq:lfio} near the origin.

For each $j=0,1$, the tangency of 
the charts $(\ty^j,\teta^j)$ and $(y^j,\eta^j)$ at the origin shows that the graph
of the the coordinate
change $(\ty^j,\teta^j)\mapsto (y^j,\eta^j)$ projects well on $(y^j,\teta^j)$,
so this change can be quantized by an operator $T_j$ of the form \eqref{eq:lfio}, microlocally
unitary near the origin. The operator
$T(t) = T_1\circ M_1\circ T^\sharp(t)\circ M_0^{*}\circ T_0^{*}$
quantizes $\Phi^t$, when written in the coordinates $(y^0,\eta^0)\mapsto (y^1,\eta^1)$, 
and can also written in the form \eqref{eq:lfio} near the origin.
\end{proof}

\subsection{Iteration of the propagators}\label{s:iterating}

Later we will compose operators of type $ U ( t_0 ) \, \Pi_a$,
where $ \Pi_a $ is a microlocal cut-off to a small neighbourhood
$W_a\subset \cE_E^\delta $. In view of Lemma \ref{l:lprop}
the estimates on these compositions can be reduced to estimates 
on compositions of operators of type \eqref{eq:lfio}. The next
proposition is similar to the results of \cite[Section~3]{AnNo}.

We take a sequence of symplectic maps $(\kappa_i)_{i=1,\ldots,J}$ defined in some
open neighbourhood $V\subset T^*\RR^n$ of the origin, which
satisfy (\ref{eq:surj}). Now the $\kappa_i$ do not necessarily leave the
origin invariant, but we assume that $\kappa_i(0,0)\subset V$ for all $i$.
We then consider
operators $(T_i)_{i=1,\ldots,J} $ which
quantize $\kappa_i$ in the sense of \eqref{eq:lfio} and are {\em microlocally unitary} near 
an open set $U\Subset V$ containing $(0,0)$. Let $\Omega\subset \RR^n$ be an open
set, such that $U\Subset T^*\Omega$, and for all $i$, $\kappa_i(U)\Subset T^*\Omega$.

For each $i$ we take a smooth cutoff function $ \chi_i \in \CIc ( U ; [0,1] )$, 
and let 
\be\label{e:S_j}
S_i \defeq \chi_i^w ( x, h D) \circ T_i\, .
\ee
We now consider a family of Lagrangian manifolds 
$\Lambda_k= \set{(x,\varphi'_k(x))\,;\,x\in\Omega}\subset T^*\RR^n$, $ k=0,\ldots, N $
sufficiently close to the ``position plane'' $\{\xi=0\}$:
\be\label{eq:varph} 
 | \varphi_k'|< \epsilon \,, \quad 
| \partial^\alpha \varphi_k | \leq C_\alpha\,, \quad  0 \leq k \leq N \,, \quad 
\alpha \in \NN^n \,.
\ee
Furthermore, we assume that these manifolds are locally mapped to one another by
the $\kappa_i$'s:
there exists a sequence of
integers $ (i_k \in [ 1 , J ])_{k=1,\ldots,N}$ such that
\be\label{eq:kap}
\kappa_{ i_{k+1} } ( \Lambda_k \cap U ) \subset  \Lambda_{k+1}\,, 
\quad  k =0,\ldots,N-1 \,.  
\ee
We want to propagate an initial
Lagrangian state $a(x)\,e^{i\varphi_0(x)/h}$, $a\in C^\infty_c(\Omega)$ through
the sequence of operators $S_{i_k}$, $k=1,\ldots,N$.

At each step, the action of
$\kappa_{i_k}^{-1}\rest_{\Lambda_k} $ can be projected 
on the position plane, to give a map $g_k$ defined on $\pi\kappa_{i_k}(U)\subset \Omega$~:
\be
\label{eq:gkx}
g_k(x)= \pi \circ \kappa^{-1}_{i_k} (x,\varphi'_k(x))\,.
\ee
For each $x=x^N\in\Omega$, we define iteratively
$x^{k-1}=g_k(x^k)$, $k=N,\ldots,1$: this procedure is possible as long as
each $x^k$ lies in the domain of definition of $g_k$.
Let us state our crucial dynamical assumptions: we assume that for all such sequences 
$(x^N,\ldots,x^0)$,
the Jacobian matrices, ${\partial x^k}/{\partial x^l}$, are uniformly bounded
from above:
\be\label{e:contraction}
\Big\|\frac{\partial x^k}{\partial x^l}\Big\|= 
\Big\|\frac{\partial (g^{k+1}\circ g^{k+2}\circ \cdots\circ g^{l})}{\partial x^l}(x^l)\Big\|\leq C_D,\,\qquad
0\leq k<l\leq N\,,
\ee
where $C_D$ is independent of $N$.
This assumption roughly
means that the maps $g_k$ are (weakly) contracting.

We will also use the notations
\be\label{e:D_k}
D_k \defeq \sup_{ x \in \Omega } |\det  d g_k(x)|^{\frac12}\,,\qquad 
J_k \defeq \prod_{k'=1}^k D_{k'}\,, 
\ee
and assume that the $D_k$ are uniformly bounded: 
$1/C_D\leq D_k\leq C_D$.

We can now state the main propagation estimate of this section which 
describes an $N$-iteration of Lemma~\ref{l:0}.
\begin{prop}\label{p:an1}
We use the above definitions and assumptions, 
and take $N$ arbitray large, possibly varying with $h$.

Take any $a\in C^\infty_c(\Omega)$, and consider the Lagrangian state 
$u=a\, e^{ i \varphi_0 / h }$ associated with the Lagrangian $\Lambda_0$.

Then we may write
\be\label{eq:lit}
( S_{i_N } \circ \cdots \circ S_{i_1} ) \big( a\, e^{ i \varphi_0 / h } \big) (x)
= e^{ i \varphi_{N} ( x) /h }\, \big(\sum_{j=0}^{L-1} h^j\,a_j^{N}( x )  
+ h^L R_L^N ( x , h ) \big) \,,
\ee
where each $a_j^{N}\in C^\infty_c (\Omega)$ is independent of $h$, while 
$R_L^N\in C^\infty((0,1]_h, \cS(\RR^n))$.

If $x^N\in\Omega$ defines a sequence (see \eqref{eq:gkx})
$x^{k-1}=g_k(x^k)$, $k=N,\ldots,1$, 
then
\be\label{e:principal}
| a_0^{N} ( x^N ) | = 
\Big( \prod_{k=1}^{N} \chi_{i_k}  ( x^k , \varphi_k' ( x^k ) )\; 
\,|\det  d g_k (x^k)|^{\frac12} \Big)\; | a ( x^0 ) | \,,
\ee
otherwise $a_j^{N}(x^N)=0$, $j=0,\ldots,L-1$. Also, we have the bounds
\begin{align}
\label{e:bounds}
\| a_j^{N} \|_{ C^\ell ( \Omega ) }& \leq C_{j,\ell}\, J_N\, (N+1)^{\ell + 3 j }\, 
\| a \|_{ C^{\ell + 2 j} (\Omega)}  ,\qquad j=0,\ldots,L-1,\quad\ell\in\NN
\,,\\ 
\label{e:bounds-rem}
\| R_L^N \|_{L^2(\RR^n)}& \leq C_L\,\| a \|_{C^{2L+n}(\Omega)} \, 
( 1 + C_0 h)^N \sum_{k=1}^{N} J_k\, k^{3L+n}  
\,.
\end{align}
The constants $ C_{ j, \ell },\,C_0,\,C_L$ depend on
constants in \eqref{eq:varph} and on the operators $ (S_i)_{i= 1,\ldots, J}$. 
\end{prop}
A crucial point in the above proposition is the 
explicit dependence on $ N $.
\begin{proof}
The proof of the proposition consists by iterating the results of Lemma~\ref{l:0},
keeping track of the bounds on the symbols and remainders.

For each $i$, the operator $S_i=\chi_i^w\, T_i$ can also be written in the form
\eqref{eq:lfio}, up to an error $\cO_{L^2\to L^2}(h^\infty)$, 
with the symbol $\alpha^i( x^1, \xi^0 ; h )$ replaced by 
$\beta^i( x^1, \xi^0 ; h )$ of compact support, and principal symbol
$\beta^i_{0}( x^1, \xi^0)=\chi_i(x^1,\psi'_{i\,x}(x^1,\xi^0))\,\alpha^i_{0}( x^1, \xi^0)$.
From the unitarity of $T_i$, $\alpha^i_{0}$ satisfies 
\eqref{e:unitarity} near $U$; as a result, when applying $S_i$ to
a Lagrangian state as in Lemma~\ref{l:0}, the first ratio in \eqref{eq:l1300}
should be replaced by $\chi_i(x^1,\xi^1)$.

To abbreviate the formulas, we set
$$
f_k(x)\defeq e^{i\beta_k/h} e^{ i \gamma_k ( x) } \,
\chi_{i_k}(x,\varphi'_k(x))\,|\det  d g_k (x)|^{\frac12},\qquad k=1,\ldots,N\,,
$$
where using unitarity \eqref{e:unitarity}, 
$$ \exp (i \gamma_k ( x ) ) = \frac{
\alpha^{i_k}_{0} ( x ,  \varphi_{k-1}'(g_k(x))  ) }
{ |\det \psi_{i_k\,x \xi}'' ( x , \varphi_{k-1}'(g_k(x)) ) |^{\frac12}}\,. $$
We will also use the short notation
$$
a_{j, \ell}^N \defeq \| a_j^N \|_{ C^{\ell } ( \Omega ) }\,,\qquad j=0,\ldots,L-1,\quad \ell\in\NN\,.
$$
We first analyze the principal symbol $a^N_0(x)$.
The formula \eqref{eq:l1300} and the definition of $f_k $ give
\be\label{e:an0}
a^N_0(x^N)= f_N(x^N)\,a^{N-1}_0(x^{N-1})\,,
\ee
which by iteration yields \eqref{e:principal}.
From $\|f_k\|_{C^0}\leq D_k$ the recursive relation \eqref{e:an0} also
implies the bound 
$a^N_{0,0}\leq J_N\,\norm{a}_{C^0}$.

To estimate higher $ C^\ell$ norms we 
differentiate \eqref{e:an0} with respect to $x^N$:
$$
\frac{\partial a^N_0}{\partial x^N}=f_N(x^N)\,
\frac{\partial x^{N-1}}{\partial x^N}\,\frac{\partial a^{N-1}_0}{\partial x^{N-1}}
+\frac{\partial f_N}{\partial x^N}\,a^{N-1}_0(x^{N-1})
$$
(to simplify the notation we omit the subscripts corresponding
to the coordinates in $x^N=(x^N_1,\ldots,x^N_n)$). 
Since we already control $a^{N-1}_{0,0}$,
and the norms $\norm{f_N}_{C^1}$ are
bounded uniformly in $N$, the above expression
can be schematically written as
$$
\frac{\partial a^N_0}{\partial x^N}=f_N\,
\frac{\partial x^{N-1}}{\partial x^N}\,\frac{\partial a^{N-1}_0}{\partial x^{N-1}}
+\cO(J_{N-1}\,\norm{a}_{C^0})\,,
$$
with an implied constant independent of $N$.
Applying this equality iteratively to $\partial a^{k}_0/\partial x^{k}$ down to $k=0$, 
we obtain
\begin{multline*}
\frac{\partial a^N_0}{\partial x^N}=
f_N f_{N-1}\cdots f_1\,\frac{\partial x^{0}}{\partial x^N} \frac{\partial a^0_0}{\partial x^0}\,+\\
+\cO\left( J_{N-1}
+f_N \frac{\partial x^{N-1}}{\partial x^N} J_{N-2}+
+f_N f_{N-1}  \frac{\partial x^{N-2}}{\partial x^N} J_{N-3}
+\ldots
+f_N f_{N-1}\cdots f_2 \frac{\partial x^{1}}{\partial x^N}  \right)\norm{a}_{C^0}\,.
\end{multline*}
Notice that $a^0_0=a$. Using the uniform bounds for the
Jacobian matrices $\frac{\partial x^{k}}{\partial x^N}$ and for the $D_k$,
this expression leads to
$$
a^N_{0,1}\leq C\,J_N\,\norm{a}_{C^1} + C\norm{a}_{C^0} \sum_{k=1}^N \frac{J_N}{D_k}
\leq C_{0,1}\,J_N\,(N+1)\,\norm{a}_{C^1}\,.
$$
The same procedure can be applied to
 higher derivatives of $a^N_0$: since $\norm{f_{N}}_{C^{\ell}}$
is uniformly bounded, the chain rule shows that
the $\ell$th derivatives of \eqref{e:an0} can be written
$$
\frac{\partial^\ell a^N_0}{(\partial x^N)^\ell}= f_N(x^N)\,
\Big(\frac{\partial x^{N-1}}{\partial x^N}\Big)^\ell\,
\frac{\partial^\ell a^{N-1}(x^{N-1})}{(\partial x^{N-1})^\ell}
+\cO(a^{N-1}_{0,\ell-1})\,.
$$
Assume we have proven the bounds \eqref{e:bounds} for the $a^{k}_{0,\ell-1}$, $k=0,\ldots,N$.
Iterating the above equality from $k=N-1$ down to $k=0$ yields the following estimate for 
$\partial^\ell a^N_0/(\partial x^N)^\ell$
\begin{multline}\label{e:aN0ell}
\frac{\partial^\ell a^N_0}{(\partial x^N)^\ell}=
f_N f_{N-1}\cdots f_1\Big(\frac{\partial x^{0}}{\partial x^N}\Big)^\ell 
\frac{\partial^\ell a^0_0}{(\partial x^0)^\ell}
+\cO\Big( J_{N-1}\,N^{\ell-1}
+f_N \Big(\frac{\partial x^{N-1}}{\partial x^N}\Big)^\ell J_{N-2}\,(N-1)^{\ell-1}+\\
+f_N f_{N-1} \Big(\frac{\partial x^{N-2}}{\partial x^N}\Big)^\ell J_{N-3}\,(N-2)^{\ell-1}
+\ldots
+f_N f_{N-1}\cdots f_2 \Big(\frac{\partial x^{1}}{\partial x^N}\Big)^\ell \Big)\norm{a}_{C^{\ell-1}}
\,.
\end{multline}
Using the uniform bounds \eqref{e:contraction} for 
$\frac{\partial x^{k}}{\partial x^N}$ and $D_k$, we get
$$
a^N_{0,\ell}\leq C_\ell\,J_N\,\norm{a}_{C^{\ell}}+ 
C\,\norm{a}_{C^{\ell-1}}\sum_{k=1}^N \frac{J_N}{D_k}\,k^{\ell-1}\leq 
C_{0,\ell}\,J_N\,(N+1)^\ell\,\norm{a}_{C^{\ell}}\,.
$$

We can now deal with higher-order coefficients
 $a^N_j$, by double induction on $j$ and $ N$. Above we
have proved the bounds for $j=0$ and all $ N $. 
Assume now 
that, for some $j\geq 1$, we have proved the bounds \eqref{e:bounds} 
for 
$a^N_{j',\ell}$ for all $j'<j$, $\ell\geq 0$ and all $N\geq 1$.
By induction on $ N $ we will prove the bounds for that $ j $ 
and all $ N$.

Applying Lemma~\ref{l:0} term by term to
$ a^{N-1} \defeq \sum_{j=0}^{L-1} h^j a^{N-1}_j+h^L R^{N-1}_L\,, $
we see that each component $a^{N}_j$
depends on the components $a^{N-1}_{j'}$, $0\leq j'\leq j$, and not on $R^{N-1}_L$. 
More precisely, 
from \eqref{eq:b_j} we get
\be\label{e:aNj0}
\begin{split}
a^{N}_j (x^{N})&= \sum_{j'=0}^j L^{N}_{j'} \big(\beta^{i_{N}} \, a^{N-1}_{j-j'} \big)(x^{N})\\
&= f_{N}(x^{N})\,a^{N-1}_{j}(x^{N-1}) + 
\sum_{j'=1}^j \sum_{|\gamma|\leq 2j'} \Gamma^{N}_{j'\gamma}(x^{N})\;
\partial^\gamma a^{N-1}_{j-j'}(x^{N-1})\,.
\end{split}
\ee
As explained in the proof of Lemma~\ref{l:0},
the functions $\Gamma^{N}_{j'\gamma}(x)$ can be expressed in terms of 
$\kappa_{i_N}$ and $\varphi_{N-1}$. From the assumptions on the latter, 
the norms, $\norm{\Gamma^N_{j\gamma}}_{C^\ell}$,
 are bounded uniformly with respect to $N$,
so \eqref{e:aNj0} implies the following upper bound:
\begin{align}
a^{N}_{j,0} &\leq D_{N}\, a_{j, 0}^{N-1} + C\, \sum_{j'=1}^j a^{N-1}_{j-j',2j'}\\
&\leq D_{N}\, a_{j, 0}^{N-1} + C\,J_{N-1}\, 
\sum_{j'=1}^j  N^{2j'+3(j-j')}\, \norm{a}_{C^{2j'}} \\
&\leq D_{N}\, a_{j, 0}^{N-1} + C\,j\,J_{N-1}\,N^{3j-1}\, \norm{a}_{C^{2j}}\,.
\end{align}
This inequality can be used in an induction with respect to $N$, starting from the trivial
$a^0_{j,0}= 0$. Assuming $a^{N-1}_{j,0}\leq C_{j,0}\,J_{N-1}\,N^{3j}\,\norm{a}_{C^{2j}}$ for 
some $C_{j,0}>0$, we obtain
\be\label{e:an1}
a^{N}_{j,0} \leq C_{j,0}\,J_{N}\big(N^{3j} + \frac{C\,j}{C_{j,0}\,D_{N}}\,N^{3j-1} \big) \norm{a}_{C^{2j}} \,.
\ee
The constant $C_{j,0}$ can be chosen large enough, so that the brackets are smaller
than\\ $N^{3j}+3jN^{3j-1}\leq (N+1)^{3j}$, which proves the induction
 step for $a^{N}_{j,0}$.

Once we have proved the bounds for the sup-norms of the symbols $a^N_j$, we can estimate
their derivatives by induction on $\ell$, as we did above for the principal symbol $a^N_0$.
Assume that we have proved the bounds \eqref{e:bounds} for 
all $a^{N}_{j,l}$, $N\geq 1$, $0\leq l\leq \ell-1$.
If we differentiate \eqref{e:aNj0} $\ell$ times with respect to $x^{N}$, we get 
$$
\frac{\partial^\ell a^N_j}{(\partial x^N)^\ell}
= f_{N}\,\Big(\frac{\partial x^{N-1}}{\partial x^{N}}\Big)^\ell\,
\frac{\partial^\ell a^{N-1}_j}{(\partial x^{N-1})^\ell}
+ \cO\big(a^{N-1}_{j,\ell-1}+ \sum_{j'=1}^j a^{N-1}_{j-j',\ell+2j'}\big)\,.
$$
where the implied constant depends on the bounds on $\norm{\Gamma^N_{j\gamma}}_{C^\ell}$.
Taking into account what we already know on $a^{N-1}_{j,\ell-1}$ and $a^{N-1}_{j-j',\ell+2j'}$, 
this takes the form
$$
\frac{\partial^\ell a^N_j}{(\partial x^N)^\ell}= f_{N}\,
\Big(\frac{\partial x^{N-1}}{\partial x^{N}}\Big)^\ell\,
\frac{\partial^\ell a^{N-1}_j}{(\partial x^{N-1})^\ell} + 
\cO\big(J_{N-1}\,N^{\ell+3j-1}\,\norm{a}_{C^{\ell+2j}}\big)\,.
$$
Applying iteratively this equality to $\partial^\ell a^{k}_j/(\partial x^{k})^\ell$ down to $k=1$ 
(as in \eqref{e:aN0ell}) and using $a^0_j(x)\equiv 0$, $ j > 0 $, we find:
\begin{multline}
\frac1{\norm{a}_{C^{\ell+2j}}}\,\frac{\partial^\ell a^{N}_j}{(\partial x^{k})^\ell} 
=\cO\Big(J_{N-1}\,N^{\ell+3j-1} + 
f_{N}\,\Big(\frac{\partial x^{N-1}}{\partial x^{N}}\Big)^\ell\,J_{N-2}\,(N-1)^{\ell+3j-1}\\
+f_{N}f_{N-1}\,\Big(\frac{\partial x^{N-2}}{\partial x^{N}}\Big)^\ell\, J_{N-3}\,(N-2)^{\ell+3j-1}\, 
+\ldots+
f_{N}f_{N-1}\cdots f_{1}\,\Big(\frac{\partial x^{1}}{\partial x^{N}}\Big)^\ell \Big)\,.
\end{multline}
From the uniform bound \eqref{e:contraction} and $\|f_k\|_{C^0}\leq D_k$,
this gives 
$$
a^{N}_{1,1}\leq C\,J_N \sum_{k=1}^N \frac{k^{\ell+3j-1}}{D_k}\leq C_{j,\ell}\,J_N\,N^{\ell+3j}\qquad
\text{for a certain }C_{1,\ell}>0\,.
$$ 
This proves the induction step $\ell-1\to \ell$, so that
we now have proved the bounds for $a^N_{j,\ell}$ for all $N\geq 1$, $\ell\geq 0$. This 
achieves to show the induction step on $j$, and \eqref{e:bounds}.

\bigskip

To estimate the remainder $ R_L^N ( x , h ) $ we define $ r_{k+1}^N ( x , h ) $
by 
\[ \begin{split}  S_{i_{k+1} } & \left( e^{ i \varphi_{k} / h } 
( a_0^{k} + h a_1^{k}  + \cdots + h^{L-1} a_{L-1}^{k} ) \right) \\
& = e^{ i \varphi_{k}  /h } \big( a_0^{k+1}  + h a_1^{k+1}  + 
\cdots + h^{L-1} a_{L-1}^{k+1} + h^L r_L^{k+1} ( \bullet, h ) \big) \,. 
\end{split}
\]
Due to the cutoff $\chi^w_i$, the remainder will be 
$\cO\big(({h}/({h+d(\bullet,\supp)})^\infty\big)$ outside 
$\pi\supp \chi_i$, so it is essentially supported inside $\Omega$.
On the other hand, from Lemma \ref{l:0} and the estimates \eqref{e:bounds}, we get
\[ \begin{split}
\| r_L^{k+1} ( \bullet, h ) \|_{ C^\ell(\RR^n)} & \leq C_{L,\ell} \sum_{j=0}^{L-1}
\| a^k_j \|_{C^{\ell+n+2(L-j)}} \leq
 C_{L,\ell}\, \sum_{j=0}^{L-1}J_k\, (k+1)^{j+\ell+n+2L} \|a\|_{ C^{\ell + n + 2 L }}\\
&  \leq C_{L,\ell}\,J_k\, (k+1)^{3L + \ell + n }\,  \|a\|_{ C^{\ell + n + 2 L }}
 \,. \end{split} \]
In particular,
\begin{equation}
\label{eq:rl2} \| r_L^{k+1} ( \bullet, h ) \|_{L^2(\RR^n)} \leq 
C_{L}\,J_k\, (k+1)^{3L+n} \| a \|_{C^{2L+n}} \,.
\end{equation}
The remainder $ R_N^L( x , h ) $ can now be written as 
\[ 
R^N_L  = r_N^L  + e^{-i \varphi_N /h }
\sum_{k=1}^{N-1} ( S_{i_{N}} \circ \cdots \circ S_{i_{k+1}} ) \big( r_L^{k}\, e^{i \varphi_k / h } \big)\,. 
\]
Since we assumed that $ T_j $'s  are microlocally unitary on the support
of $ \chi_j$'s, and that $ 0 \leq \chi_j \leq 1$, we have 
from the sharp G{\aa}rding inequality:
$$
\norm{ S_j }_{ L^2(\RR^n) \to L^2(\RR^n) } \leq 1 + C_0 h \,.
$$
The above formula for $ R_N^L $ and \eqref{eq:rl2}
give the estimate \eqref{e:bounds-rem}.
\end{proof}

\begin{rem} We can also obtain slightly weaker 
pointwise estimates on $ R_L^N $ in place of the $ L^2 $ estimates
of \eqref{eq:lit}. In fact, since $ \chi_j $'s are compactly 
supported we have 
\[  h^{\frac{n}2 + \ell}\, \norm{ R_L^N }_{C^\ell(\RR^n)} 
\leq C_\ell\, \| R_L^N \|_{ H^\ell_h} \leq C'_\ell\, \norm{ R_L^N }_{L^2(\RR^n)} \,, \]
and hence 
\[ \| R_L^N ( \bullet, h ) \|_{ C^{\ell}(\RR^n)} \leq C_{L\ell}\, h^{-\frac n 2 
- \ell }\, \| a \|_{C^{2L+n}}\,  ( 1 + C_0 h)^N 
\sum_{k=1}^{N} J_k\, k^{3L+n}  \,. \]
\end{rem}


\section{Classical Dynamics}\label{s:cl}
In this section we analyse the evolution of a family of Lagrangian leaves through the
classical flow. We will check that these Lagrangians (which remain in the vicinity of the
trapped set) stay ``under control'' uniformly
with respect to time. Eventually, this uniform control, which implies that 
the conditions \eqref{eq:varph} hold, 
will allow us to apply Proposition~\ref{p:inclination} in \S\ref{s:estimate}.

\subsection{Evolution of Lagrangian leaves}\label{s:elf}

\subsubsection{Poincar\'e sections and Poincar\'e maps}\label{s:Poincare}
\renewcommand{\eps}{\varepsilon}
We describe the construction of Poincar\'e sections and maps associated with the 
flow $\Phi^t$ on $\cE_E$ in the vicinity of $K_E$. This construction
will be used in the next section.

Take $\rho_0\in K_E$. We use an adapted coordinate chart $(y^0,\eta^0)$ centered at
$\rho_0\equiv (0,0)$ to parametrize 
the neighbourhood of $\rho_0$ in $T^*X$, with properties as 
described in Lemma~\ref{l:coord}. To keep in mind that 
$$
E^+_{\rho_0}=\Span\Big\{\frac{\partial}{\partial y_i}(0) \,,
\quad i=2,\ldots,n\Big\} \,,
$$ 
(and similarly
for $E^-_{\rho_0}$), we keep the ``time'' and ``energy'' coordinates $y^0_1,\,\eta^0_1$,
but rename the transversal coordinates as
$$
u^0_j\defi y^0_{j+1},\quad s^0_j\defi \eta^0_{j+1},\qquad j=1,\ldots,n-1\,.
$$
For any small
$\eps>0$ and using the Euclidean disk $D_\eps=\{u\in\RR^{n-1},\ |u| < \eps\}$, 
we define a neighbourhood of $\rho_0$ as the polydisk
\be
U_0(\eps)\equiv \set{|y^0_1| < \eps,\ |\eta^0_1| < \delta,\ u^0 \in D_\eps,\ 
s^0 \in D_\eps}\,.
\ee
Here $\delta>0$ corresponds to an energy interval where the dynamics remains uniformly hyperbolic,
as mentioned in \eqref{e:struct-stab}.
The intersection $U_0(\eps)\cap\cE_E$ is obtained by imposing the condition $\eta^0_1=0$,
and a {\em Poincar\'e section} $\Sigma_0=\Sigma_0(\eps)$ transversal to the flow
is obtained by imposing both $\eta^0_1=0$ and $y^0_1=0$.
The chart $(u^{0}, s^0)$ on $\Sigma_0$ 
is symplectic with respect to the induced symplectic structure on $\Sigma_0$.

Let us assume that the point $\Phi^1(\rho_0)$ belongs to a polydisk
$U_1(\eps)$ constructed similarly around a certain point $\rho_1\in K_E$, using an adapted
chart $(y^1,\eta^1)$. As a result,  
the Poincar\'e section $\Sigma_1\equiv \{y^1_1=\eta^1_1=0\}$ will intersect
the trajectory $(\Phi^s(\rho_0))_{|s-1|\leq\eps}$ at a single point, which we call $\rho'_0$.
The {\em Poincar\'e map}
$\kappa$ is defined, for $\rho\in \Sigma_0(\eps)$ near $\rho_0$, by taking the 
intersection of the trajectory $(\Phi^t(\rho))_{|t-1|\leq\eps}$ 
with the section $\Sigma_1$ (this intersection consists of at most one point). 
This map is automatically {\em symplectic}. 
In general, the 
strong (un)stable spaces $E^{\pm}_{\rho'_0}$ are not exactly tangent to $\Sigma_1$, but
close to it: they form ``angles'' $\cO(\eps)$ with the intersections,
$$
\tilde E^{\pm}_{\rho'_0}\defeq E^{\pm 0}_{\rho'_0}\cap T_{\rho'_0}\Sigma_1\,. 
$$
Furthermore, since the (un)stable subspaces $E^{\pm}_\rho$ are {\em H\"older continuous} 
with respect to 
$\rho\in K_E^\delta$, with some H\"older
exponent $\gamma>0$, and $d(\rho'_0,\rho_1)\leq\eps$, the subspaces 
$E^{\pm}_{\rho'_0}$ form ``angles'' $\cO(\eps^\gamma)$ with $E^{\pm}_{\rho_1}$.
The tangent map $d\kappa(\rho_0)$ maps $E^{\pm}_{\rho_0}$ to $\tilde E^{\pm}_{\rho'_0}$.
Hence, using the coordinate frames $\{(u^0,s^0)\}$ on $\Sigma_0$ (centered at $\rho_0$) 
and $\{(u^1,s^1)\}$ on $\Sigma_1$ (centered at $\rho_1$), the symplectic matrix
representing $d\kappa(\rho_0)$ can be written in the form
\be\label{e:Dkappa2}
d\kappa(\rho_0)\equiv \begin{pmatrix}A&0\\0& ^t\!A^{-1}\end{pmatrix}
+\eps^\gamma\begin{pmatrix}0&*\\ *& *\end{pmatrix}\,,
\ee
where the second matrix on the right
has uniformly bounded entries.
From the assumptions \eqref{eq:aa} on hyperbolicity, for $\eps$ small enough
there exists 
\be
\nu=e^{-\lambda}+\cO(\eps^\gamma)<1
\ee
such that the matrix $A$ satisfies
\be
\norm{A^{-1}}\leq \nu\quad\text{and}\quad\norm{^t\!A^{-1}}\leq \nu\,,
\ee
where $\norm{^t\!A^{-1}}$ is computed using the norms on $T_{\rho_0}\Sigma_0$,
$T_{\rho_1}\Sigma_1$ induced by the adapted metric $g_{\rm{ad}}$ (see \S\ref{defhyp}).
By extension, in the neighbourhood $V\subset\Sigma_0$ where it is defined, 
$\kappa$ takes the following form in the coordinates $(u^0,s^0)\mapsto (u^1,s^1)$:
\be\label{e:kappa}
\kappa(u^0,s^0)=(u^1,s^1)(\rho'_0) + 
\big(A u^0 + \alpha(u^0,s^0), \,^t\!A^{-1} s^0 + \beta(u^0,s^0)\big),\quad (u^0,s^0)\in V\,,
\ee
and the smooth functions $\alpha$, $\beta$ satisfy
\be\label{e:alpha-beta}
\alpha(0,0)=\beta(0,0)=0,\qquad \norm{\alpha}_{C^1(V)}\leq C\eps^\gamma,
\quad \norm{\beta}_{C^1(V)}\leq C\eps^\gamma\,.
\ee

\subsubsection{Evolving Lagrangian leaves}\label{s:evolving}

Given $\eps>0$, one can choose a
finite set of points $(\rho_i\in K_E)_{i\in I}$, adapted charts $(y^i,\eta^i)$
centered on $\rho_i$, such that the polydisks
$U_i(\eps) \equiv  \set{|y^i_1| < \eps,\ |\eta^i_1| < \delta,\ u^{i}\in D_\eps,\ s^i\in D_\eps}$
form an open cover of $K^\delta_E$:
\be\label{e:covering}
K^\delta_E\subset \bigcup_{i\in I} U_i(\eps)\,.
\ee
For some index $i_0\in I$, let $ \Lambda=\Lambda^0_{\loc} \subset U_{i_0}(\eps)\cap \cE_E$ 
be a connected isoenergetic Lagrangian leaf\footnote{Here and below, a leaf is a contractible submanifold with
piecewise smooth boundary.}.
For any $t>0$ we call $\Lambda^t = \Phi^t ( \Lambda )$.

We consider a point $\rho_0\in \Lambda$, and assume that there exists an integer $N>0$ such
that, for at each integer
time $0\leq k\leq N$, 
the point $\rho_k=\Phi^k(\rho_0)$ belongs to the set $U_{i_k}(\eps)$ for some $i_k\in I$.
We then call $\Lambda^k_{\loc}$ the connected part of $(\cup_{|s|<\vareps}\Phi^s \Lambda^k)\cap U_{i_k}(\eps)$
containing $\rho_k$.

We may use the symplectic coordinate chart $(y^{i_k},\eta^{i_k})$ to represent $\Lambda^k_{\loc}$.
Being contained in a single energy shell $\cE_E$, the Lagrangian leaf $\Lambda^k_{\loc}$ 
is foliated by flow trajectories (bicharacteristics). It can be put
into the form 
\be\label{e:Lambda}
\Lambda = \cup_{|s| < \eps} \Phi^s (S^k)\,,
\ee
where $S^k=\Lambda^k_{\loc}\cap \Sigma_{i_k}$ is an $n-1$-dimensional
Lagrangian leaf in the symplectic
section $\Sigma_{i_k}(\eps)=U_{i_k}(\eps)\cap\{y^{i_k}_1=\eta^{i_k}_1=0\}$ (see Fig.~\ref{f:Poincare} for
a representation of the above objects).

We will be interested in Lagrangian leaves which are ``transversal enough'' 
to the stable subspace $E^{-}_{\rho_k}$,
and can therefore be represented by graphs of smooth functions in the adapted charts:
\be\label{e:graph}
\Lambda^k_{\loc}\equiv \set{ ( y^{i_k} ,  \eta^{i_k}  ) 
\; : \; \eta^{i_k} = F^k ( y^{i_k} ) }\,.
\ee
The intersection $S^k=\Lambda^k_{\loc}\cap \Sigma_{i_k}$ is then also given by a graph:
$$
S^k\equiv \set{(u^{i_k},s^{i_k}) \; : \; s^{i_k} = 
f^{k}(u^{i_k}) ,\ u^{i_k}\in D_\eps}\,,
$$
and \eqref{e:Lambda} implies that $F^k(y^{i_k})=(0,f^k(u^{i_k}))$, so that \eqref{e:graph}
takes the form
\be\label{e:reduction}
\Lambda^k_{\loc} \equiv \set{(y^{i_k}_1,u^{i_k}; 0, f^{k}(u^{i_k})),\ |y^{i_k}_1| < \eps,\ u^{i_k}\in D_\eps}\,.
\ee

\noindent{\bf Convention:}
In the rest of this section the norm $\norm{\cdot}$ applying to an object living on
$\Sigma_{i_k}\equiv D_\eps\times D_\eps$ corresponds to the Euclidean norm 
on $T_{\rho_{i_k}}\Sigma_{i_k}$
relative to the adapted metric $g_{\rm{ad}}(\rho_{i_k})$. 
The same convention applies to the norm $\norm{\cdot}$ of a linear
operator sending an object on $\Sigma_{i_k}$ to an object on $\Sigma_{i_{k+1}}$ (or vice-versa).

\medskip

The following proposition (similar to the Inclination Lemma of \cite[Proposition~6.2.23]{KaHa})
shows that, if $\eps$ has been chosen small enough and $\Lambda$ is
``transversal enough'' to the stable manifolds (that is, in some ``unstable cone''), 
then the local Lagrangian leaves
$\Lambda^k_{\loc}$ remain in the same unstable cone, uniformly with respect to $k=0,\ldots,N$. 
\begin{prop}
\label{p:inclination}
Fix some $\gamma_1>0$. Then there exists $\eps_{\gamma_1}>0$ such that, provided the diameter
$\eps\in (0,\eps_{\gamma_1})$, the following holds.

Suppose the Lagrangian $\Lambda=\Lambda^0_{\loc}\subset \cE_E \cap U_{i_0}(\eps)$ is
the graph of a smooth function $f^0$
in the adapted frame $(y^{i_0},\eta^{i_0})$, and is
contained in the {\em unstable $\gamma_1$-cone}:
$$
\Lambda^0_{\loc} \equiv \set{ ( y^{i_0}_1,u^{i_0} ;0, f^0 ( u^{i_0} ) ),
\ |y^{i_0}_1| < \eps,\ u^{i_0}\in D_\eps },\quad 
\text{with}\quad \sup_{u^{i_0}}\|df^0(u^{i_0})\|\leq\gamma_1\,.
$$

i) Then, for any $0\leq k\leq N$, the connected component 
$\Lambda^k_{\loc}\subset U_{i_k}(\eps)$ containing $\rho_k$ is also a graph in the frame
$(y^{i_k},\eta^{i_k})$, and is also contained in the
unstable $\gamma_1$-cone:
$$
\Lambda^k_{\loc}  \equiv \set{ ( y^{i_k}_1,u^{i_k};0, f^k ( u^{i_k} ) ),
\ |y^{i_k}_1| < \eps,\ u^{i_k}\in D_\eps},
\quad \text{with}\quad \sup_{u^{i_k}\in D_\eps}\|df^k ( u^{i_k} )\|\leq\gamma_1\,.
$$
$ii)$ For any integer $\ell\geq 2$ and $\tilde\gamma_\ell>0$, we may choose
$\eps>0$ small enough such that,
if $f^0$ is in the unstable $\gamma_1$-cone and
satisfies $\norm{f^0}_{C^\ell}\leq \tilde\gamma_\ell$, then
\be\label{eq:l21}
\forall k=0,\ldots,N,\qquad \norm{f^k}_{C^\ell(D_\eps)} \leq \gamma_\ell\,.
\ee
$iii)$ From the above properties, near $\rho_0$ the map $\Phi^N\rest_\Lambda$ can be projected
on the planes $\{\eta^{i_0}=\nolinebreak 0\}$ and $\{\eta^{i_N}=0\}$, inducing
a map $y^{i_0}\mapsto y^{i_N}$.

In the case where the sets $U_{i_k}(\eps)$ contain a trajectory in $K_E^\delta$,
(so these sets may be centered on $\rho_{i_k}=\Phi^k(\rho_{i_0})$), 
the projected map $y^{i_0}\mapsto y^{i_N}$ satisfies the following estimate on its
domain of definition:
$$
\det \Big( \frac{\partial y^{i_N}}{\partial y^{i_0}} \Big) = (1+\cO(\eps))\,
\exp\big(\lambda^+_{N}(\rho_{i_0}) \big) \,.
$$
Here $\lambda^+_{N}$ is the unstable Jacobian given in \eqref{e:unst-jac}. 
The crucial point is that the implied constant is independent of $N$.
\end{prop}
\begin{proof}
We follow the proof of the stable/unstable manifold
theorem for hyperbolic flows \cite[Thm 6.2.8 and Thm 17.4.3]{KaHa}.

For each $k=0,\ldots,N$, the Poincar\'e
section $\Sigma_{i_k}$ 
does generally not contain $\rho_k$, but it contains a unique 
iterate $\rho'_k=\Phi^s\rho_k$ for some $s\in (-\eps,\eps)$.
The Poincar\'e map $\kappa_k$ from $V_k\subset \Sigma_{i_k}(\eps)$ to $\Sigma_{i_{k+1}}(\eps)$
will satisfy $\kappa_k(\rho'_k)=\rho'_{k+1}$. 

Since $d(\rho_{i_k},\rho'_k)\leq\eps$ and $d(\rho_{i_{k+1}},\rho'_{k+1})\leq\eps$,
there exists $C>1$ such that the extended Poincar\'e map from
$\Sigma_{i_{k}}(\eps)$ to $\Sigma_{i_{k+1}}(C\eps)$ sends $\rho_{i_k}$ to a point 
$\rho'_{i_k}\in \Sigma_{i_{k+1}}(C\eps)$.
We are thus in the situation of \S\ref{s:Poincare}, with $\rho_0,\rho'_0,\rho_1$ being replaced by
$\rho_{i_k},\,\rho'_{i_k}\,,\rho_{i_{k+1}}$ (see Fig.~\ref{f:Poincare}).
In the charts
$(u^{i_k},s^{i_k})\mapsto (u^{i_{k+1}},s^{i_{k+1}})$,
the map $\kappa_k$ takes the form
\be\label{e:kappa-ik}
\kappa_k(u^{i_k},s^{i_k})=(u^{i_{k+1}},s^{i_{k+1}})(\rho'_{i_k}) +
\big(A_k u^{i_k} +\tilde\alpha_k(u^{i_k},s^{i_k}),\, ^t\!A_k^{-1} s^{i_k} + 
\tilde\beta_k(u^{i_k},s^{i_k})\big)\,,
\ee
where $\norm{A_k^{-1}},\ \norm{^t\!A_k^{-1}}\leq\nu$ 
and the smooth functions $\tilde\alpha_k$, $\tilde\beta_k$ satisfy
\eqref{e:alpha-beta}. 
It is convenient to shift the origin of the coordinates $(u^{i_k},s^{i_k})$ 
(respectively  $(u^{i_{k+1}},s^{i_{k+1}})$) 
such as to center them at $\rho'_k$ (respectively  at $\rho'_{k+1}$). We call the
shifted coordinates $(u^{k},s^{k})$ (respectively  $(u^{k+1},s^{k+1})$). In these shifted coordinates, 
we get
\be\label{e:kappa_k}
\kappa_k(u^k,s^k)=
\big(A_k u^k +\alpha_k(u^k,s^k),\, ^t\!A_k^{-1} s^k + \beta_k(u^k,s^k)\big)\,,\qquad
(u^k,s^k)\in V_k\,.
\ee
The shifted functions $\alpha_k$, $\beta_k$ still satisfy \eqref{e:alpha-beta},
where $V=V_k$ corresponds to the neighbourhood of $\rho'_k$ where $\kappa_k$ is defined.
\begin{figure}[ht]
\begin{center}
\includegraphics[width=16cm]{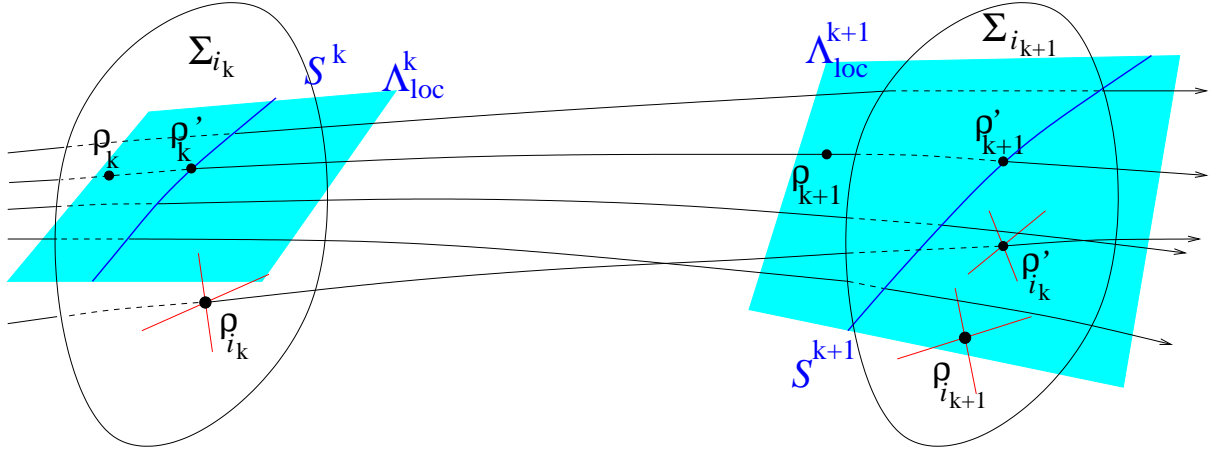}
\end{center}
\caption{Illustration of the objects appearing in 
the proof of Proposition~\ref{p:inclination}. The local Lagrangians $\Lambda^{k}_{\loc}$, 
$\Lambda^{k+1}_{\loc}$ appear in light blue, and are foliated by bicharacteristics.
The axes around $\rho_{i_k}$, $\rho_{i_{k+1}}$
represent the stable and unstable subspaces $E^{\pm}_\rho$ on those points. The axes around
$\rho'_{i_k}=\kappa_k(\rho_{i_k})$ are the projected subspaces $\tilde E^{\pm}_{\rho'}$.
\label{f:Poincare}
}\end{figure}

After fixing the coordinate charts, we can study the behaviour of the
intersections $S^k=\Lambda^k_{\loc}\cap \Sigma_k$ when $k$ grows. 
We are exactly in the framework of \cite[Theorem 6.2.8]{KaHa}, and we will use the same
method to control the $S^k$.

We first show that, if $\eps$ is chosen small enough, the 
unstable $\gamma_1$-cone in $S^k$ is sent by $\kappa_k$ inside the $\gamma_1$-cone in $S^{k+1}$.
Let us assume that 
$$ 
S^k= \{(u^k,f^k(u^k))\} \,, \qquad \sup_{u^k\in D_\eps}\norm{df^k}\leq\gamma_1\,. 
$$
The projection of $\kappa_k\rest_{S^k}$ on the horizontal subspace reads:
\be\label{e:induced}
u^k\mapsto u^{k+1}=\pi\kappa_k (u^k,f^k(u^k))= A_k u^k + \alpha_k(u^k,f^k(u^k))\,,
\ee
so by differentiation we get it is uniformly expanding from some neighbourhood $D'_\eps\subset D_\eps$ to $D_\eps$:
\be\label{e:partial-tg}
\frac{\partial u^{k+1}}{\partial u^k}= A_k + 
\frac{\partial\alpha_k}{\partial u^k}+
\frac{\partial\alpha_k}{\partial s^k}\frac{\partial f^k}{\partial u^k}
=A_k + \cO(\eps^\gamma(1+\gamma_1))\,.
\ee
The property $\norm{A_k^{-1}}\leq\nu<1$ shows that, for $\eps^\gamma(1+\gamma_1)$ 
small enough, this map is uniformly expanding.
Hence, this map is {\em invertible}, and its inverse, 
\be\label{e:g_k}
u^{k+1}\mapsto u^k\defeq \tg_{k+1}(u^{k+1})\,,
\ee
is uniformly contracting:
\be\label{e:bound10}
\norm{d\tg_{k+1}(u^{k+1})}=\Big\|\frac{\partial u^{k}}{\partial u^{k+1}}\Big\|\leq \nu_1\,,
\qquad u^{k+1}\in D_\eps\,,
\ee
with $\nu_1 = \nu+C_{\alpha}(\eps^\gamma(1+\gamma'_1))<1$.
As a result, since $\tg_{k+1}(0)=0$, we have
$$
\norm{u^k}=\norm{\tg_{k+1}(u^{k+1})}\leq \nu_1\,\norm{u^{k+1}},\qquad u^{k+1}\in D_\eps\,.
$$
We also see that the intersection $S^{k+1}=\kappa_k(S^k)$ can be represented as the graph
$S^{k+1}=\set{(u^{k+1},f^{k+1}(u^{k+1})),\ u^{k+1}\in D_\eps}$ in the coordinates centered
at $\rho'_{k+1}$, with the explicit expression
\be\label{e:ft}
f^{k+1}(u^{k+1})=\,^t\!A_k^{-1} f^k(u^k) + \beta_k(u^k,f^k(u^k)),\quad u^{k+1}\in D_\eps, \quad
u^k=\tg_{k+1}(u^{k+1})\,.
\ee
Differentiating this expression with respect to $u^{k+1}$ leads to
$$
\frac{\partial f^{k+1}}{\partial u^{k+1}}=\Big(\frac{\partial u^k}{\partial u^{k+1}}\Big)\,
\Big[\big(\,^t\!A_k^{-1}+\partial_{s}\beta^k(u^k,f^k(u^k))\big)\,\partial f^k (u^k)+
\partial_{u}\beta^k(u^k,f^k(u^k))\Big]\,.
$$
Since for $\eps$ small enough we have uniformly
$$
\norm{\,^t\!A_k^{-1}+\partial_{s}\beta^k(u^k,f^k(u^k))}\leq \nu_2,\qquad \nu_2=\nu + C_\beta\eps^\gamma<1\,,
$$
the above Jacobian is bounded from above by
$$
\Big\|\frac{\partial f^{k+1}}{\partial u^{k+1}}\Big\|\leq 
\nu_1 \,\big( \nu_2 \gamma_1 + C\eps^\gamma \big)\,.
$$
If $\eps>0$ is small enough, 
the above right hand side is smaller than $\nu_2\gamma_1$.
We have thus proved that the $\gamma_1$-unstable cones in $\Sigma_k$ are invariant through 
$\kappa_k$, which proves the statement $i)$ of the proposition.

\medskip

Let us now study the higher derivatives
of the functions $f^k$, obtained by further differentiating \eqref{e:ft}.
We use the norms 
$$
\norm{f}_{C^\ell(V)}=\max_{\alpha\in \NN^{n-1}, |\alpha|\leq \ell}\; 
\sup_{u\in V} \|\partial^\alpha f(u)\|\,,
$$
and will proceed by induction of the degree $\ell$ of differentiation. 
Let us assume that
for some $\ell\geq 2$, there exists $\gamma_{\ell-1}$ such that
all functions $f^k$, $0\leq k\leq N$, satisfy $\norm{f^k}_{C^{\ell-1}}\leq \gamma_{\ell-1}$.
Above we have proved this property for $\ell=2$.
By differentiating $\ell$ times \eqref{e:ft}, we get
\begin{align*}
\frac{\partial^\ell f^{k+1}}{(\partial u^{k+1})^\ell} &= 
\Big(\frac{\partial u^k}{\partial u^{k+1}}\Big)^\ell\, 
\big(\,^t\!A_k^{-1}+\partial_s\beta^k\big)\,\frac{\partial^\ell f^k}{(\partial u^{k})^\ell} +
P_{\ell,k}(\partial f^k,\ldots,\partial^{\ell-1}f^k )\,,\nonumber\\
\Longrightarrow\ \Big\|\frac{\partial^\ell f^{k+1}}{(\partial u^{k+1})^\ell}\Big\|&
\leq  \nu_1^\ell \,\nu_2\,\Big\|\frac{\partial^\ell f^{k}}{(\partial u^{k})^\ell}\Big\|
+\|P_{\ell,k}(\partial f^k,\ldots,\partial^{\ell-1}f^k)\|\,.
\end{align*}
Here $P_{\ell,k}$ is a polynomial of degree $\ell$, 
with coefficients uniformly bounded with respect to $k$ and $u^k\in D_\eps$. Using the assumption
$\norm{f^k}_{C^{\ell-1}}\leq \gamma_{\ell-1}$, there exists $C_{\ell-1}(\gamma_{\ell-1})>0$ 
such that the following inequality holds:
$$
\Big\|\frac{\partial^\ell f^{k+1}}{(\partial u^{k+1})^\ell}\Big\|
\leq \nu_1^\ell\,\nu_2\,
\Big\|\frac{\partial^\ell f^{k}}{(\partial u^{k})^\ell}\Big\| + 
C_{\ell-1}(\gamma_{\ell-1})\,.
$$
If we now choose $\gamma_\ell>0$ such that
$$ 
\gamma_\ell>\max\Big(\frac{C_{\ell-1}(\gamma_{\ell-1})}{\nu_2(1-\nu_1^\ell)},\, 
\gamma_{\ell-1},\,\norm{f^0}_{C^\ell}\Big)\,,
$$
we check that the condition
$\norm{f^{k}}_{C^\ell}\leq \gamma_\ell$ implies that 
$\norm{\partial^\ell f^{k+1}}_{C^0}\leq \nu_2\,\gamma_\ell$. Hence, all functions
$f^k$, $0\leq k\leq N$, satisfy $\norm{f^k}_{C^{\ell}}\leq \gamma_{\ell}$, which 
proves statement $ii)$.

\medskip

The important point in $iii)$ is the uniformity of the estimate with respect to $N$. 
To prove such a uniform estimate, one needs to analyze the trajectory 
$(\rho'_k)_{k=0,\ldots,N}$ with respect to the ``reference trajectory'' $(\rho_{i_k})$.
It is useful to replace the coordinates $(u^{i_k},\eta^{i_k})$ on $\Sigma^{i_k}$
by coordinates $(\tu^{k},\ts^{k})$ with the following properties.
We define
the local (un)stable
manifolds on the Poincar\'e sections:
$$ 
W^{\pm}_k\defeq W^{0\pm}_{\loc}(\rho_{i_k})\cap\Sigma^{i_k}\,.
$$
The new coordinates $(\tu^{k},\ts^{k})$ satisfy:
\begin{gather*}
W^+_k\equiv \{(\tu^{k},0),\ \tu^k\in D_\eps\},\qquad W^-_k\equiv \{(0,\ts^{k}),\ \ts^k\in D_\eps\},\\
(\tu^k,\ts^k) = (u^{i_k},s^{i_k}) + \cO((u^{i_k},s^{i_k})^2)\quad \text{near the origin,}
\end{gather*}
and they need not be symplectic.
In these coordinates, the Poincar\'e
map $\kappa_k:\Sigma^{i_k}\to\Sigma^{i_{k+1}}$ has a more precise form than in \eqref{e:kappa_k}: we
can still write it as
$$
\kappa_k(\tu^k,\ts^k)=
\big(A_k \tu^k +\alpha_k(\tu^k,\ts^k),\, ^t\!A_k^{-1} \ts^k + \beta_k(\tu^k,\ts^k)\big)\,,
$$
but the smooth functions $\alpha_k$, $\beta_k$ satisfy more constraints than before:
$$
\alpha_k(0,\ts^k)=\beta_k(\tu^k,0)\equiv 0\quad\text{and}\quad d\talpha_k(0,0)= d\tbeta_k(0,0)=0\,.
$$
This shows that, near the origin, $\alpha_k(\tu^k,\ts^k)=\cO\big(\norm{\tu^k} + \norm{\ts^k}\big)\norm{u^k}$
and similarly for $\beta_k$.
Using these coordinates, we can show that most of the points along the trajectory $(\rho'_k)$ are
very close to the reference points $(\rho_{i_k})$. If we call $(\tu^k,\ts^k)$ the coordinates 
of $\rho'_k\in S^k$, we have
\begin{align*}
\tu^{k+1}&= A_k \tu^k + \alpha_k(\tu^k,\ts^k)= A_k\tu^k + \cO(\eps\,\norm{\tu^k})\,,\\
\ts^{k+1}&=\,^t\!A_k^{-1} \ts^{k} + \beta_k(\tu^k,\ts^k)= \,^t\!A_k^{-1} \ts^k+\cO(\eps\,\norm{\ts^k})\,.
\end{align*}
Taking into account the fact that $\norm{\tu^N}\leq C\eps$ and $\norm{\ts^0}\leq C\eps$, 
for $\eps$ small enough there exists $\nu_3 = \nu+\cO(\eps)<1$ such that
$$
\norm{\tu^k}\leq C\,\eps\,\nu_3^{N-k}\,,\qquad
\norm{\ts^k}\leq C\,\eps\,\nu_3^k\,,\qquad k=0,\ldots,N\,.
$$
These estimates prove that, if $N$ is long, the points $\rho'_k$ for $k\gg 1$, $N-k\gg 1$ are close to
$\rho_{i_k}$. The tangent of the map $\tu^k\mapsto \tu^{k+1}$ induced by projecting
$\kappa_k\rest_{S^k}$ on the planes $\{\ts=0\}$ is given as in \eqref{e:partial-tg} by
$$
\frac{\partial \tu^{k+1}}{\partial \tu^k}=A_k + 
\frac{\partial\alpha_k}{\partial \tu^k}+
\frac{\partial\alpha_k}{\partial \ts^k}\frac{\partial f^k}{\partial \tu^k}
=A_k +\cO(\norm{\tu^k} + \norm{\ts^k})\,.
$$
To obtain the last equality we used the fact that $\norm{d f^k}$ is uniformly 
bounded, as shown above.
The tangent of the map obtained by projecting $\kappa_{N-1}\circ\cdots\circ \kappa_0\rest_{S^0}$
on the planes $\{\ts=0\}$ then reads
\begin{align*}
\frac{\partial \tu^{N}}{\partial \tu^0}&=\prod_{k=0}^{N-1}\Big(A_k +\cO(\norm{\tu^k} + \norm{\ts^k})\Big)
= \prod_{k=0}^{N-1}\Big(A_k +\cO(\eps\,(\nu_3^{N-k} + \nu_3^k))\Big)\\
&=\big(\prod_{k=0}^{N-1} A_k\big)\, \prod_{k=0}^{N-1}\Big(I + \cO(\eps\,(\nu_3^{N-k} + \nu_3^k))\Big)
\end{align*}
The determinant of the last factor is of order $1+\cO(\eps)$, so we deduce
\be\label{e:deter}
\det\Big(\frac{\partial \tu^{N}}{\partial \tu^0}\Big)=(1+\cO(\eps))\,\det\big(\prod_{k=0}^{N-1} A_k\big)\,.
\ee
We then recall that the change of variables $(\tu^k,\ts^k)\mapsto (u^{i_k},s^{u_k})$ is close to
the identity:
$$
\Big(\frac{\partial (\tu^k,\ts^k)}{\partial (u^{i_k},s^{i_k})}\Big)=I+\cO(\eps).
$$
As a result, the estimate \eqref{e:deter} applies as well to the Jacobian of the map 
$\kappa_{N-1}\circ\cdots\circ \kappa_0\rest_{S^0}$, projected
in the planes $\{s^{i_0}=0\}$, $\{s^{i_N}=0\}$, which we denote by 
$\det(\partial u^{i_N}/\partial u^{i_0})$.

We now consider the map $y^{i_0}\mapsto y^{i_N}$
induced by projecting $\Phi^N\rest_{\Lambda^0}$
on the planes $\{\eta^{i_0}=0\}$, $\{\eta^{i_N}=0\}$. 
From the structure of the adapted coordinates, the tangent to this map has the form
$$
\Big(\frac{\partial y^{i_N}}{\partial y^{i_0}}\Big)=
\begin{pmatrix}1&*\\
0& \big({\partial u^{i_N}}/{\partial u^{i_0}} \big)\end{pmatrix}\,,
$$
so the estimate \eqref{e:deter} also applies to 
$\det({\partial y^{i_N}}/{\partial y^{i_0}})$.

Finally, we remark that if we take $\Lambda=W^{+0}_{\rho_{i_0}}$, 
then the tangent map at $\rho_0=\rho_{i_0}$ is given by
$$
\Big(\frac{\partial u^{i_N}}{\partial u^{i_0}}\Big)(0)=\prod_{k=0}^{N-1} A_k\,. 
$$
Hence in this case we find
$$
\det\big(\prod_{k=0}^{N-1} A_k\big) = \det(\frac{\partial y^{i_N}}{\partial y^{i_0}})
=\det \big(d\Phi^N\rest_{E^{+0}_{\rho_{i_0}}}\big)=\exp(\lambda^+_N(\rho_{i_0}))\,.
$$
For the second equality we have used \eqref{e:determ} and 
the fact that, for each $k$, the adapted coordinates satisfy 
$\Omega(\partial/\partial y^{i_k}_1\wedge \ldots\wedge \partial/\partial y^{i_k}_n)=1$ at the origin
(see Lemma~\ref{l:coord}).
\end{proof}

\begin{rem}\label{r:rem2}
Due to structural stability, the results of Proposition~\ref{p:inclination} 
apply to Lagrangian leaves $\Lambda\in \cE_{E'}$ transversal to the stable lamination,
for any energy
$E'\in (E-\delta,E+\delta)$, with the difference that the evolved
local Lagrangians are of the form
\be\label{e:Lag-E'}
\Lambda^k_{\loc}  \equiv \set{ ( y^{i_k}_1,u^{i_k};E'-E, f^k ( u^{i_k} ) ),
\ |y^{i_k}_1|\leq\eps,\ |u^{i_k}|\leq \eps},
\quad \text{with}\quad \|df^k ( u^{i_k} )\|\leq\gamma_1\,.
\ee
The Poincar\'e sections used in the proof are taken as $U_i(\eps)\cap\{y^i_1=0,\,\eta^i_1=E'-E\}$.
All constants can be taken to be independent of $E'\in (E-\delta,E+\delta)$.
\end{rem}
\begin{rem}\label{r:rem3}
Each $f^k:(D_\eps)_u\to \RR^{n-1}_s$ representing the Lagrangian $\Lambda^k_{\loc}$ 
of \eqref{e:Lag-E'}
can be written as 
$ f^k(u) = \phi'_k(u)$ for some function $\phi_k:(D_\eps)_u\to\RR$. Therefore, the function
$$
\varphi_k(y_1,u)\defeq \phi_k(u)+(E'-E)y_1,\qquad u\in D_{\vareps},\ |y_1|\leq\vareps\,,
$$ 
generates $\Lambda^k_{\loc}$ in the symplectic coordinates $(y^{i_k},\eta^{i_k})$.
\end{rem}

\subsection{An alternative definition of the topological pressure}\label{s:selection}
\renewcommand{\eps}{\epsilon}

To connect the resonance spectrum with the 
topological pressure \eqref{eq:defpr} of the
flow, we use an alternative definition of the pressure \cite[Theorem 5.2]{Pes88}, which will 
provide us with a convenient open cover of $K^\delta_E$. 

Taking $\delta>0$ small enough to satisfy \eqref{e:struct-stab}, consider 
$\cV=(V_b)_{b\in B}$ an open finite cover of $K^{\delta}_E$, made of 
sets of small diameters contained in the energy strip $\cE_E^{\delta}$.
For any integer $T>0$, the refined cover $\cV^{(T)}$ is made of the
sets
\be\label{e:V_beta}
V_{\beta}\defeq\bigcap_{k=0}^{T-1} \Phi^{-k} (V_{b_k})\,,\qquad 
\beta=b_0 b_2\ldots b_{T-1}\in B^T\,.
\ee
The $T$-strings $\beta$ such that $V_{\beta}\cap K_E^\delta\neq\emptyset$ make up a subset 
$\cB'_T\subset B^T$.
Below it is 
convenient to coarse-grain the unstable Jacobian \eqref{e:unst-jac} 
on subsets $W\subset \cE_E^\delta$: 
\be\label{e:coarse-jac}
\forall W\subset \cE_E^\delta,\quad W\cap K_E^\delta\neq\emptyset,\qquad S_{T}(W)\defeq 
-\inf_{\rho\in W \cap 
K^\delta_E } \lambda^+_{T}(\rho)\,.
\ee
We define the following quantity, 
similar to \eqref{e:partitionf}:
$$
Z_{T}(\cV,s)\defeq \inf\set{\sum_{\beta\in\cB_T} 
\exp\{s\,S_{T} (V_\beta)\}\;:\;\cB_T \subset\cB'_T,\ 
K^{\delta}_E\subset\bigcup_{\beta\in\cB_T} V_\beta}\,.
$$
The topological pressure of the flow on $K_E^{\delta}$ can then be obtained as follows:
$$
\cP^\delta_E(s)=\lim_{{\rm diam} \cV\to 0}\lim_{T\to\infty}\frac{1}{T}\,\log Z_{T}(\cV,s)\,.
$$
Here the covers $\cV$ are as above: they cover $K_E^{\delta}$ in the 
energy strip $\cE_E^{\delta}$.
Finally, because the pressure is continuous as a function of 
the energy, the pressure \eqref{eq:defpr}
can be obtained through the limit $\cP_E(s)=\lim_{\delta\to 0}\cP^\delta_E(s)$.

From now on we fix some small $\eps_0>0$. From the above limits, 
a cover $\cV_0$ of $K_E^\delta$ inside $\cE_E^\delta$ (of arbitrary small 
diameter $\vareps>0$) and an integer $t_0>0$ depending on $\cV_0$, such that
\be\label{e:pressure-approx}
\Big|\frac{1}{t_0}\,\log Z_{t_0}(\cV,s)-\cP^\delta_E(s)\Big|\leq \eps_0\,.
\ee
As a consequence, there exists a subset $\cB_{t_0}\subset \cB'_{t_0}$, such that
$\set{V_\beta,\ \beta\in \cB_{t_0}}$ is an open cover of $K^\delta_E$ inside 
$\cE_E^{\delta}$, which satisfies
$$
\sum_{\beta\in \cB_{t_0}} \exp\{s\,S_{t_0} (V_\beta)\}\leq \exp\big\{t_0(\cP^\delta_E(s)+\eps_0)\big\}\,.
$$
We rename the family $\set{V_\beta,\ \beta\in \cB_{t_0}}$ as
$\set{W_a,\ a\in A_1}$, so the above bound reads
\be\label{eq:press2}
\sum_{a\in A_1} \exp\{s\,S_{t_0} (W_a)\}\leq \exp\big\{t_0(\cP^\delta_E(s)+\eps_0)\big\}\,.
\ee
Each set $W_a$ contains at least one point $\rho_a\in K_E^\delta$, which we may
set as reference point: following Lemma~\ref{l:coord}, we can represent $W_a$ by 
an adapted chart $(y^a,\eta^a)$ centered at $\rho_a$. Similarly, we can
also equip any $V_b\in\cV_0$ with adapted charts $(y^b,\eta^b)$ centered at some point 
$\rho_b\in V_b\cap K_E^\delta$.

Each point $\rho\in W_a=V_\beta$ evolves such that 
$\Phi^{k}(\rho)\in V_{b_k}$ for all $k=0,\ldots,t_0-1$. Therefore, as long as $\vareps$ has
been chosen small enough, we are in position to apply
Proposition~\ref{p:inclination} and Remarks~\ref{r:rem2},\ref{r:rem3} to isoenergetic 
$\gamma_1$-unstable Lagrangian leaves in $W_a$.

\begin{prop}\label{c:fold}
Take any energy $E'\in (E-\delta,E+\delta)$ and any index $a\in A_1$.
Assume that $\Lambda\subset\cE_{E'}\cap W_a$ is a Lagrangian leaf generated in the
chart $(y^a,\eta^a)$ by a function $\varphi$ defined on a 
subset $D_a\subset D_\vareps$, and is contained in the unstable $\gamma_1$-cone:
\[
\Lambda \simeq \set{ (y^a_1, u^a;E'-E, \varphi' ( u^a ) ) \; : \; u^a \in D_a},\quad  
\text{with}\quad \|\varphi''\|_{C^0(D_a)}\leq\gamma_1\,.
\]
Then for any index $a'\in A_1$, the Lagrangian leaf $ \Phi^{t_0} ( \Lambda ) \cap W_{a'}$ is also
in the unstable $\gamma_1$-cone in the chart $(y^{a'},\eta^{a'})$.

Besides, the map $y^{a}\mapsto y^{a'}$ obtained by projecting $\Phi^{t_0}\rest_{\Lambda}$
on the planes $\{\eta^a=0\}$ and $\{\eta^{a'}=0\}$, satisfies
the following estimate on its domain of definition:
$$
\det \Big( \frac{\partial y^{a'}}{\partial y^{a}} \Big) =  (1+\cO(\vareps^\gamma))\,
\exp\big(\lambda^+_{t_0}(\rho_a))\,.
$$
Here $\lambda^+_{t_0}(\rho_a)$ is the unstable Jacobian \eqref{e:unst-jac} of the reference
point 
$\rho_a\in W_a\cap K_E^\delta$,
and $\gamma>0$ is the H\"older exponent of the unstable lamination. The implied
constant is uniform with respect to $t_0$.
\end{prop}
\begin{proof}
From Proposition~\ref{p:inclination}, we know that for any $\rho\in \Lambda$ and any $k=0,\ldots,t_0-1$, 
the connected component $\Lambda^{k}_{\loc}$ of $\Phi^{k}(\Lambda)\cap V_{b_k}$ containing 
$\Phi^{k}(\rho)$ lies
in the unstable $\gamma_1$-cone with respect to the chart $(y^{b_k},\eta^{b_k})$. 
On the other hand, since $\Lambda$ is a connected leaf inside $W_a$, at each 
step $k=0,\ldots,t_0-1$ its image $\Phi^{k}(\Lambda)$ is fully contained in $V_{b_k}$
and is {\em connected}, so that $\Lambda^k_{\loc}$ is actually equal to $\Phi^{k}(\Lambda)$
for all $k=0,\ldots,t_0-1$. Finally, we apply one iteration of Proposition~\ref{p:inclination} 
to the leaf $\Lambda'=\Phi^{t_0-1}(\Lambda)\subset V_{b_{t_0-1}}\cap \cE_{E'}$, and deduce that 
any intersection $\Phi(\Lambda')\cap W_{a'}=\Phi^{t_0} ( \Lambda ) \cap W_{a'}$
is also in the $\gamma_1$-unstable cone.

We now prove the statement concerning the Jacobian of the induced map. It is a direct 
consequence of part $iii)$ in 
 Proposition~\ref{p:inclination}, after replacing the time
$N$ by $t_0$.
Let $\rho_a$ be the reference point in $W_a\cap K_E^\delta$,
on which the coordinates $(y^a,\eta^a)$ are centered. 
If $V_{b}$ is a set containing 
$\Phi^{t_0}(\rho_a)$, we may enlarge it into a set of diameter $C\vareps$, such that
$\Phi^{t_0}(W_a)\subset V_{b}$ and $W_{a'}\subset V_{b}$.
On $V_b$ we may use adapted coordinates $(y^{b},\eta^b)$ centered on the point 
$\rho_b\defeq\Phi^{t_0}(\rho_a)$, and
represent $\Phi^{t_0}\rest_{\Lambda}$ by a map $y^a\mapsto y^b$. In this setting,
Proposition ~\ref{p:inclination},$iii)$ shows that the associated Jacobian satisfies
$$
\det\Big(\frac{\partial y^{b}}{\partial y^{a}}\Big) = (1+\cO(\eps))\,\exp(\lambda^+_{t_0}(\rho_{a}))\,.
$$
There remains to compare the coordinates $(y^b,\eta^b)$ with the coordinates $(y^{a'},\eta^{a'})$
centered on $\rho_{a'}\in W_{a'}$. Since the (un)stable subspaces at $\rho_b$ and $\rho_{a'}$
form angles $\cO(\eps^\gamma)$ and $d(\rho_b,\rho_{a'})=\cO(\eps)$, 
the representation of $\Phi^{t_0}\rest_{\Lambda}$ through $y^a\mapsto y^{a'}$ satisfies
\be\label{e:ya->yb}
\det\Big(\frac{\partial y^{a'}}{\partial y^{a}}\Big)=(1+\cO(\eps^\gamma))
\det\Big(\frac{\partial y^{b}}{\partial y^{a}}\Big)
=(1+\cO(\eps^\gamma))\,\exp(\lambda^+_{t_0}(\rho_{a}))\,.
\ee
\end{proof}
Notice that, even though $t_0$ (depending on the cover 
$\cV_0$ in an unknown way) can be very large,
applying $\Phi^{t_0}$ onto a near-unstable isoenergetic leaf $\Lambda\subset W_a$ 
does not fold it.

\subsection{Completing the cover}\label{s:covers}
We need to complete the family $(W_a)_{a\in A_1}$ in order to
cover the full energy strip $\cE_E^\delta$. 
Far from the interaction region 
(which we define using the radius $R_0$ of \S\ref{pr}),
we take the unbounded set 
$$
W_0 = \cE_E^\delta \cap \{|x(\rho)| > 3 R_0 \} \,.
$$
We complete the cover with a finite family of sets 
$$
( W_a \subset \cE_E^\delta)_{a\in A_2}\,, 
$$
with the following properties. These sets should have
sufficiently small diameters,
and for some uniform $d_1>0$ they should satisfy:
$$
d ( W_a ,  \Gamma_E^{+ \delta} ) + d ( W_a ,  \Gamma_E^{- \delta} ) > d_1 \,, \quad  
\text{where}\quad
\Gamma_E^{\pm \delta} \defeq \bigcup_{| E' - E| < \delta } \Gamma_{E'}^\pm \,, 
$$
where $\Gamma_E^{\delta}$ are the incoming/outgoing sets given in \eqref{eq:gaha}.
Finally, the
full family should cover $\cE_E^\delta$:
$$
\cE_E^\delta=\bigcup_{a\in A}W_a\,,\quad\text{where}\quad A=\{0\}\cup A_1\cup A_2\,.
$$
\begin{lem}
\label{l:cover}
Such a cover exists. Consequently, there exists $ N_0 \in \NN $, such that
for any index $ a \in A_2 $ we have
$$
\begin{cases}
\quad \Phi^{t} ( W_a ) \cap \{x(\rho)< 3 R_0 \} = \emptyset &\text{for any } \ 
 t\geq N_0t_0,\\
\quad \quad \quad \quad\text{or}\\
\quad
 \Phi^{- t} ( W_a ) \cap \{x(\rho) < 3 R_0 \}  = \emptyset  & \text{for any }\ t \geq N_0t_0\,. 
\end{cases}
$$
\end{lem}
\begin{proof}
The complement of $\cup_{a\in A_1}W_a$ in $\cE_E^\delta\cap T^*_{B(0,3R_0)}X$ 
is at a certain distance $D>0$ from $K_E^\delta$. On the other hand, from
the {\em uniform transversality} of stable and unstable manifolds, 
there exists $d_1>0$ such that
\be\label{e:unif-transv}
\forall \rho\in \cE_E^\delta\cap  T^*_{B(0,3R_0)}X,\qquad
d (\rho, \Gamma_E^{+ \delta}) + d (\rho,\Gamma_E^{- \delta})\leq 4\,d_1
\Longrightarrow d(\rho,K_E^\delta)\leq  D \,.
\ee
We first cover the set 
$S_-=\{\rho\in \cE_E^\delta\cap T^*_{B(0,3R_0)}X\,:\,d(\rho,\Gamma_E^{- \delta})> 2d_1\}$
by small open sets $\{W_a,\ a\in A_2^-\}$ at distance $\geq d_1$ from $\Gamma_E^{- \delta}$.
There exists $T_->0$
such that at any time $t\geq T_-$, the iterate
$\Phi^t(W_a)$ has escaped outside $T^*_{B(0,3R_0)}X$ for any $a\in A_2^-$.

We then cover the set 
$S_+=\{\rho\in \cE_E^\delta\cap T^*_{B(0,3R_0)}X\,:
\,d(\rho,\Gamma_E^{- \delta})\leq 2d_1,\ d(\rho,\Gamma_E^{+ \delta})> 2d_1\}$
by small open sets $\{W_a,\ a\in A_2^+\}$ at distance $\geq d_1$ from $\Gamma_E^{+ \delta}$. 
Now, there exists $T_+>0$
such that all these sets have escaped outside $T^*_{B(0,3R_0)}X$ for times $t\leq -T_+$. 
From \eqref{e:unif-transv}, points $\rho\in \cE_E^\delta\cap T^*_{B(0,3R_0)}X$ 
which are neither in $S_-$ nor in $S_+$
are at distance $\leq D$ from $K_E^\delta$, and therefore already belong to some $W_a$, $a\in A_1$.
Finally, we take $A_2\defeq A_2^-\cup A_2^+$, and $N_0\in \NN$ such that $N_0\, t_0\geq \max(T_-,T_+)$.
\end{proof} 

\section{Quantum Dynamics}
\label{qd}

As reviewed in \S \ref{defcs} resonances are the eigenvalues of
the complex scaled operator $ P_\theta $. To prove the lower bound
on the size of the imaginary part of a resonance $ z ( h ) $, with 
a resonant state $ u_\theta ( h ) \in L^2 ( X_\theta ) $,  
$ \| u_\theta \| = 1$, we want to estimate
$$
\exp ( - t | \Im z ( h ) | /h ) = \| \exp ( - i t P_\theta / h ) \,
u_\theta ( h ) \| \,, \quad   t \gg 1 \,,
$$
where the exponential of $ - i t P_\theta / h $ is considered
purely formally.
In principle that could be done by estimating $ \| \exp ( - i t P_\theta 
/ h ) \chi^w ( x, h D ) \|$, where $ \chi^w $ provides a localization 
to the energy surface. However, the imaginary part of $ P_\theta $
can be positive of size $ \sim  \theta  \sim  M h \log (1/h ) $ 
and that poses problems for such estimates. 

Hence the first step is to modify the operator $P_\theta$ without
changing its spectrum. To make the notation simpler, we normalize
the operator so that we work near energy $ 0$. In the case of \eqref{eq:Ph}
that means considering 
$$
P ( h ) = -h^2 \Delta + V ( x ) - E \,, \qquad  p ( x , \xi ) 
= |\xi|^2 + V ( x ) - E \,.
$$
Accordingly, the energy strips and trapped sets will be denoted by $\cE^\delta$, $K^\delta$.

\subsection{Modification of the scaled operator}\label{s:mso}

To modify the operator $ P_\theta $  we follow the 
presentation of \cite[\S\S 4.1,4.2,7.3]{SjZw04} which is 
based on many earlier works cited there.
 
Thus instead of $ P_\theta $ we consider the operator
$ P_{\theta, \eps } $ obtained by conjugation with an exponential weight:
\be\label{eq:pthep} 
P_{\theta , \eps } \defeq  e^{-\eps G^w / h } P_\theta
e^{ \eps G^w / h } \,, \ \ \eps = M_2 \theta\,, \ \
\theta = M_1 h \log ( 1/ h )  \,.
\ee
This section is devoted to the construction of an appropriate weight $G^w=G^w(x,hD)$.
The large constant $ M_1 $ will be of crucial importance for error
estimates in our argument and will be chosen large enough to control
propagation up to time $ M \log ( 1/h ) $, roughly $ M_1 \gg M $. The constant $M_2$
will also be given below.

We start with the construction of the weight $ G(x,\xi) $:
\begin{lem}
\label{l:gsa}
Suppose that $ p $ satisfies the general assumptions \eqref{eq:gac} 
(with the energy $E>0$ now in the interval $(-\delta,\delta)$).
Then for any open neighbourhood $ V $
of $ K^\delta $, $V\Subset T^*_{B(0,R_0)}X$, and any $ \delta_0 \in (0,1/2)$, 
there exists $ G \in \CIc ( T^*X ) $ such that 
\be
\label{eq:gsa} 
\begin{gathered}
\rho\in T^*_{B(0,3R_0)}X \ \Longrightarrow \ H_p G ( \rho ) \geq 0 \,, \\
\rho \in T^*_{B(0,3R_0)}X \cap (\cE^\delta \setminus V)  \ 
\Longrightarrow \ H_p G ( \rho ) \geq 1 \,,\\
\forall \rho \in T^* X,\ \ \ H_p G( \rho)  \geq -\delta_0 \,. 
\end{gathered}
\ee
\end{lem}
\begin{proof}
The construction of the function $ G $ is based on the following 
result of \cite[Appendix]{GeSj}:
for any open neighbourhoods $ U, V $
of $ K^\delta $, $  \overline U \subset V $, there exists 
$ G_0 \in \CI ( T^* X ) $, such that
$$
G_0\rest_U\equiv 0\,, \  \ H_p G_0 \geq 0 \,, \ \
H_p G_0 \rest_{ \cE^{2\delta} } \leq C \,,\ \ 
H_p G_0 \rest_{  \cE^{\delta} \setminus V } \geq 1 \,.
$$
Such a $G_0$ is an escape function, and is necessarily of unbounded support.
We need to truncate $ G_0 $ into a compactly-supported function, 
without making $ H_p G_0 $ too 
negative. For $ T > 0 $, $ \alpha\in (0,1)$ to be fixed later, let
$ \chi \in \CI ( \RR ) $ satisfy
\[ 
\chi ( t ) = \begin{cases} 0\,,  & | t | > T\,,  \\
t\,,  & | t | < \alpha T \,, \end{cases}\quad 
|\chi(t)|\leq 2\alpha T,\quad \chi' ( t ) \geq -2 \alpha,\quad t\in\RR\,.
\]
(we obtain $ \chi $ by regularizing a piecewise linear function with 
these properties). Let $ \psi \in \CIc ( \RR ; [ 0 , 1] ) $ be 
equal to $ 1 $ for $ | t| \leq 1 $ and $ 0 $ for $ | t | \geq 2 $. 
For $ R > 0 $ to be fixed later, we  define 
$$ 
G ( \rho ) \defeq \chi ( G_0 ( \rho ))\, \psi ( p ( \rho ) / \delta ) \,
\psi ( | x ( \rho ) | / R ) \,,
$$  
which vanishes on $U$, outside $\cE^{2\delta}$ and for $|x|>2R$. We then compute
$$  
H_p G   = \chi' ( G_0 )\, H_p G_0 \, \psi ( p / \delta )\, 
\psi ( | x| / R ) + (1/R)\, \chi ( G_0 )\, \psi ( p / \delta )\, \psi' ( |x|/ R )\, 
H_p ( | x | ) \,.
$$
This is bounded from below by $ 0 $ for 
$  \{| x| < R \,, \ | G_0 | \leq \alpha T\} $, and by $ 1 $, if in addition
$ \rho \in \cE^\delta\setminus V$. For any $ \rho \in T^* X $ we have 
$$
H_p G ( \rho ) \geq  - C_0\,\alpha (1 + T/ R ) \,,
$$
for some $ C_0 > 0 $, 
since \eqref{eq:validC}
shows that
$  | H_p ( | x | ) | \leq C_1 $ on $\cE^{2\delta}$.
Choosing $ R > 3 R_0 $ and $ T = T ( \alpha, R_0 ) $ large enough
so that $ | G_0 ( \rho ) | \leq \alpha T $ for $\rho\in \cE^{2\delta}\cap T^*_{B(0,3R_0)}X$, 
we have now guaranteed the first two conditions in 
\eqref{eq:gsa}. To obtain the last condition we need
$$ 
C_0\,\alpha (1  + T ( \alpha, R_0 ) / R ) < \delta_0 \,,
$$
and this follows from choosing $ \alpha $ small enough and then 
$ R $ large enough. 
\end{proof}
Using the identification \eqref{eq:idt}
 we consider $ G $ given in Lemma \ref{l:gsa} 
as a function on $ T^* X_\theta $, 
and define $ P_{\theta, \epsilon } $
by \eqref{eq:pthep}. We note that $ \exp ( \pm \epsilon G ( x , h D)  / h ) $ 
is a pseudodifferential operator with the symbol in the class 
$ S_\delta ( h^{ -C_0 } ) $ for any $ \delta > 0 $ and some $ C_0 $, and 
that the operator
$$
P_{\theta, \eps}  \defi
e^{-\eps G^w/h}\,  P_\theta \, e^{\eps G^w/h}=
e^{-\frac{\eps}{ h}{\rm ad}_{G^w}}\,P_\theta \sim 
\sum_{k=0}^\infty  \frac{\eps ^k}{ k!} (-\frac1{h}{\rm ad}_{G^w})^k(P_\theta ) 
$$
has its symbol in the class $S(\la\xi\ra^2)$.
This expansion shows that 
\begin{align*} 
P_{\theta,\eps} (h)
&  =P_\theta ( h ) - i\eps \{ p_\theta,G \}^w ( x, h D )  + \eps ^2 e_0^w ( x, h D )   \\
&  =p_\theta^w ( x , h D)  - i\eps \{ p_\theta ,G \}^w ( x , h D )  + \eps^2 e_1^w ( x , h D ) + 
h e_2^w ( x, h D) \,,  \quad  e_j \in S \,, 
\end{align*}
where $ p_\theta $ is the principal symbol of $ P_\theta $ given 
by \eqref{eq:pth}.
In particular, denoting by $ \cO ( \alpha ) $ the quantization 
of a symbol in $ \alpha S$, we have
\be\label{eq:split1}
\begin{split}
\Re P_{\theta, \eps} & \defeq  ( P_{\theta, \eps }
+ P_{\theta, \eps}^* )/2 = ( \Re p_\theta)^w ( x, h D) 
+ \eps \{ \Im p_\theta ,G\}^w ( x, h D)  + \cO(h+\eps^2) \\
& =
 \Re p_\theta ^w ( x, h D ) +\cO ( h + \theta \eps + \eps^2 ) \,, \\
\Im P_{\theta, \eps } & \defeq
( P_{\theta, \eps } - P_{\theta, \eps  }^* )/(2i) = 
 \Im p_\theta ^w ( x, h D ) - \eps \{ \Re p_\theta,G\}^w ( x , h D)  +\cO(h+\eps^2) \\
& =  \Im p_\theta ^w ( x, h D ) - \eps ( H_p G)^w ( x, h D)  + \cO(h+\eps^2) \,.
\end{split}
\ee
We can use our knowledge of $p_\theta$, see \eqref{eq:pth}-\eqref{eq:pth2},
and the fact that the set
$V$ used to define $G$ is contained in $T^*_{B(0,R_0)}X$,
to deduce that, for any $\rho\in \cE^\delta$, 
\be\label{e:sign}
\Im p_\theta (\rho ) - \epsilon H_p G ( \rho )
\leq \begin{cases}\qquad\qquad  0 &  \rho \in V\,,  \\
\ \  C \theta - \eps = - (M_2- C)\theta & \rho \notin V \,, \ \ | x ( \rho) | \leq 
2 R_0\,,  \\
- C \theta + \epsilon \delta_0 = - \theta(C-\delta_0 M_2) & |x ( \rho ) | > 2 R_0 \,, 
\end{cases}
\ee
We now choose 
$ M_2$ in \eqref{eq:pthep} such that $C< M_2 < C/\delta_0$,
so that 
\be\label{e:negative}  
\Im p_\theta (\rho )  - \epsilon H_p G ( \rho ) \leq 0 \, \quad 
\text{for any $\rho\in \cE^\delta$}\,.
\ee
The sharp G{\aa}rding inequality \eqref{eq:sharpg} and \eqref{e:negative}
give, in the sense of operators, 
\be
\label{eq:gaa}
\Im \chi^w ( x, h D) \, P_{\theta, \epsilon } ( h ) \, \chi^w ( x , h D) 
 \leq C h \,, \quad 
\supp \chi \subset \cE^{\delta/2} \,, 
\ee
where $ \chi \in S (1 )$ is real valued.
Achieving this approximate negativity was the main reason for 
introducing the weight $ G $. Indeed, we notice that, before conjugating by this
weight, we only had $ \Im (\chi^w P_\theta \chi^w) \leq  C h \log ( 1/ h ) $.

\subsection{The evolution operator}\label{s:eo}
We take the energy width $ \delta > 0 $ as in \S\ref{s:cl}, and construct the weight $G$
accordingly, as explained in the previous section.
Let the function $\chi_\delta \in S ( T^*X ) $ satisfy 
\be\label{eq:suppch}
\supp \chi_\delta \subset \cE^{\delta/2}  \,, \quad 
\chi_\delta\rest_{\cE^{\delta/3} } \equiv 1 \,. 
\ee
In this section we will compare the two energy-localized operators 
\be
\tP_0  \defi  \chi_\delta^w ( x , h D )\, P(h)\, \chi_\delta^w ( x, h D ) \qquad\text{and}\qquad 
\tP \defi  \chi_\delta^w ( x , h D )\, P_{ \theta, \eps }(h)\, \chi_\delta^w ( x, h D )\,.
\ee
$\tP_0$ is obviously 
bounded and hermitian on $L^2(X)$, and $\tP$ is bounded on 
$ L^2 ( X_\theta ) \simeq L^2 ( X ) $ (using
the map $ x \mapsto \Re x $). We may thus define a unitary group and
a non-unitary group as follows ($t\in\RR$):
\be\label{eq:Uthep}
U_0(t)\defi \exp ( - i t \tP_0 / h ),\quad\text{respectively}\quad 
U(t)\defi \exp ( - i t \tP / h ) \,. 
\ee
The need for the cutoff $ \chi_\delta^w$ comes from the non-dissipative contributions
of $ \Im P_\theta $, which are compensated by the weight $ G$  only close
to the energy surface. In view of the bound \eqref{eq:gaa} we have 
\be\label{eq:gaU}
\| U ( t ) \|_{ L^2 \to L^2 } \leq \exp ( C t) \,,  \qquad t \geq 0 \,.
\ee
We make the following observation based on \S\ref{defcs} and the
boundedness of $ e^{ \pm \eps G^w ( x , h ) / h } $ on $ L^2 $:
\begin{gather*}
\Res ( P ( h ) )  \cap D_{\delta,\theta/C}  = 
 \Spec ( P_\theta ( h ) ) \cap D_{\delta,\theta/C}  = 
 \Spec ( P_{\theta, \eps } ( h )  ) \cap D_{\delta,\theta/C}\,, \\ 
 D_{\delta,\theta/C} \defeq  \{ z \; : \; | \Re z | \leq \delta \,, \ \
 \Im z >  - \theta/C  \} \,. 
\end{gather*}
Hence, from now on, by a normalized resonant state of $ z ( h ) \in 
\Res ( P ( h ))  \cap D_{\delta,\theta/C}$ we  mean
\be\label{eq:resst}
 u( h )  \in L^2 ( X_\theta )  \,, \quad
\| u ( h ) \| = 1 \,, \ \  P_{\theta, \eps }\, u ( h ) = z ( h ) u ( h ) \,.
\ee
\begin{prop}\label{p:res}
Let us put $\delta_1=\delta/4$, $C>0$, 
and let $ u ( h ) $ be given by \eqref{eq:resst} 
with  $|\Re z(h)|<\nolinebreak \delta_1$, $ \Im z ( h ) > - C h $. Then for any fixed $M>0$ and any time  
$ 0 \leq t \leq M \log( 1/h ) $, we have 
\be\label{eq:approxU}  
U ( t )\, u ( h ) = \exp ( - i t z ( h ) / h )\, u ( h ) 
+ \cO_{L^2} ( h^\infty ) \,,
\end{equation}
where $ U ( t) $ is the modified propagator given by 
\eqref{eq:Uthep}. More precisely, the $ L^2 $ norm of the 
error in \eqref{eq:approxU} is bounded by $  h^L $ for any 
$L$ and $ 0 < h < h_0 = h_0 ( L, M ) $.
\end{prop}
\begin{proof}
Let $ v ( t ) \defeq U ( t ) u -  \exp ( - i t z / h ) u $, so that
\[ 
\begin{split} i h \partial_t v ( t ) & = \tP\, U ( t )\,  u - z e^{ - i t z/h }
u \\
& = \tP \, v ( t ) + e ( t) \,, \qquad  e( t) \defi e^{-it z / h} ( \tP - z ) u \,.
\end{split}
\]
Since $ ( P_{\theta, \eps } - z)u = 0 $, 
we know that $ \WF_h ( u ( h ) ) $ lies in $\cE^{\delta/3} $,
so that $ \chi_\delta^w u = u + \cO_{L^2}(h^\infty ) $. Hence, 
$ \| e ( t) \| = \cO ( h^\infty) $ and, using \eqref{eq:gaa},
\[ \begin{split} 
\partial_t \| v ( t ) \|^2 & = 2 \Re \la \partial_t v ( t ) , v ( t ) 
\ra = \frac2h \, \la \Im \tP v ( t) , v ( t ) \ra + 
2 \Im \la e ( t ) , v ( t) \ra \\ 
& \leq C \| v ( t ) \|^2 + \| e ( t ) \|^2 \,, \quad  v ( 0 ) = 0 \,. 
\end{split}
\]
The Gronwall inequality implies that
\[ \| v ( t ) \|^2 \leq  e^{ C  t  }\int_0^t \| e ( s )\|^2 ds \,,
\]
and the lemma follows from the logarithmic bound on $t$.
\end{proof}

The following lemma compares the two propagators in \eqref{eq:Uthep}. 
\begin{lem}
\label{l:new1}
For any fixed $t>0$, the operator
\be
\label{eq:uot}
V( t) \defeq  U_0(t)^{-1}\, U( t)
\ee
is a pseudodifferential operator of symbol
$v(t)\in S_{\gamma} ( T^* X )$ for any $ \gamma\in (0,1/2) $.
\end{lem}
\begin{proof}
To prove both statements, we simply differentiate $V(s)$ with respect to $s$: 
\begin{gather*}
 \partial_s V( s ) = \frac{1}{h}\, a(s)^w ( x , h D ) V ( s) \,, \quad  V( 0) = I \,, \\ 
a ( s )^w ( x , h D ) \defeq \frac{1}{i}\, U_0 ( s )^{-1} ( \tP - \tP_0 ) \ U_0 ( s ) \,. 
\end{gather*}
Using Egorov's theorem, we obtain the following general bounds on the symbol $a(s)$,
uniform for $s\in [0,t]$:
$$ 
- C h \log ( 1/ h ) \leq \Re a(s) \leq C h \,, \quad
| \partial^\alpha a(s) | \leq C_\alpha h \log ( 1/ h ) \,, \ \forall \, \alpha 
\in \NN^{2n} \,. 
$$
To show that $V(t)$ is the quantization of a symbol $v(t)\in S_{\gamma}$ 
we use the Beals's characterization of pseudodifferential operators recalled in 
\eqref{eq:beals}. We proceed by induction: suppose we know that
$$ 
V_{N-1} ( t ) \defeq \ad_{W_{N-1}} \cdots \ad_{W_{1} } V ( t ) = 
\cO_{L^2 \to L^2 } ( h^{ ( 1 - \gamma ) (N - 1 ) }) \,, \quad 
N \geq 1 \,,  \quad  V_0 ( t ) = V ( t) \,,
$$
where $ W_j $'s are as in \eqref{eq:beals}.
We now consider the differential equation satisfied by 
\[ V_N ( t) \defeq \ad_{W_N} V_{N_1} ( t) \,. \]
Using the
derivation property  $ \ad_W ( AB ) = (\ad_W A ) B + 
A ( \ad_W B ) $ we see that
\[\begin{split} 
 \partial_t V_N ( t ) 
& =  \ad_{W_N} \cdots \ad_{W_1} \big((a/h)^w ( x, h ) V ( t ) \big)  \\
& = ( a/ h )^w ( x, h D) V_N ( t ) + E_N ( t ) \,, \ \ 
E_N ( t) =  \cO_{L^2 \to L^2 }  ( h^{N ( 1 - \gamma ) })\,,
\end{split} \]
where we used the induction hypothesis and the fact that
\[  \ad_{W_{j_1}} \cdots \ad_{W_{ j_k} } ( a/h)^w ( x, h D) = 
\cO_{L^2 \to L^2 } ( \log ( 1/h ) h^{k} ) = 
\cO_{L^2 \to L^2 } ( h^{N( 1- \gamma) })\,.
\]
Since $ V_N ( 0 ) = 0 $, Duhamel's formula shows that
\[ V_N ( t ) = \int_0^t V ( t - s ) E_N ( s ) ds = 
\cO_{L^2 \to L^2 }  ( h^{N ( 1 - \gamma ) } )\,, \]
concluding the inductive step and the proof.
\end{proof}

The following lemma shows that the propagators $U(t)$ and
$U_0(t)$ act very similarly on wavepackets localized close to the trapped set.
\begin{lem}\label{l:l9}
Take $\delta_1=\delta/4$ as in Proposition ~\ref{p:res}, and consider the open sets 
$U_G\Subset \tilde U_G \Subset \cE^{\delta_1}\cap T^*_{B(0,R_0/2)}X$ such that 
$U_G$ is a neighbourhood of $K^\delta$, while the weight $G$ 
constructed in Lemma~\ref{l:gsa} vanishes identically on $\tilde U_G$.

Fix some $t>0$. 
Assume that the open set $V$ is such that
$\Phi^s(V)\Subset U_G $ for all times $s\in [0, t]$. Take
any $\Pi\in C^\infty_c(V)$. The propagators $U(t)$ and $U_0(t)$ then satisfy 
$$
(U(t)-U_0(t))\,\Pi^w(x,hD)=\cO_{L^2\to L^2}(h^\infty)\,.
$$
\end{lem}
\begin{proof}
The proof is very similar to that of the previous lemma.
The norm is equal to $\norm{(V(t)-1)\Pi^w}_{L^2\to L^2}$. Differentiating this
operator with respect to $t$, we find for all $s\in [0,t]$:
$$
\partial_s V( s)\, \Pi^w=
\frac{1}{ih}\, U_0 ( s )^{-1} ( \tP - \tP_0 ) \,U_0 ( s )\, V ( s )\,\Pi^w\,.
$$
From the dynamical
assumption and using Egorov's
theorem, we easily deduce that 
$$
U_0 ( s )^{-1} ( \tP - \tP_0 ) \, U_0 ( s ) = 0\,, \ \ \text{microlocally near $U_G$,}
$$  
uniformly for all $s\in [0,t]$. Since $\Pi$ is supported 
inside $V$, we obtain $\partial_s V( s)\, \Pi^w = \cO_{L^2\to L^2}(h^\infty)$.
\end{proof}

Using Lemma \ref{l:new1} we also prove a basic semiclassical propagation estimate for
$U(t)$.
\begin{prop}\label{p:new1}
Take $\delta_1$ as in Proposition \ref{p:res} and 
fix some $t>0$ and some $\gamma\in [0,1/2)$. 

i) Take $\psi_0,\ \psi_1\in S_{\gamma}(1)$ such that $\psi_1\circ\Phi^t$ takes the value $1$ near
$\supp \psi_0$: precisely, assume
\be\label{e:nested-delta}
d\big( \supp \psi_0, \complement\{\rho\,:\,\psi_1\circ\Phi^t(\rho)= 1\}\big)\geq h^{\gamma}/C\,,
\quad \supp\psi_1\subset\cE^{\delta_1}\,,
\ee
where $d(\bullet,\bullet)$ is a Riemannian distance on $T^*X$
which coincides with the standard Euclidean distance
outside $T^*_{B(0,R_0)}X$.
Then
\be\label{e:propag}
\psi_1^w( x, h D )\, U(t)\, \psi_0^w( x, h D ) = U(t)\, \psi_0^w( x, h D )
+\cO_{L^2\to L^2}(h^\infty)\,.
\ee
ii) If $\psi_0,\,\psi_1\in S_\gamma(1)$ are such that $\psi_0=1$ near $\supp \psi_1\circ \Phi^t$,
then 
\be\label{e:propag2}
\psi_1^w( x, h D )\, U(t)\, \psi_0^w( x, h D ) = \psi_1^w( x, h D )\,U(t)
+\cO_{L^2\to L^2}(h^\infty)\,.
\ee
\end{prop}
Before proving the proposition we remark that if instead
$\psi_0,\ \psi_2\in S(1)$ satisfy
\be
d\big(\supp \psi_0,\,\supp \psi_2\circ\Phi^t \big)\geq 1/C,
\quad \supp\psi_j\subset\cE^{\delta_1}\,,
\ee
then 
\be\label{e:kill}
\psi_2^w( x, h D )\, U(t)\, \psi_0^w( x, h D ) = \cO_{L^2\to L^2}(h^\infty)\,.
\ee
Indeed, we can apply \eqref{e:propag} with $ \psi_1 = 1 - \psi_2 $. 

\medskip
\noindent
{\em Proof of Proposition \ref{p:new1}:}
We use  Lemma \ref{l:new1} to write
\be\label{e:product}
\psi_1^w( x, h D ) U(t) \psi_0^w( x, h D )= 
U_0(t)\big(U_0(t)^{-1}\,\psi_1^w( x, h D )\, U_0(t)\big)\,V(t)\psi_0^w( x, h D )\,.
\ee
Pseudodifferential calculus on $\Psi_{h,\gamma,h}$ (see for instance 
\cite[Chapter 7]{DiSj} or \cite[Chapter 4]{EZ})
shows that the wavefront set of the operator $V(t) \psi_0^w( x, h D )$ 
is a subset of $\supp\psi_0$, while Egorov's theorem and the 
condition \eqref{e:nested-delta} implies that 
$U_0(t)^{-1}\,\psi_1^w( x, h D )\, U_0(t)=I$ microlocally 
in an $ h^\gamma$-neighbourhood 
of $\supp\psi_0$. 
This operator
can thus be omitted in \eqref{e:product}, up to an error $\cO(h^\infty)$, which proves the 
first statement.

The proof of second statement goes similarly: $\psi_0^w( x, h D )=1$ microlocally near 
the wavefront set of $U_0(t)^{-1}\,\psi_1^w( x, h D )\, U_0(t)$.
\stopthm

We can use this proposition to show that the 
``deep complex scaling'' region acts as an absorbing potential, that is, strongly
damps the propagating wavepackets.
\begin{lem}\label{l:deep}
Take $\delta_1$ as in Proposition \ref{p:res}, $R_0$ as in \eqref{eq:gpr} and fix some time $ t_1 \geq 0 $. 
Then, for any symbol $ \psi \in S ( T^* X  ) $ satisfying
\be
\label{eq:condps} 
\forall t\in[0,t_1], \quad \supp (\psi\circ\Phi^{-t}) \subset \cE^{4\delta_1/5}\cap \{|x(\rho)| > 5\, R_0/2\}\,,
\ee
we have 
\be\label{eq:deep}
 \| U ( t_1 ) \psi^w ( x , h D) \|_{L^2\to L^2} \leq 
\exp\left(- \frac{ \theta}{h C_0} 
\right) \| \psi^w ( x, h D ) \|_{L^2\to L^2} + \cO_{L^2\to L^2}(h^\infty )
\,,
\ee
where $C_0>0$ is independent of the choice of $\psi$.
\end{lem}
\begin{proof}
For any symbol $ \psi_0 \in S(1)$ supported inside $\cE^{\delta_1}\cap\{x>2 R_0\}$, 
the estimates \eqref{e:sign} imply that
\be\label{eq:ptil}
 \Im \la \tP \psi_0^w ( x, h D )  u , \psi_0^w ( x , h D) u \ra \leq 
- \frac {\theta}{C_1 } \|  \psi_0^w ( x, h D )  u\|^2 
+ \cO ( h^\infty ) \| u \|^2 
\ee
for some $C_1>0$.
From the hypothesis \eqref{eq:condps} on $\psi$, and assuming $R_0$ is large enough,
there exists a symbol $ \psi_1\in S(1)$ such that
$$
\supp\psi_1\subset \cE^{\delta_1}\cap\{x>2 R_0\}\quad\text{and}\quad
d(\supp \psi,\,\complement\{\rho\, :\, \psi_1\circ\Phi^t(\rho)=1\}\big)>1/C\,,\quad t\in [0,t_1]\,.
$$
Proposition \ref{p:new1}, $i)$ then shows that
$$
\psi_1^w  ( x , h D )\,U ( t ) \psi^w ( x, h D )= 
 U ( t )\, \psi^w ( x, h D ) + \cO_{L^2\to L^2} ( h^\infty ) \,,\quad\text{uniformly for }t\in [0,t_1]\,.
$$
Combining this with \eqref{eq:ptil} we obtain, uniformly for $t\in [0,t_1]$:
\[ \begin{split} 
\partial_t \| U ( t ) \psi^w u \|^2 & = 
\frac{2}{h}\, \la \Im \tP \psi_1^w   U ( t ) 
\psi^w u , \psi_1^w  U ( t ) \psi^w  
u \ra + \cO ( h^\infty ) \| u \|^2 \\
 & \leq - \frac{2 \theta} { C_1 h }\, \|  U ( t ) \psi^w u \|^2 
+ \cO ( h^\infty ) \| u \|^2 \,, 
\end{split}\]
from which the lemma follows by Gronwall's inequality, with $1/C_0 = 2t_1/C_1$..
\end{proof}

\subsection{Microlocal Partition}
We consider $\delta_1=\delta/4$ as in Proposition \ref{p:res}, 
and take a smooth partition of 
unity adapted to $(W_a\cap \cE^{\delta_1})_{a\in A}$, which by quantization produces
a family $ (\Pi_a \in \Psi_h ( 1 ))_{a\in A} $ such that  
$$
\WF_h ( \Pi_a)\subset W_a\cap \cE^{3\delta_1/4},\quad \Pi_a = \Pi_a^*, \quad \text{and}\quad
\sum_{ a\in A} \Pi_a = I \quad \text{microlocally near $\cE^{\delta_1/2}$.}
$$
The difference
$$
\Pi_\infty\defeq I - \sum_{ a\in A} \Pi_a
$$
is also a pseudodifferential operator in $\Psi_h(1)$, and
$$
\WF_h ( \Pi_\infty)  \cap \cE^{\delta_1/2} = \emptyset\,.
$$
Using this microlocal partition of unity, we 
decompose the modified propagator \eqref{eq:Uthep} at time $t_0$:
\be\label{eq:partu}
U ( t_0 ) = \sum_{a\in {A\cup\infty}} U_a\,, \qquad 
U_a \defeq U(t_0) \, \Pi_a \,. 
\ee
We then decompose the $N$-th power of the propagator as follows:
\begin{equation}
\label{e:power}
U( N t_0 ) = \sum_{ \alpha \in A^N} 
U_{\alpha_N} \circ \cdots \circ U_{ \alpha_1 } + R_N \,.
\end{equation}
The remainder $R_N$ is the sum over all sequences $\alpha$ containing at least one
index $\alpha_j=\infty$.
The following lemma shows that the remainder $R_N$ is irrelevant when applied to
states microlocalized near $\cE$:
\begin{lem}\label{l:3} 
Suppose that $ \chi \in \CIc ( T^* \RR ) $ is supported inside $\cE^{\delta_1/5}$
and that we consider logarithmic times in the semiclassical limit:
\begin{equation}
\label{eq:bddN}
N \leq M \log \frac 1 h \,, \quad M>0\ \text{fixed}\,.
\end{equation}
Then the remainder term in \eqref{e:power} satisfies
$$   
\| R_N\, \chi^w ( x , h D) \|_{L^2 \rightarrow L^2 } = \cO ( h^\infty )  \,,
$$
with the implied constants depending only on $ M $.
\end{lem}
\begin{proof}
Let $\alpha\in A^N$ be a sequence containing at least one index $\alpha_j=\infty$.
Call $j_m$ the smallest integer such that $a_j=\infty$, so
the corresponding term in $R_N$ reads
$$
U_{\alpha_N}\cdots U_{\alpha_{j_m+1}}\,U(t_0)\,\Pi_\infty\,U_{\alpha_{j_m-1}}\cdots U_{\alpha_1},\quad
\text{with}\quad\alpha_1,\ldots,\alpha_{j_m-1}\in A\,.
$$
The lemma will be proved once we show that
\be\label{e:pi_infty}
\Pi_\infty\,U_{\alpha_{j_m-1}}\cdots U_{\alpha_1}\,\chi^w(x,hD) = \cO_{L^2\to L^2}(h^\infty)\,,
\ee
with implied constants uniform with respect to the sequence $\alpha$. Indeed, the remaining
factor on the left is bounded as
$$
\norm{U_{\alpha_N}\cdots U_{\alpha_{j_m+1}}\,U(t_0)}\leq C\,e^{CN}\leq C\,h^{-CM}\,,
$$
and the full number of sequences is $(|A|+1)^N=\cO(h^{-C'M})$.

The estimate \eqref{e:pi_infty} is obvious if $j_m=0$, because $\WFh(\Pi_\infty)$
and $\WFh(\chi^w)$ are at a positive distance from each other. To treat the cases $j_m>0$, 
we will define a family of $N$ nested symbols which cutoff in energy in various ranges
between $\delta_1/4$ and $\delta_1/2$. Because $N\sim \log(1/h)$, we must
use symbols in some class $S_{\delta'}(1)$,  $\delta'\in(0,1/2)$. We first define a sequence of 
functions $\tchi_j\in \CIc(\RR,[0,1])$, $j=1,\ldots,N$, as follows:
$$
 \tchi_1(t)= \begin{cases}1 &|t|\leq \delta_1/4\\0& |t|\geq \delta_1/4 + h^{\delta'}/2,\end{cases}\qquad
\tchi_{j+1}(t)=\begin{cases}1 & |t|\leq \delta_1/4\\
\tchi_j(|t|-h^{\delta'}) & |t|\geq \delta_1/4\,,\end{cases}\quad j\geq 1\,.
$$
The function $\tchi_N$ vanishes for $|t|\geq \delta_1/4 + Nh^{\delta'}$, and we will take $h$ small enough
so that $\delta_1/4 + Nh^{\delta'}< \delta_1/2$.
From there, the energy cutoffs $\chi_j\in S_{\delta'}(1)$ are defined by
$$
\chi_j(x,\xi)\defeq \tchi_j \big(p(x,\xi)\big)\,,j=1,\ldots,N\,.
$$
From the support properties of $\chi$, the first cutoff satisfies 
\be\label{e:chi_1}
\chi_1^w(x,hD)\,\chi^w(x,hD) = \chi^w(x,hD)+\cO_{L^2\to L^2}(h^\infty)\,.
\ee
For any $j=1,\ldots,j_m-1$, we have $\chi_j=\chi_j\circ\Phi^{t_0}$, and the nesting between $\chi_j$
and $\chi_{j+1}$ allows us to apply the propagation results of Proposition~\ref{p:new1}, $i)$:
\be\label{e:j-j+1}
\chi_{j+1}^w(x,hD)\, U_{\alpha_j}\,\chi_j^w(x,hD) = U_{\alpha_j}\,\chi_j^w(x,hD) +\cO(h^\infty)\,,
\quad j=1,\ldots,j_m-1\,.
\ee
Therefore, inserting $\chi_{j+1}^w$ after each $U_{\alpha_j}$ leaves
the operator \eqref{e:pi_infty} almost unchanged.
Finally, the cutoff, 
$\chi_{j_m}$ is supported in the energy shell $\{ |p(\rho)|\leq \delta_1/4 + N\, h^{\delta'}  \}$
which, for $h$ small enough, is at finite
distance from $\WFh(\Pi_\infty)$, so that
$$
\Pi_\infty\,\chi_{j_m}^w(x,hD)=\cO_{L^2\to L^2}(h^\infty)\,.
$$
Combining this expression with (\ref{e:chi_1},\,\ref{e:j-j+1}) proves \eqref{e:pi_infty} and the lemma.
\end{proof}

\medskip

The set $A^N$ of $N$-sequences can be split between several subsets.
The set $ {\mathcal A}_N \subset A^N  $ is defined as follows:
\be\label{e:A_N}
\alpha=\alpha_1\ldots\alpha_N  \in \cA_N  \; \Longleftrightarrow 
\begin{cases} 
\Phi^{t_0} ( W_{\alpha_j} ) \cap {W_{ \alpha_{j+1} }}  \neq \emptyset \,,
\quad j=1,\ldots,N-1\,, \\
\text{and} \ \ \alpha_j \in A_1 \,,\quad  N_0 < j < N - N_0 \,. \end{cases}
\ee
The sequences in $\cA_N$ spend most of the time in the vicinity of
the trapped set.

The next Lemma
shows that we can discard all sequences except for those
in $\cA_N $:
\begin{lem}\label{l:4} 
Suppose that \eqref{eq:bddN} holds. Then there exists $C_1>0$ such that, 
for $h$ small enough,
\[ 
\sum_{ \alpha \in A^N \setminus \cA_N } 
\| U_{\alpha_N } \circ \cdots \circ U_{\alpha_1} \| \leq 
C_1\,|A|^N\,e^{C_1 N t_0}\, e^{- \theta/{C_1 h } } \,. 
\]
If $ N \leq M \log ( 1/ h ) $, 
$ \theta = M_1 h \log ( 1/h ) $, and 
$ M_1 \gg M t_0 $, this implies that
\begin{equation}
\label{eq:l4}
\sum_{ \alpha \in A^N \setminus \cA_N } 
\| U_{\alpha_N } \circ \cdots \circ U_{\alpha_1} \| \leq  h^{ M_1/C_2} \,,
\quad 0 < h < h_0 ( M,M_1,|A|).
\end{equation}
\end{lem}
\begin{proof}
Take $\alpha\in A^N\setminus \cA_N$. If the first condition on the right 
in \eqref{e:A_N} is violated, then
the property $\WFh(\Pi_a)\Subset W_a\cap\cE^{3\delta_1/4}$ for $a\in A$
and Eq.~\eqref{e:kill} imply that $\|U_\alpha\|=\cO(h^\infty)$.

Assume that for some $j$, $ N_0 < j < N - N_0 $, we have 
 $ \alpha_j \notin A_1 $. 
We have three possibilities. First, assume $\alpha_j=0$. In that case, the factor
$U_{\alpha_j}=U(t_0)\Pi_0$ can be decomposed as $U(t_0-1)U(1)\Pi_0$. If $R_0$ has
been chosen large enough, the set $W_0=\cE^{\delta}\cap\{x(\rho)>3 R_0\}$ satisfies
the property 
$$
\Phi^t(W_0)\subset \{|x|> 8 R_0/3\},\quad t\in [0,1]\,.
$$
Using the fact that $WF_h(\Pi_0)\subset W_0\cap \cE^{3\delta_1/4}$ 
and applying Lemma~\ref{l:deep} for $t_1=1$, we find
\be\label{e:deep1}
\|U(t_0-1) U(1) \Pi_0 \| \leq e^{C\,(t_0-1)}\,C_0 \exp(- \theta / h C_0) + \cO(h^\infty )\,.
\ee
Second, assume $\alpha_j\in A_2^-$, where we use the same notation
as in the proof of Lemma~\ref{l:cover}. In that case,
\be\label{e:propag1}
 \Phi^{ t } ( W_{\alpha_j }) \subset W_0 \quad\text{for any}\quad t\geq N_0 t_0 \,. 
\ee
Applying Proposition~\ref{p:new1}, $i)$, $N_0$ times, one realizes that the operator
$$
\Pi_{ \alpha_{ j+ N_0} }
\, U_{\alpha_{j+N_0-1} } \cdots U_{\alpha_{j+1} } \,U_{\alpha_j } 
$$
is negligible unless $WF_h(\Pi_{ \alpha_{ j+ N_0} })$ intersects $W_0$. This is the case if 
$ \alpha_{ j+ N_0}=0$, or $\alpha_{ j+ N_0}\in A_2$ and $W_{\alpha_{ j+ N_0}}\cap W_0\neq \emptyset$.
In both cases, we have (as long as $R_0$ has been taken large enough)
$$
\Phi^{ t } ( W_{\alpha_{j+ N_0} })\subset \{|x|> 8 R_0/3\},\quad t\in [0,1]\,,
$$
and the estimate \eqref{e:deep1} applies to $\| U(t_0)\Pi_{ \alpha_{ j+ N_0} }\|$.

Third, if $j\in A_2^+$, we have $\Phi^{ t } ( W_{\alpha_j})\in W_0$ for $t<-N_0 t_0$. 
Again, iterating Proposition~\ref{p:new1}, $i)$ $N_0$ times shows that the operator 
$$
\Pi_{\alpha_j}\,U_{\alpha_{j-1} } \cdots \,U_{\alpha_{j-N_0+1} } U(t_0)\,\Pi_{\alpha_{j-N_0}}
$$
will be negligible unless $W_{\alpha_{j-N_0}}$ intersects $W_0$. This yields 
\[ \|U(t_0)\,\Pi_{\alpha_{j-N_0}} \| 
 \leq e^{C\,(t_0-1)}\,C_0 \exp(- \theta / h C_0) + \cO(h^\infty )\,. \]

For these three cases, we find, using \eqref{eq:gaU}, 
\[ 
\| U_{ \alpha_N }\,  \cdots \circ U_{\alpha_1 }\| \leq  e^{C\,(N-N_0) t_0}\,\exp(- \theta / h C_0)\,.
\]
This estimate concerns an individual element $\alpha\in A^N\setminus \cA_N$.  
Summing over all such elements produces a factor $|A|^N$, which proves the first estimate. 
The second estimate follows from the assumptions on $ N $ and $\theta$.
\end{proof}

The following proposition, which is at the center of the method, 
controls the terms $\alpha\in \cA_N$ in \eqref{e:power}. The proof
is more subtle than for the above Lemmas, and uses the whole machinery of
Sections~\ref{s:iterating} and \ref{s:cl}. In particular, a crucial use is made
of the hyperbolicity of the classical dynamics on $K^\delta$. For this reason,
we call the following bound a {\em hyperbolic dispersion estimate}.
\begin{prop}\label{p:crucial}
Assume the time $N\leq M\log(1/h)$ for some $M>0$. Then, if the 
diameter $\vareps>0$ of the cover $\cV_0$ has been
chosen small enough, for any 
$ \alpha \in {\mathcal A}_N \cap A_1^N $ we have the following bound:
\be\label{eq:crest}
\| U_{\alpha_N } \circ \cdots \circ U_{\alpha_1} \| \leq 
 h^{-n/2} ( 1 + \eps_0 )^N 
\prod_{j=1}^N \exp\Big\{\frac12\, S_{t_0}(W_{\alpha_j })\Big\} \,, 
\ee
where the coarse-grained Jacobian $S_{t_0}(\bullet)$ is defined in \eqref{e:coarse-jac}, and
$\eps_0$ is the parameter appearing in \eqref{e:pressure-approx}.
\end{prop}
Before proving this proposition in \S\ref{s:estimate}, we show how it
implies Theorem \ref{t:1g}.

\subsection{End of the proof of Theorem~\ref{t:1g}}\label{s:pressure}

Suppose that $ \| u ( h ) \| =1 $ is an 
eigenfuction of $ P_{\theta,\eps} ( h ) $, with the same conditions as in Proposition~\ref{p:res}:
$ P_{\theta,\eps}( h )  u( h )  = z ( h )  u ( h ) $, 
$ | \Re z( h ) - E | \leq \delta_1 $, $\Im z(h)> - Ch$. 
Then, taking $ t = N t_0 $, $ N \leq M \log ( 1/h ) $ in Proposition \ref{p:res}, we get
\[ 
\exp ( Nt_0\,\Im z ( h )/ h )  = \| U ( Nt_0  ) u ( h ) \| + \cO ( h^\infty ) \,. 
\]
Using the decomposition \eqref{e:power} and Lemmas~\ref{l:3}, \ref{l:4},
the state $U(Nt_0)\,u(h)$ can be decomposed as
\[ 
U(Nt_0)\,u(h)=\sum_{ \alpha \in {\mathcal A}_N } 
 U_{\alpha_N } \circ \cdots \circ U_{\alpha_1} u ( h ) + \cO_{L^2}( h^{M_3} ) \,, 
\]
where $ M_3 $ can be as large as we like, if we take $ \theta = 
M_1 h \log ( 1/h ) $ with $ M_1 $ large, depending on $ M t_0 $.

The norm of the right hand side can be estimated by applying \eqref{eq:crest} to the 
factors $U_{\alpha_{N-N_0}\cdots U_{N_0}}$. This leads to
\be\label{eq:press1}
\exp (N t_0  \Im z ( h ) / h ) \leq C  h^{-n/2} 
( 1 + \eps_0 )^N \,\sum_{ \alpha \in {\mathcal A}_N } 
 \prod_{j=N_0}^{N-N_0}  e^{\frac12\,S_{t_0}(W_{\alpha_j})} + \cO ( h^{M_3} )\,.
\ee
The sum over $\cA_N$ can be factorized:
$$
\sum_{ \alpha \in \cA_N } \prod_{j=N_0}^{N-N_0} e^{\frac12 
\,S_{t_0} (W_{\alpha_j})}
\leq
\Big(\sum_{ a\in A_1} e^{\frac12 \,S_{t_0}(W_a)}\Big)^{N-2N_0}\,.
$$
Combining this bound with \eqref{eq:press2}, we finally obtain:
\be
\exp \big(N t_0  \Im z ( h ) / h \big) \leq C'\, h^{-n/2} ( 1 +  \eps_0 )^N 
\exp \big( Nt_0 ( \cP_E^\delta ( 1/2 )+\eps )\big) + \Oo ( h^{M_3} ) \,.
\ee
Taking the logarithm and dividing by $ N t_0 $, we get
\[ 
\Im z ( h ) / h \leq  \cP_{E}^\delta ( 1/2 )   +  3 \eps_0  + n \frac{\log ( 1/h )}{ 2 N t_0} 
+ \log C' / Nt_0\,.
\]
We can take $ N = M \log ( 1/h ) $
with $ M $ arbitrary large (and consequently with $ M_1 $ in the
definition of $ \theta $, large), so that, for any $h$ sufficiently small (say, $h<h(\delta,\eps_0)$):
$$
\Im z ( h ) / h \leq  \cP_{E}^\delta  ( 1/2 )   + 4\eps_0\,. 
$$
In \S\ref{s:selection} we could take $ \eps_0 > 0 $ as small as we wished. 
This proves Theorem~\ref{t:1g} with a bound 
slightly sharper than \eqref{eq:press}.

\section{Proof of the hyperbolic dispersion estimate}\label{s:estimate}
To prove the estimate in Proposition~\ref{p:crucial}, we adapt the strategy of 
\cite{An,AnNo} to the present setting. We decompose an arbitrary state microlocalized
inside $W_{\alpha_1}$ into a combination of Lagrangian states associated with
``horizontal'' Lagrangian leaves (namely, Lagrangian leaves situated in some unstable cone). 
By linearity, the evolution of the full initial state can be estimated
by first evolving each of these Lagrangian states. Proposition~\ref{p:inclination}
shows that, being in an unstable cone,
the Lagrangians spread uniformly along the unstable direction, at a rate governed by
the unstable Jacobian. 
Proposition~\ref{p:an1} shows that this spreading implies a uniform exponential
decay of the norm of the evolved Lagrangian state, and by linearity, a uniform
decay of the full evolved state.

\subsection{Decomposing localized states into a Lagrangian foliation}

In this section we consider states $w\in L^2(\RR^n)$ with wavefront sets 
contained in an open neighbourhood $W$ of the
origin, $\WFh(w)\subset W\defi B(\vareps)_y\times B(\vareps)_\eta$. 
Here $B(\vareps)$ is the open ball of radius $\vareps$ in $\RR^n$. 
We will decompose such a state $w$ into a linear combination of ``local momentum states'' 
$(e_{\eta})_{\eta\in B(2\vareps)}$, associated
with horizontal Lagrangian leaves $(\Lambda_{\eta})_{\eta\in B(2\vareps)}$. Each Lagrangian
leaves $\Lambda_{\eta}$ is defined by
$$
\Lambda_{\eta}\defeq \set{(y; \eta)\in T^*\RR^n,\  y\in B(2\vareps)}\,,\qquad \eta\in B(2\vareps)\,.
$$
This family of Lagrangian foliates $B(\vareps)\times B(\vareps)$:
$$
W\Subset \bigcup_{ \eta \in B(2\vareps)} \Lambda_{\eta},\qquad
\Lambda_{\eta} \cap \Lambda_{\eta'} = \emptyset
\quad\text{if}\quad \eta \neq \eta' \,.
$$
The associated Lagrangian states $e_\eta$ are defined as follows. 
We start from the ``full'' momentum states $\tilde e_\eta\in \CI(\RR^n)$:
$$
\tilde e_\eta(y)= \exp(i \la \eta,y\ra/ h)\,,\quad y\in\RR^n\,,\quad \eta\in \RR^n\,,
$$
and we smoothly truncate these states in a fixed ball:
\be\label{e:e_eta}
e_\eta(y)\defeq\tilde e_\eta(y)\,\chi_{\vareps}(y),\quad \chi_{\vareps}\in \CIc(B(2\vareps)),\quad
\chi_{\vareps}\rest_{B(3\vareps/2)}=1\,.
\ee
Notice that all states $e_{\eta}$ satisfy 
\be\label{e:norm-e_eta}
\| e_{\eta} \|_{L^2}=\|\chi_{\eps}\| \leq C\,\vareps\,.
\ee
The $h$-Fourier decomposition of an arbitrary state $w\in L^2(\RR^n_y)$ reads
$$ 
w = \int_{ \RR^n }\frac{d\eta}{( 2 \pi h )^{n/2}}\,
 (\cF_h w) ( \eta )\,\tilde e_\eta \,.
$$
With the assumption 
$\WFh(w)\subset  B(\vareps)_{y}\times B(\vareps)_{\eta}$, one deduces that
\be\label{e:decomp}
w = \int_{B(2\vareps)}\frac{d\eta}{( 2 \pi h )^{n/2}}\,(\cF_h w) ( \eta )\,e_{\eta}
+\cO(h^\infty)\|w\|\,.
\ee
This is the decomposition into horizontal Lagrangian states we were aiming at.
If we apply a semiclassically tempered operator $T$ to this state (see \S\ref{s:semiclass}), 
we obtain
$$
T\,w = \int_{B(2\vareps)}\frac{d\eta}{( 2 \pi h )^{n/2}}\,(\cF_h w) ( \eta )\,(T\,e_{\eta})
+\cO(h^\infty)\|w\|\,.
$$
This gives the following bound for the norm of $T\,w$:
\be\label{e:bound11}
\begin{split}
 \| T\,w\|_{L^2} & \leq C\,h^{-n/2}\,\int_{B(2\vareps)}d \eta\,| (\cF_h w) ( \eta ) | \,\|T\, e_\eta \| 
+ \cO ( h^\infty ) \| w \|  \\
& \leq C\, h^{-n/2}\, \max_{\eta\in B(2\vareps) } \|T\, e_\eta \|\,\| w\| + \cO ( h^\infty ) \| w \|\,. 
\end{split}\ee

\subsubsection{Decomposition of the initial state into near-unstable Lagrangian states}
By using semiclassical Fourier integral operators,
see for instance \cite[Chapter 10]{EZ}, we can transplant the construction of the previous paragraph 
to any local coordinate representation. 
Here we will decompose states microlocalized in the 
sets $W_{a}$, $a\in A_1$. The horizontal Lagrangians are constructed with respect to the 
coordinate chart $(y^a,\eta^a)$ centered at some point $\rho_a\in W_a\cap K^{\delta}$, as described 
in Lemma~\ref{l:coord}. In order to cover the set $W_a$, we use the following  family $(\Lambda_{\eta,a})$:
\be\label{eq:anno}
\Lambda_{\eta,a}\equiv \set{(y^a; \eta^a),\  y^a\in B(2\vareps)}\,,\qquad \eta\in B(\delta,2\vareps)\,,
\ee
where $B(\delta,\vareps)\defeq \{\eta=(\eta_1,s)\in \RR^n,\ \ |\eta_1| < \delta,\ |s| < \vareps\}$.
Notice that these Lagrangians are isoenergetic ($\Lambda_{\eta,a}\subset \cE_{\eta_1}$), 
and they belong to arbitrarily thin unstable cones 
in $W_a$, in particular
the cones used in Proposition~\ref{p:inclination} and Remark~\ref{r:rem2}.

Using the Fourier integral operator $\cU_{a}$ associated with 
the coordinate change $(x,\xi)\mapsto(y^a,\eta^a) $ (see Lemma~\ref{l:Gam01}), 
each state \eqref{e:e_eta} can be brought to a Lagrangian state: 
$$
e_{\eta,a}= \cU_{a}^{*}\,e_{\eta}\quad
\text{associated with the Lagrangian leaf }\Lambda_{\eta,a}\subset \cE_{\eta_1}\,,
$$ 
with norms bounded as in \eqref{e:norm-e_eta}.


\subsection{Evolving the Lagrangian states through $U_{\alpha_N}\cdots U_{\alpha_1}$}
We now consider an arbitrary sequence $\alpha\in \cA_N\cap A_1^N$. 
For any normalized $u\in L^2(X)$, the state 
$$
w\defeq \Pi_{\alpha_1}u\quad\text{satisfies }\quad \WFh(w)\subset W_{\alpha_1}\cap \cE^{3\delta_1/4}\,,
$$ 
and can thus be decomposed into the
Lagrangian states $(e_{\eta,\alpha_1})$ associated with the leaves $\Lambda_{\eta,\alpha_1}$,
as in \eqref{e:decomp}.
In order to prove the estimate \eqref{eq:crest}, we will first study the individual
states
\be\label{e:opera}
U_{\alpha_N } \circ \cdots\circ U_{\alpha_2}\circ U(t_0)\,e_{\eta,\alpha_1}\,,\qquad 
\eta\in B(3\delta_1/4,2\vareps)\,.
\ee
We recall that each set $W_a$, $a\in A_1$, has the property
\be\label{e:succession}
\Phi^k(W_a)\subset V_{b_k},\qquad k=0,\ldots,t_0-1,\quad\text{for some sequence }b_0,\ldots, b_{t_0-1}\,.
\ee
Therefore, 
to the sequence $\alpha=\alpha_1\ldots\alpha_N\in A_1^N$ corresponds
a sequence $\beta=\beta_0\ldots\beta_{Nt_0-1}$ of neighbourhoods $V_{\beta_k}$ visited at the times 
$k=0,\ldots, Nt_0-1$. For later convenience, we also consider a set $V'_{Nt_0}$ (of diameter $C\vareps$),
which contains $\Phi^{t_0}(W_{\alpha_N})$.

From now on, we fix some $\eta\in B(3\delta_1/4,2\vareps)$ and compute the state \eqref{e:opera}, 
making use of various 
properties proved in the preceding sections.

\subsubsection{Evolution of the near-unstable Lagrangians $\Lambda_{\eta,\alpha_1}$}
The results of \S\ref{s:Schrod=FIO} and Lemma~\ref{l:0} show that 
it is relevant to study the evolution
of the Lagrangian $\Lambda^0_{\loc}\defeq \Lambda_{\eta,\alpha_1}\cap W_{\alpha_1}$ 
through the following operations:
one evolves $\Lambda^0_{\loc}$ through $\Phi^{t_0}$,
then restricts the result on $W_{\alpha_2}$, then evolves it through $\Phi^{t_0}$, 
restrict on $W_{\alpha_3}$, and so on. 
It is also useful to
consider the intermediate steps, that is, for $k=mt_0+m'$, $0\leq m<N$, $0\leq m'<t_0$,
we take
\begin{align*}
\Lambda^{mt_0}_{\loc}&\defeq \Phi^{t_0}(\Lambda^{(m-1)t_0}_{\loc})\cap W_{\alpha_m},\qquad m=1,\ldots,N-1\,,\\
\Lambda^{m t_0 + m'}_{\loc}&\defeq \Phi^{m'} (\Lambda^{mt_0}_{\loc})\,,\qquad m'=1,\ldots,t_0-1\,.
\end{align*}
Fix $\gamma_1=1/2$. By construction, $\Lambda^0_{\loc}$ is contained in the unstable 
$\gamma_1$-cone in the coordinates $(y^{\alpha_1},\eta^{\alpha_1})$. 
We can thus apply Proposition~\ref{p:inclination}, $i)$ and Proposition~\ref{c:fold} to this
sequence of Lagrangian leaves: each $\Lambda^k_{\loc}$ is contained in the unstable $\gamma_1$-cone
(when expressed in the coordinates $(y^{\beta_k},\eta^{\beta_k})$ on the set $V_{\beta_k}$).
Furthermore, part $ii)$ of the proposition shows that the higher
derivatives of the functions $\varphi_k$ generating $\Lambda^k_{\loc}$ also remain uniformly bounded with
respect to $k$. 
The sequence of Lagrangians is thus totally ``under control'', and the implied constants are
independent of the choice of $\eta\in  B(3\delta_1/4,2\vareps)$ parametrizing the initial
state $e_{\eta,\alpha_1}$.

\subsubsection{Analysis of the operator $U_{\alpha_N}\cdots U_{\alpha_1}$}
We now show that all the propagators $U(1)$ in \eqref{e:opera} 
may be replaced by the unitary propagators $U_0(1)$, up to a negligible error.
For each $a\in A_1$ we recall
that the set $W_a$ satisfies \eqref{e:succession}.
All the sets $V_b\in \cV_0$ were chosen so 
that $\Phi^t(V_b)$ remains close to $K^\delta$ in the interval $t\in [0,1]$.
As a result, one can apply Lemma~\ref{l:l9}
to the differences $(U(1)-U_0(1))\tPi_{b}^w$, where $\tPi_b\in\Psi_h(1)$ satisfies 
$$
\tPi_{b}= I\quad\text{near}\quad V_{b},\qquad 
\Phi^t(\WFh(\tPi_b))\Subset U_G\quad \text{for all}\quad t\in [0,1]\,.
$$
Each factor $U_a=U(t_0)\Pi_{a}$ can then be decomposed as follows:
\be\label{e:exchange}
\begin{split}
U(t_0)\,\Pi_{a}&= U(1)\tPi_{b_{t_0-1}}\cdots U(1) \tPi_{b_1} U(1)\Pi_{a}+\cO(h^\infty)\\
&= U_0(1)\tPi_{b_{t_0-1}}\cdots U_0(1) \tPi_{b_1} U_0(1)\Pi_{a}+\cO(h^\infty)\,.
\end{split}\ee
The first equality uses the propagation properties of Proposition~\ref{p:new1}, $i)$ and
\eqref{e:succession}. The second one is obtained by applying 
Lemma~\ref{l:l9} to all factors $U(1)\tPi_{b_k}$. 
The operator \eqref{e:opera} can thus be expanded as follows:
\be\label{e:prod10}
U_{\alpha_N}\circ \cdots \circ U_{\alpha_1}= 
S_{\beta_{Nt_0},\beta_{Nt_0-1}}\circ \cdots \circ S_{\beta_{1},\beta_0}\,\Pi_{\alpha_1} +\cO(h^\infty)\,,
\ee
where we called
\begin{align*}
S_{\beta_{k+1},\beta_k}&\defeq  \tPi_{\beta_{k+1}}\,U_0(1),\qquad k=0,\ldots, Nt_0-1\,,
\quad t_0 \nmid k+1\,,\\
S_{\beta_{k+1},\beta_k}&\defeq \Pi_{\alpha_{m+1}}\,U_0(1),\qquad k+1=m t_0,\ m=1,\ldots, N-1\,,\\
S_{\beta_{Nt_0},\beta_{Nt_0-1}}&\defeq \tPi'_{Nt_0}\,U_0(1)\,.
\end{align*}
The operator $\tPi'_{Nt_0}\in \Psi_h(1)$ on the last line has a compactly supported symbol,
and is equal to the identity, microlocally
near the set $V'_{Nt_0}$, so that $\tPi'_{Nt_0}\,U(t_0)\,\Pi_{\alpha_N}=U(t_0)\,\Pi_{\alpha_N}+\cO(h^\infty)$. 

From Lemmas~\ref{l:lprop} and \ref{l:Gam01}, each of the propagators 
$S_{\beta_{k+1},\beta_k}$ can be put in the form 
\be\label{e:change-coord}
S_{\beta_{k+1},\beta_k} = \cU_{\beta_{k+1}}^{*}\,T_{\beta_{k+1},\beta_k}\,\cU_{\beta_{k}}+\cO(h^\infty)\,,
\ee
where $\cU_{\beta_{k}}$ is the Fourier integral operator quantizing the local change of coordinates 
$(x,\xi)\to (y^{\beta_k},\eta^{\beta_k})$ (see Lemma~\ref{l:Gam01}),
while $T_{\beta_{k+1},\beta_k}$ is an operator of the form \eqref{e:S_j}, which quantizes 
the map $\kappa_{\beta_{k},\beta_{k-1}}$, obtained by expressing 
$\Phi^1$ in the coordinates $(y^{b_k},\eta^{b_k})\mapsto (y^{b_{k+1}},\eta^{b_{k+1}})$.

Inserting \eqref{e:change-coord} in \eqref{e:prod10}, we obtain
\begin{align*}
U_{\alpha_N}\cdots U_{\alpha_2}\,U(t_0)\, e_{\eta,\alpha_1}&= 
\cU_{\beta_{Nt_0}}^{*}\,\circ T_{\beta_{Nt_0},\beta_{0}}\,e_{\eta} 
+\cO_{L^2}(h^\infty)\,,\\
\text{where we took for short }\quad
T_{\beta_{Nt_0},\beta_{0}}&\defeq T_{\beta_{Nt_0},\beta_{Nt_0-1}}\circ \ldots \circ T_{\beta_1,\beta_0}\,.
\end{align*}
Here we used the fact that $\cU_{\beta_{k}}^{*}\cU_{\beta_{k}}=I$ microlocally near the
wavefront set of $\tPi^{(\prime)}_{\beta_{k}}$, $\Pi_{\alpha_m}$ or $\Pi'_{Nt_0}$.

\subsubsection{Applying the semiclassical evolution estimate}
The state $T_{\beta_{Nt_0},\beta_0}\,e_{\eta}$
has the same form as the left hand side in \eqref{eq:lit}.
Since the Lagrangians 
$$
\Lambda^k_{\loc}\equiv \{(y^{\beta_k},\eta^{\beta_k}=\varphi_k'(y^{\beta_k}))\}
$$ 
remain under control uniformly for $1\leq k\leq N$,
we can apply Proposition~\ref{p:an1} to obtain a precise description of
that state: for any integer $L>0$, we may write
$$
T_{\beta_{Nt_0},\beta_0}\,e_{\eta} (y) =
a^{Nt_0}(y)\,e^{i\varphi_{Nt_0}(y)/h} + h^L\,R^{Nt_0}_L(y)\,,\qquad
y\in \RR^n\,.
$$
The symbol $a^{Nt_0}$ admits an expansion,
$$
a^{Nt_0}(y)=\sum_{j=0}^{L-1}h^j\,a^{Nt_0}_j(y)\,, 
$$ 
which we now analyze. Starting from some $y\in B(C\vareps)$, 
let us assume that there exists no sequence of coordinates 
$$ 
(y^k)_{k=0,\ldots,Nt_0}\,, \ \ y=y^{Nt_0}\,, \ \ y^{k-1}=g_k(y^k)\,, 
$$
where $g_k$ is the projection of the 
map $\kappa_{\beta_{k},\beta_{k-1}}^{-1}\rest_{\Lambda^k_{\loc}}$ 
on the axes $\{\eta^{\beta_k}=0\}$, $\{\eta^{\beta_{k-1}}=0\}$. In that case,
Proposition~\ref{p:an1} shows that $a^{Nt_0}(y)=0$.

On the other hand, if such a sequence exists,
the principal symbol $a^{Nt_0}_0(y)$ satisfies a formula of the type \eqref{e:principal}. The functions
$\chi_{i_k}$ now correspond to the symbols of 
the operators $\tPi_{\beta_k}$, $\Pi_{\alpha_m}$ or $\Pi'_{Nt_0}$, 
which are uniformly bounded from above by $1+\cO(h)$.

The main factor in \eqref{e:principal} is the product of determinants $|\det dg_k(y^k)|^{1/2}$, which
corresponds to
the uniform expansion of the Lagrangians along the horizontal direction.
To estimate this product, we follow 
\S\ref{s:iterating} and group these determinants by packets of length $t_0$. According to Proposition~\ref{c:fold},
we have for any $t_0$-packet:
\[
\begin{split}
\prod_{k=mt_0+1}^{(m+1)t_0} |\det  d g_k (y^{k})|^{1/2} & = 
\det\Big( \frac{\partial y^{(m+1)t_0}}{\partial y^{mt_0}}\Big)^{-1/2}
\\ 
& = \, (1+\cO(\vareps^\gamma))\,
e^{-\lambda^+_{t_0}(\rho_{\alpha_m})/2}\,,\quad m=0,\ldots,N-1\,.
\end{split} \]
Here we have used the coordinate frames $(y^{\beta_{mt_0}},\eta^{\beta_{mt_0}})$ to label points in $W_{\alpha_m}$
instead of the coordinates $(y^{\alpha_m},\eta^{\alpha_m})$ centered at $\rho_{\alpha_m}\in W_{\alpha_m}$; 
this change does not modify the estimate of the Corollary, as is clear from \eqref{e:ya->yb}.
The product of determinants is thus governed by the unstable Jacobian along the trajectory.
Because the points $\rho_{\alpha_m}\in W_{\alpha_m}\cap K^\delta$ are somewhat arbitrary,
we prefer to use the coarse-grained Jacobian \eqref{e:coarse-jac} to bound the above right hand side.
Taking the product over all $t_0$-packets, we thus obtain, for some $C>0$ independent of $N$:
\be\label{e:sup-bound}
|a^{Nt_0}_0(y)|\leq \prod_{m=0}^{N-1}(1+ C\,h)^{t_0}\,(1+C\,\vareps^\gamma)\,
\exp\big(\frac12\,S_{t_0}(W_{\alpha_m})\big)\,,\quad y\in \supp a^{Nt_0}_0\subset B(C\vareps)\,.
\ee
The proof of Proposition~\ref{p:inclination} (see \eqref{e:bound10}) also shows that 
the determinants $\det dg_k(y)$ satisfy
$$
\sup_{y\in {\rm Dom}(g_k)}|\det dg_k(y)| \leq \det(A_k)^{-1}+ C\,\vareps^\gamma\,,\qquad k=1,\ldots Nt_0\,.
$$
We will assume $\vareps$ small enough, such that the right hand side is bounded from above
by $\nu_3<1$. 
This implies that the Jacobians $J_k$ of \eqref{e:D_k} decay exponentially
when $k\to\infty$. Henceforth, the higher-order symbols $a^{Nt_0}_j$, 
bounded as in \eqref{e:bounds},
are smaller than the principal symbol, so that the upper bound \eqref{e:sup-bound}
also holds if we replace $a^{Nt_0}_0$ by the full symbol $a^{Nt_0}$.
This decay of the $J_k$ also shows that the remainder
$R_L^{Nt_0}$, estimated in \eqref{e:bounds-rem}, is uniformly bounded in $L^2$. 
As a result, the bound \eqref{e:sup-bound} implies the following bound:
\be
\norm{T_{\beta_{Nt_0},\beta_{0}}\,e_{\eta}
}
\leq  C\,\vareps\,(1+ C\,\vareps^\gamma)^{N}\,
\prod_{m=0}^{N-1}\,e^{S_{t_0}(W_{\alpha_m})/2}\,,\qquad \eta\in B(3\delta_1/4,2\vareps)\,.
\ee
To end the proof of Proposition~\ref{p:crucial}, there remains to apply 
the decomposition \eqref{e:decomp} to the $w=\Pi_{\alpha_1} u$, 
with $u\in L^2(X)$ of norm unity,
and the bound \eqref{e:bound11} that follows:
$$
\norm{U_{\alpha_N } \circ \cdots\circ U_{\alpha_1}\,u}
\leq  C\,\vareps\, h^{-n/2}\,(1+C\,\vareps^\gamma)^{N}\,
\prod_{m=0}^{N-1}\,e^{S_{t_0}(W_{\alpha_m})/2} + \cO ( h^\infty ) \,.
$$
Notice that 
the main term on the right hand side is 
larger than $h^{M_3}$ for some $M_3>0$. This bound thus proves 
Proposition~\ref{p:crucial} if, given $\eps_0>0$, we choose the diameter $\vareps$ of the
partition $\cV_0$ small enough.
\stopthm

\section{Microlocal properties of the resonant eigenstates}\label{st2}

In this section we will use the results of \S \ref{s:mso} and \S \ref{s:eo}
to prove Theorem 4. We will turn back to the notation of \S\ref{ass}, that is,
the operator to keep in mind is 
$$ 
P ( h ) = -h^2 \Delta +  V( x )\,, \quad\text{with symbol}\quad  p ( x, \xi) =\xi^2 + V ( x ) \,.
$$
We also recall that
\be\label{eq:VG0}  
P_{\theta, \epsilon} =  e^{-\epsilon G^w / h } P_\theta\,
e^{ \epsilon G^w / h } \,, \ \ \epsilon = M_2 \theta\,, \ \
\theta = M_1 h \log ( 1/ h ) \,, \ \
  G^w = G^w ( x , h D) \,, 
\ee
where $ G $ is given by Lemma \ref{l:gsa}, and $M_1>0$ can be arbitrary large.
In this section we will choose the set $ V $  in Lemma \ref{l:gsa} to be 
\be\label{eq:VG}  
V = T^*_{B( 0 , 3R_0/4 )}X,\quad  \text{and assume that} \quad
G(x,\xi)=0\ \text{ for }\ x\in B ( 0 ,R_0/2  ) \,. 
\ee
We now consider a resonant state $ u $ in the sense of \eqref{eq:gath}, in particular
$ u\rest_{ B( 0 , R_0 )} =  u_\theta \rest_{B(0,R_0)} $. 
If $ u $ satisfies \eqref{eq:gath} for some choice of $ R_0>0 $ (which implies
a choice of deformation $X_\theta$, see \S\ref{defcs}), then it has the
same property with any larger $ R_0 $ (and associated $X_\theta$). 
The state
$$  
u_{\theta, \epsilon} \defeq e^{-\epsilon G^w / h }\, u_\theta \quad\text{is in }L^2 ( X_\theta ),
\quad\text{and satisfies}\ \ ( P_{\theta, \epsilon} - z )u_{\theta, \epsilon} = 0 \,.
$$
Furthermore, the support properties of $G$ imply
\be\label{eq:ute}
\|u_{\theta, \epsilon} - u \|_{L^2( B(0,R_0/2))}= 
\Oo( h^\infty)\, \| u_{\theta, \epsilon} \|_{L^2( X_\theta)} \,.
\ee
This first lemma controls the behaviour of $u_{\theta,\eps}$ near infinity.
\begin{lem}\label{l:1r}
Let $ P_{\theta,\epsilon} $ be the operator given by \eqref{eq:VG0} 
for some choice of $R_0\gg 1$ and $M_1\gg 1$.
Suppose that
\be\label{eq:l.1} 
 ( P_{\theta,\epsilon} - z ) u_{\theta,\epsilon}  = 0 \,, \ \ \Im z \geq -C\,h \,, 
\ \   \Re z = E + o(1) \,, \ \
 \| u_{\theta,\epsilon} \|_{L^2 ( X_\theta ) } = 1 \,.
\ee
Then, there exist $ R_1 >  4 R_0 $ and $ C_0>1 $, independent of $ M_1 $, 
such that 
\be\label{eq:l11}
\| u_{\theta, \epsilon} \|_{ L^2 ( X_\theta \setminus B_{\CC^n} ( 0 , R_1 ) )} 
= \Oo ( h^{M_1/C_0 } ) \,, 
\quad 0 < h < h_0 ( M_1 ) \,. 
\ee
\end{lem}
\begin{proof}
We will use the properties of the ``deep complex scaling'' region, explained in
Lemma~\ref{l:deep}. The first step is localization in energy. Take  $ \psi \in \CIc ( (-2,2) , [0,1] )$,
$\psi\rest_{[-1,1]}=1$, and define
\be\label{e:Energy0}
\psi_0(\rho)=\psi(4 (p(\rho)-E)/\delta_1),\qquad \psi_1(\rho)=\psi(8 (p(\rho)-E)/\delta_1)\,.
\ee
Fix some time $t_1>0$, and consider spatial cutoffs $\chi_0,\ \chi_1\in C^\infty(X,[0,1])$ localized
near infinity:
$$
\chi_j(x)=0\ \ \text{for } x\in B(0,\tR_j),\qquad \chi_j(x)=1\quad 
\text{for } x\in X\setminus B(0,\tR_j+1)\,,\quad
j=0,1\,,
$$
where the radii $\tR_1 > \tR_0+ 2 \gg 1$ are sufficiently large so that the following conditions
are satisfied:
\begin{align}
\forall t\in [0,t_1],\quad &\supp((\chi_0\,\psi_0)\circ \Phi^{-t})\subset 
\cE_E^{4\delta_1/5}\cap \{|x|>5R_0/2\},\label{e:cond-chi0}\\
&(\chi_0\,\psi_0)(\rho) = 1 \quad\text{near  }\supp((\chi_1\,\psi_1)\circ \Phi^{t_1})\,.
\label{e:cond-chi1}
\end{align}
We will now estimate the norm of the following state:
\be
v\defeq \chi_1\,\psi_1^w(x,hD)\,U(t_1)\,\chi_0\,\psi_0^w(x,hD)\,u_{\theta,\eps}
\ee
Using the condition \eqref{e:cond-chi1}, we apply Proposition~\ref{p:new1}, $ii)$
to the operator $\chi_1\,\psi_1^w\,U(t_1)\,\chi_0\,\psi_0^w$, and obtain
\begin{align*}
v &= \chi_1\,\psi_1^w(x,hD)\,U(t_1)\,u_{\theta,\eps} +\Oo_{L^2}(h^\infty)\\
&= e^{-it_1 z/h}\,\chi_1\,\psi_1^w(x,hD)\,u_{\theta,\eps} +\Oo_{L^2}(h^\infty)\\
&= e^{-it_1 z/h}\,\chi_1\,u_{\theta,\eps} +\Oo_{L^2}(h^\infty)\,.
\end{align*}
In the second equality we have applied Proposition~\ref{p:res}. For the third one we 
used the microlocalization of $u_{\theta,\eps}$ on $\cE_E$:
\be\label{e:microloc}
\psi_1^w(x,hD)\,u_{\theta,\eps} = u_{\theta,\eps} + \Oo_{H_h^k}(h^\infty),\quad\forall k\,.
\ee
On the other hand, the condition \eqref{e:cond-chi0} allows us to apply Lemma~\ref{l:deep}:
\begin{align*}
\| U(t_1)\,\chi_0\,\psi_0^w(x,hD)\,u_{\theta,\eps}\|&\leq 
e^{-\theta /hC_0}\,\| \chi_0\,\psi_0^w(x,hD)\,u_{\theta,\eps}\|+\cO(h^\infty)\\
&\leq h^{M_1/C_0}\,\| \chi_0\,u_{\theta,\eps}\|+\cO(h^\infty)\,.
\end{align*}
Here we have taken $\theta = M_1 h\log(1/h)$, and used again the microlocalization of $u_{\theta,\eps}$
near $\cE_E$. 

Using $\Im z\geq -Ch$ and combining the above estimates, we find
$$
\| \chi_1\,u_{\theta,\eps}\| \leq e^{Ct_1}\,h^{M_1 /C_0}\,+\cO(h^\infty)\,.
$$
This proves the Proposition once we take $R_1\geq \max(\tR_1+1,4 R_0)$.
\end{proof}

\noindent
{\bf Remark.} The statement of the lemma can be refined using 
exponential weights to give a
stronger statement about $ u_{\theta,\epsilon}  $ (including the
case of $ \epsilon = 0 $):
$$  
\| e^{ \theta |x|/ C_2 } u_{\theta , \epsilon} \|
_{ L^2 ( X_\theta \setminus B_{\CC^n} ( 0 , R_1 ) )}  = {\mathcal O}( 1) \,,
$$
see \cite{SS} for a similar argument.

\begin{lem}
\label{l:11}
Let $ K=K_E $ be the trapped set \eqref{eq:trapp} for $ p ( x, \xi ) $ 
at energy $ E $.
Suppose that $ u_{\theta, \epsilon } $ is as in Lemma \ref{l:1r}
and $ G $ and $ \epsilon $ have the properties in \eqref{eq:VG0}
and \eqref{eq:VG}.
Then for any $ \delta > 0 $ there exists $ C(\delta)>0 $ such that
\be\label{eq:normte}
\| u \|_{\theta, \epsilon} \leq C(\delta)\, \| u \|_{ L^2 ( \pi ( K ) + 
B_X ( 0 , \delta ) )} \,, \quad  0 < h \leq h_0 ( \delta ) \,.
\ee
As a consequence for any resonant state $ u = u ( h)  $
with $ \Re z-E = o(1) $, $ \Im z \geq -Ch $, we have 
\be\label{eq:normR}
 \forall \, R > 0 \,, \ \exists \, C ( \delta , R ), \, 
h_0 (\delta, R ) \,, \quad 
 \| u \|_{L^2 ( B( 0 , R ) ) } \leq C ( \delta )\, 
\| u \|_{L^2 ( \pi ( K ) +  B_X ( 0 , \delta ) ) } \,, \quad h \leq h_0 ( \delta, R )\,.
\ee
\end{lem}

This means that a normalization in any small neighbourhood of $ \pi ( K ) $
leads to an $h$-independent 
normalization in any compact set. That property allows us to define a
global measure $ \mu $ in Theorem \ref{t:2}.
\begin{proof}
Lemma \ref{l:1r} shows that, to establish 
\eqref{eq:normte}, it is enough to prove 
\be\label{eq:enough}  
\| u_{\theta, \epsilon} \|_{B_{X_\theta} (0, R_1 ) } 
 \leq C\, \| u \|_{ L^2 ( \pi ( K ) +  
B_X ( 0 , \delta ) )} \,, \quad  0 < h \leq h_0 ( \delta ) \,. 
\end{equation}
For $\delta_0>0$ small and $R_1$ as in Lemma~\ref{l:1r},
we consider the compact set
\[ 
S \defeq \overline{\cE_E^{\delta_0}\cap T^*_{B ( 0 , R_1)}X } 
\,.
\]
If $ \rho \in S \cap \Gamma_0^{+\delta_0} $, there exists
$ T_\rho \geq 0 $ and a neighbourhood of $ \rho $, $ U_\rho \subset \cE_E^{2 \delta_0}$, 
such that 
\be\label{eq:Tr}
\Phi^{-T_{\rho} } \left( U_\rho \right) \subset T^* ( \pi ( K ) + B( 0 , \delta)) \,, 
\ee
provided $ \delta_0 $ is small enough depending on $ \delta $
(so that $ K^{2\delta_0}_E \Subset T^* ( \pi ( K_E ) + B( 0 , \delta))$). 

On the other hand, if $ \rho \notin S \cap \Gamma_0^{+\delta_0} $,
there exists $ T_\rho > 0 $ and a neighbourhood of $ \rho$, 
$U_\rho \subset \cE_E^{2 \delta_0} $, such that
\be\label{eq:taur}   
\Phi^{-T_{\rho} } \left( U_\rho \right) 
\subset T^* ( X \setminus B ( 0 , 2 R_1 ) ) \,. 
\ee
Since the set $ S $ is compact, we can cover it with the union of two families of sets
$\{ U_{ \rho_j},\ j\in J_1\} $ and $\{ U_{\rho_j},\  j\in J_2\}$ of the preceding two
types, where $J_1,\ J_2$ are (disjoint) finite index sets.
We can also choose open sets $ U'_{\rho_j}\Subset U_{\rho_j} $ such that 
$\cup_{j\in J_1\cup J_2} U'_{\rho_j}$
still covers $ S $. 
We note
that these covers have different properties than the cover $(W_a)_{a\in A}$
constructed in \S\ref{s:selection}.

We now construct a ``quantum cover'' adapted to the above classical cover: 
$$  
A_j \in \Psi ( X_\theta ) ,\quad \WF_h(A_j) \Subset U_{\rho_j} \,, \qquad 
A_j=I\ \text{microlocally near }  U'_{\rho_j}\,,\quad  j\in J_1\cup J_2 \,.
$$
In view of the localization 
of $ u_{\theta, \epsilon } $ to the energy shell (see \eqref{e:microloc}) 
we have
\[  
\| u_{\theta, \epsilon }\|_{L^2(B(0,R_1))} \leq C\, 
\sum_{j\in J_1\cup J_2} \| A_j\, u_{\theta, \epsilon } \| \,. 
\]
Hence \eqref{eq:enough} will follow from the bounds
\begin{align}
\label{eq:enough1} 
\| A_j\, u_{\theta,\eps} \|  &\leq C \, 
\| u_{\theta,\eps}  \|_{ L^2 ( \pi ( K ) + B( 0 , \delta ) )} + \Oo ( h^\infty) \,,\quad j\in J_1\\
\| A_j\, u_{\theta,\eps} \| &\leq C\, h^{ M_1/C_0} \,\quad j\in J_2 \,. \label{eq:enough2} 
\end{align}
With $ U ( t ) $ defined by \eqref{eq:Uthep}, 
Proposition \ref{p:res} and the condition $|\Im z|\leq C\,h$ imply that 
for any bounded $t>0$:
\[ 
\| A_j\, u_{\theta, \epsilon} \| \leq e^{C t} \| A_j\, U ( t ) u_{\theta, \epsilon} \| 
+ \Oo ( h^\infty ) \,,\quad j\in J_1\cup J_2\,.
\]
Considering operators $tA_j\in \Psi_h ( X_\theta)$ with the properties
$$
\WF_h ( \tA_j ) \subset \Phi^{-T_{\rho_j} } ( U_{\rho_j})\,,\qquad
\tA_j=I\ \ \text{microlocally near }\Phi^{-T_{\rho_j}}(\WFh(A_j))\,,
$$
we may apply Proposition \ref{p:new1}, $ii)$:
\begin{align*}
\| A_j\,u_{\theta,\eps}\| &\leq e^{C\,T_{\rho_j}} \| A_j\,U(T_{\rho_j})\,u_{\theta,\eps}\|+\cO(h^\infty)\\
&\leq \| A_j\,U(T_{\rho_j})\,\tA_j\,u_{\theta,\eps}\|+\cO(h^\infty)\\
&\leq C'\,e^{C'\, T_{\rho_j}}\,\|\tA_j\,u_{\theta,\eps}\| +\cO(h^\infty)\,.
\end{align*}
In the last line we used the bound \eqref{eq:gaU} and $\|A_j\|\leq C'$.
Notice that the times $T_{\rho_j}$ are uniformly bounded, depending on $\delta,\ R_1$.
From \eqref{eq:Tr} we obtain the first estimate in \eqref{eq:enough1}.
Lemma \ref{l:1r} and \eqref{eq:taur} provide the second estimate.
This completes the proof of \eqref{eq:normte}.

\medskip

To see how \eqref{eq:normR} follows from \eqref{eq:normte} we choose
$ R_0 > 2 R $ in the construction of $ P_{\theta, \epsilon} $ (see 
\eqref{eq:VG0} and \eqref{eq:VG} above). From the support properties of the weight $G$,
we have the following relationship between
$ u_{\theta , \epsilon} $ and the corresponding resonant state $u$:
$$ 
\| u - u_{\theta, \epsilon} \|_{H^k_h(B ( 0 , R_0/2 )) } = \Oo (  h^\infty )  
\,
\| u_{\theta, \epsilon} \|_{ L^2 ( X_\theta ) } \,,\qquad\forall \, k. 
$$
Then 
\[ 
\| u \|_{ L^2 (B( 0 , R )) } \leq  \| u_{\theta,\eps}\|_{ L^2 (X_\theta) } ( 1 
+ {\mathcal O} ( h^\infty )  ) \leq 
C(\delta, R )\, \| u \|_{ L^2 ( \pi ( K ) + B ( 0, \delta)  ) } \,. \]
\end{proof}

The next proposition  is a refined version of \eqref{eq:t21}
appearing in Theorem \ref{t:2}:
\begin{prop}
\label{p:3}
Suppose $ u $ satisfies the assumptions of Theorem \ref{t:2}, 
and that  $ a \in 
\CIc  $ is supported in $ T^* X \setminus \Gamma_E^+ $.
Then for any $ \chi \in \CIc ( X) $ we have
\be\label{eq:l2} 
\| a^w ( x , h D )\, \chi\, u \| \leq C_M\, h^{ M} \quad\text{for any $ M>0 $,}
\ee
that is $ u \equiv 0 $ microlocally in $ T^* X \setminus \Gamma_E^+ $.
The constant $ C_M $ in \eqref{eq:l2} depends on  $ E $, $a$ and $\chi$. 
\end{prop}
\begin{proof}
We choose $ R_0$ such that $ \supp \chi\subset B(0,R_0/2) $ in the construction of 
$ P_{\theta, \epsilon } $ described in the beginning of this 
section. Then, by Lemma \ref{l:11}, 
the normalization in Theorem \ref{t:2} is, up to uniform constants,
equivalent to the normalization  $ \| u_{\theta, \epsilon } \|_{L^2 ( X_\theta)} = 1 $. 
From \eqref{eq:ute} we see that
$$ 
\| a^w ( x, hD )\, \chi\, u \| = \| a^w ( x, h D)\, \chi\, u_{\theta,\eps} \| + \Oo ( h^\infty) \,. 
$$
The condition on the support of $a$ shows that for $\delta_0>0$ small enough,
$\supp a \cap \Gamma_E^{+\delta_0}=\emptyset $. Using an energy cutoff $\psi_0$ of the form
\eqref{e:Energy0} supported
inside $\cE_E^{\delta_0}$, there exists a time $ T>0 $ such that
\[ 
\Phi^{-T} \supp (a\,\psi_0) \Subset  T^*(X \setminus B(0,2R_1)) \cap \cE_E^{\delta_0} \,, 
\]
where $ R_1 $ is given Lemma \ref{l:1r}. Taking into account the microlocalization of type
\eqref{e:microloc}, we get 
$$
\| a^w ( x, hD )\, \chi\, u \| =
\| a^w ( x, h D)\,\psi_0^w(x,h_D)\, \chi\, u_{\theta,\eps }\| + \Oo ( h^\infty) \,. 
$$
We can not proceed as in the previous lemma:
$$
\| a^w\,\psi_0^w\, \chi\, u_{\theta,\eps }\|\leq C\,\| a^w\,\psi_0^w\,\chi \,U(T)\,u_{\theta,\eps }\|
\leq C\,\| a^w\,\psi_0^w\,\chi\,U(T)\,\psi_1^w\chi_1 \chi u_{\theta,\eps }\|\,,
$$
where $\psi_1\in \CI(T^*X,[0,1])$ satisfies $\psi_1\rest_{\cE_E^{2\delta_0}}=1$, while $\chi_1\in \CI(X)$
vanishes on $B(0,R_1)$ and takes the value $1$ for $|x|\geq 2R_1$. The second line above is
then due to Proposition~\ref{p:new1}, $ii)$.
Lemma~\ref{l:1r} shows that 
$\|\chi_1\,u_{\theta,\eps }\|= \Oo ( h^{M_1/C_0 } )$, so we finally get
\[ 
\| a^w ( x, h D )\, \chi\, u \| = \Oo ( h^{M_1/C_0}  ) \,,
\]
where $ M_1 $ can be taken arbitrary large.
\end{proof}

\medskip

\noindent
{\em Proof of Theorem \ref{t:2}:} The inclusion
\eqref{eq:t21} follows directly from Proposition \ref{p:3},
which shows that only points in $ \Gamma_E^+ $ can be in the support
of the limit measure.

The proof of \eqref{eq:t22} follows the standard approach 
(see \cite{PG} and for a textbook presentation \cite[Chapter 5]{EZ}). 
Suppose that  $ \chi $ and $ u $ are as in \eqref{eq:t20}. 
From Lemma \ref{l:11} we know that $ \| \chi u ( h ) \| \leq C_\chi $
with the constant $ C_\chi $ independent of $ h$. Hence, there exists
a sequence $( h_k\searrow 0)_{k\in\NN} $ for which \eqref{eq:t20} holds for any 
$ A = a^w ( x , h D ) $, $ a \in \CIc ( T^* X ) $, 
$ \supp a  \Subset  ( \pi^*\chi) ^{-1} ( 1) $.
From this support property we get $A[P,\chi]=\Oo_{L^2\to L^2}(h^\infty)$, so that
\[ \begin{split}
\Oo ( h^\infty )  & = \Im \la ( P - z ) \chi u , A \chi u \ra 
= \Im \la P \chi u , A \chi u  \ra - \Im z \la A \chi u , \chi u \ra \\
& = \frac{h}2 \la (H_p a )^w ( x, h D )  \chi u ,  \chi u \ra
- \Im z \la A \chi u  , \chi u \ra + \Oo ( h^2 ) \| \chi u  \|^2  \,.
\end{split}
\] 
For the sequence $ (h_k) $ appearing in \eqref{eq:t20} we obtain
\[ 
\frac 12 \int H_p a \, d \mu - \frac{ \Im z ( h_k) }{ h_k } \int a \, d \mu 
 = o ( 1 ) \,, \ \ k \rightarrow \infty \,. 
\]
Hence there exists $ \lambda \geq 0 $ such that 
$ \Im z ( h_k ) / h_k  \to - \lambda/2 $, and 
$$
\int H_p a \, d \mu + \lambda \int a \, d \mu = 0  \,.
$$
which is the same as \eqref{eq:t22}.
\stopthm

\section{Resolvent estimates}\label{re}

In this section we will prove Theorem \ref{th:4} and consequently 
we assume that the hypothesis of that theorem hold throughout this
section. In particular $ E > 0 $ is an energy level at which the
pressure, $ P_E ( 1/2 ) $, is negative.
We first need a result which is a simpler version of the estimates on the propagator $ U ( t ) $
described in \S\ref{s:pressure}. 
\begin{prop}\label{p:re1}
Suppose $ W \in \CI ( X ; [ 0 , 1] ) $, $ W \geq 0 $ satisfies the following 
conditions
$$
\supp W \subset X \setminus B ( 0 , R_1 ) \,, \qquad 
W\rest_{ X \setminus B ( 0 , R_1 + r_1 )  } = 1 \,,
$$
for $ R_1 , r_1$ sufficiently large. Assume $P_E ( 1/2 )<0$ and choose 
$ \lambda \in (0 , | P_E ( 1/2 )|) $. 
Then there exists $\delta_0>0$, such that,
for any $\psi\in S(1)$ supported inside $\cE_E^{\delta_0}$, and any $ M > 0 $, 
\be\label{eq:re1}
\| e^{ - i t ( P ( h ) - i  W) / h } \psi^w ( x , h D ) \|_{ 
L^2 \to L^2 } \leq C\, h^{-n/2}\, e^{ - \lambda t }+ \Oo_M ( h^\infty ) \,,
\quad 0 \leq t \leq M \log ( 1/h ) \,.
\ee
\end{prop}
The proof of this proposition is very similar to the proof 
in the case of the complex-scaled
operator $P_{\theta,\eps}$ treated in \S\ref{qd}. In fact 
the case of the absorbing potential 
is easier to deal with than complex scaling, and in particular
we do not need the weights $ G $. The modifications needed to apply 
\S\ref{qd} directly
are given in the appendix.

Before proving \eqref{eq:t4} we will establish a resolvent estimate
for the operator with the absorbing potential.
\begin{prop}\label{c:resolvent}
Let $P=P(h)$, the energy $E>0$, and the absorbing potential $W$ be as in Proposition~\ref{p:re1}.
Then for any $ \epsilon > 0 $, 
\begin{equation}
\label{eq:phw}
\| ( P ( h ) - i W - E)^{-1} \|_{L^2 ( X ) \to L^2 ( X ) }
 \leq  \frac{ n( 1  +  \epsilon ) } { 2 | P_E ( 1/2 ) |} \,
\frac{ \log ( 1/h ) } {h} \,, \quad 0 < h < h_0 ( \epsilon ) \,.
\end{equation}
\end{prop}
\begin{proof}
We will use Proposition \ref{p:re1} and $h$-dependent complex 
interpolation similar to that in \cite{TZ}.

If we put 
\[  U_1 ( t ) \defeq \exp (- i t ( P ( h ) - i  W) / h )\,  \psi^w ( x , h D )
 \,, \]
where $ \psi $ is as in \eqref{eq:re1},
then the following estimates valid for any $ M>0 $ and $ 0 < h < h_M $:
\be\label{eq:U1}
\| U_1 ( t ) \|
\leq 
\left\{ \begin{array}{lll}  1 + {\mathcal O} ( h )\,, 
 & 0 \leq t \leq T_E\,, &  
 T_E ( h ) \defeq n  \log ( 1 / h )  /
 (2\lambda ) \,, \\
\ & \ & \ \\
C_0 h^{-n/2} e^{ - \lambda t }\,,  & T_E \leq 
t \leq T_M \,, &   T_M (h ) \defeq M \log ( 1/ h ) \,, 
\\
\ & \ & \ \\
h^{M/C_0 }\,,  & t \geq T_M  \,, & \ \end{array} \right.
\ee
where $ C_0 $ is independent of $ M $. The notation $T_E(h)$ comes from the analogy with
the {\em Ehrenfest time} (the time the system needs to delocalize a Gaussian wavepacket).

The first estimate in 
\eqref{eq:U1} follows from the subunitarity of $ \exp \big( - i t ( P ( h ) - i W ) /h\big) $ and the bound 
$ \| \psi^w\|_{L^2 \to L^2} \leq 1 + \Oo ( h ) $. 
The second estimate follows from Proposition \ref{p:re1} by 
absorbing the remainder $ \Oo_M ( h^\infty ) $ in the leading
term by taking $ h < h_M $, small enough. The last estimate
follows by writing 
$$ 
U_1  ( t) = \exp \big( - i ( t - T_M ) ( P ( h )- i W )/h \big)\, U_1 ( T_M ) \,,
$$
and using subunitarity for the first factor and the previous estimate
for $\|U_1(T_M)\|$. 

The estimates \eqref{eq:U1} and ellipticity away 
from the energy surface give the following
\begin{lem}\label{l:re1}
In the notations of Proposition \ref{p:re1} and \eqref{eq:U1} we have,
for any $ N>0$, 
\begin{gather}
\label{eq:re2}
\begin{gathered}  \| ( P ( h ) - i W - z )^{-1} \|_{L^2 \rightarrow L^2}
\leq \frac{C_1 + T_E ( h )}{h} + \frac{ h^N }{ \Im z } \,, \\ \Im z > 0 \,,  \ \ 
| z - E | < \delta \,, 
\quad  0 < h < h_N \,. 
\end{gathered}
\end{gather}
\end{lem}
\begin{proof}
We first prove the same estimate for the energy-localized operator
\[  
( P ( h ) - i W - z )^{-1} \psi^w ( x, h  D) 
= \frac1 h \int_0^\infty U_1 ( t )\, e^{  i t z / h } \,, \quad \Im z > 0 \,. \]
From \eqref{eq:U1} we obtain
\[ \begin{split} 
\|  ( P ( h ) - i W - z )^{-1} \psi^w ( x, h  D)  \|
& \leq \frac 1 h \left( \int_0^{T_E} + 
 \int_{T_E}^{T_M} + \int_{T_M}^\infty \right) 
\| U_1 ( t ) \| e^{-\Im z t /h } dt \\
& \leq T_E ( h ) / h + C/ ( h \lambda + \Im z ) + h^{M/C_0} / \Im z \,.
\end{split}
\]
This is the estimate on the right hand side of 
 \eqref{eq:re2} once we take $ M$ large enough and $h$ small enough.

To solve $ ( P - i W - z ) u = ( 1 - \psi^w ) f $, $ f \in L^2(X) $,
we follow a standard procedure. There exists $ \psi_1 \in \CIc ( T^*X; [0,1])$
supported near the energy surface $\cE_E $, such that the
pseudodifferential operator 
$ ( P - i W - z - i \psi_1^w )^{-1} $ is uniformly bounded in $L^2$ for $z$ as in the
Lemma and $h\in (0,h_0)$, while 
$ \psi_1^w ( 1 - \psi^w ) = \Oo_{L^2 \to L^2 } ( h^\infty ) $. 
It follows that
\[ 
( P - i W - z )  ( P - i W - z - i \psi_1^w )^{-1} ( 1 - \psi^w ) f 
= ( 1 - \psi^w ) f + R f \,,  
\]
where $ R = \Oo_{L^2 \to L^2 }( h^\infty )$.
If we put
\[  L \defeq   ( P - i W - z - i \psi_1^w )^{-1} ( 1 - \psi^w ) + 
( P - i W - z)^{-1} \psi^w \,, \]
then 
$$
( P - i W - z ) L = I + R \,, \quad
\| L \| \leq \frac{{C/\lambda} + T_E ( h )}{h} + \frac{ h^N }{ \Im z } \,, \quad
\| R \| = \Oo ( h^\infty )  \,,
$$
and $ ( P - i W - z )^{-1} = L ( I + R)^{-1} $ satisfies the estimate \eqref{eq:re2}.
\end{proof}

To estimate the norm of resolvent 
on the energy axis, $ \|( P - i W - E)^{-1}\|$, 
we need the following
parametric version of the maximum principle 
\begin{lem}\label{l:three}
Suppose that $ \zeta \mapsto F ( \zeta ) $ is holomorphic
in a neighbourhood of 
$ [-1,1] + i [ - c_-, c_+ ] $, for some fixed $ c_\pm > 0 $,
and that 
\begin{align*} 
\log | F ( \zeta ) | \leq M \,,\quad & \zeta \in [-1,1] + i [ - c_-, c_+ ] \,,\\
| F ( \zeta ) | \leq \alpha + \frac{\gamma}{ \Im \zeta } \,, \quad
&\zeta \in [-1, 1] + i ( 0, c_+] \,,
\end{align*}
where $ M ,  \alpha \gg 1$,  while $\gamma\ll 1$. 
Then for $ \epsilon $ satisfying 
$ \gamma M^{\frac32}/\alpha \ll \epsilon^{5/2} \ll 1 $ we have 
$$
|F ( 0 ) | \leq ( 1 + \epsilon )\, \alpha \,. 
$$
\end{lem}
\begin{proof}
Let $ g ( \zeta ) = \exp ( - 3 M \zeta^2 + i a \zeta ) $, with $ a\in \RR $
to be chosen later. Then $ g ( 0 ) = 1 $ and
$$ 
|g ( \zeta ) | \leq \exp \big( - 3 M (\Re \zeta )^2 + 3 M ( \Im \zeta )^2 + 
|a | | \Im \zeta | \big) \,. 
$$ 
Let $ 1 \gg \delta_- \gg \delta_+ > 0 $.
Then the following bounds hold on the boundary of $[-1,1] + i [ - \delta_-, \delta_+ ]$:
\[ \log |F ( \zeta ) g ( \zeta ) | \leq \left\{ \begin{array}{ll} 
- 2 M + 3 M\delta_-^2  + |a| \delta_- &  \Re z = \pm 1 \,, \ -\delta_-
\leq \Im \zeta \leq \delta_+ \,, \\
\ & \ \\
M + 3 M \delta_-^2  + a \delta_- \,, & |\Re z| \leq 1 \,, \ \Im \zeta
= - \delta_- \,, \\
\ & \ \\ 
 \log ( \alpha + \gamma / \delta_+ )  + 3 M \delta_+^2 - a \delta_+  \,, & 
 |\Re z| \leq 1 \,, \ \Im \zeta
= \delta_+ \,. \end{array} \right.
\]
Following the standard ``three-line'' argument we select
$$ 
a = \frac{1} {\delta_+ + \delta_-} ( - M + \log ( \alpha + 
\gamma/\delta_+ ) ) \simeq \frac{1}{\delta_-} ( - M + \log ( \alpha + \gamma/\delta_+ ) ) \,, 
$$
so that
the bounds for $ \Im \zeta 
= \pm \delta_\pm $, $|\Re z|\leq 1$ are the same:
\begin{align*} \log |F ( \zeta )  g (\zeta)|  &  \leq 
 M  \frac{\delta_+}{ \delta_+ + \delta_-} + 
\log( \alpha + \gamma/\delta_+) 
  \frac{\delta_-}{ \delta_+ + \delta_-} + 3 M \delta_-^2\\
& \lesssim
 M  { \delta_+}\delta_-^{-1}  + 
\log( \alpha + \gamma/\delta_+) 
 + 3 M \delta_-^2  \,. 
\end{align*}
To ensure that the above right hand side is smaller than $\log ( \alpha   ( 1 + \epsilon )  ) $, 
we need the following conditions
to be satisfied:
\[  
\delta_+ \delta^{-1}_- M \ll \epsilon \,, \quad  M \delta_-^2 \ll \epsilon \,,
 \quad \gamma/ \alpha\delta_+ \ll \epsilon \,. 
\]
These conditions can be arranged if $\eps^{5/2}$ is large enough compared with $\gamma M^{\frac32}/\alpha$,
which is the condition in the statement of the lemma. One easily checks that the bound 
$\log |F ( \zeta ) g (\zeta)|\leq \log (\alpha( 1 + \eps))$ then also holds for 
$|\Re z|=1$, $\Im z=\pm\delta_{\pm}$, and therefore for $z=0$ by the maximum principle.
\end{proof}
\noindent{\it End of the proof of Proposition~\ref{c:resolvent}:} 
To apply Lemma~\ref{l:three} we need the estimate of Lemma~\ref{l:re1}, but also
an estimate of $ \| ( P - i W - z )^{-1} \| $
for $|\Re z-E|\leq\delta$,  $ | \Im z | \leq  \lambda h $ (where we recall that
$ ( P - i W - z )^{-1} $ has no poles in that strip for $h$ small enough).  
We can cite \cite[Lemma 6.1]{DSZ} and obtain
\[ 
\| ( P - i W -z  )^{-1} \| \leq C_\epsilon \exp ( { C_\epsilon 
h^{-n-\epsilon} }) \,, \quad \Im z > - \lambda h \,. 
\]
Lemma \ref{l:three} applied to the data
\begin{gather*}
 F ( \zeta ) = \langle ( P - i W - E - h \zeta )^{-1} f , g \rangle \,, \quad
f, g \in L^2 (X ) \,, \quad \| f \|=\|g\| = 1 \,, \\
M = C h^{-n-1} \,, \quad \alpha = \frac{C_1 + T_E ( h )}{h}\,, \quad  \gamma = h^N \,, 
\end{gather*}
proves the Corollary (observe the condition $\gamma M^{\frac32}/\alpha\ll 1$ is satisfied 
for $h$ small enough).
\end{proof}

To pass from the estimate \eqref{eq:phw} to an estimate on 
$\chi ( P - z )^{-1} \chi $, $ \chi \in \CIc ( X ) $, we first
recall (see for instance \cite{TZ}) 
that if $ \supp \chi \subset B ( 0 , R_0 ) $, where $ R_0 $ is as
in \S \ref{defcs}, then 
$$
 \chi ( P - z )^{-1} \chi = \chi ( P_\theta - z )^{-1} \chi \,. 
$$
Also, if $ \supp \pi^* \chi \cap \supp G = \emptyset $, then 
\[
\begin{split}  \chi ( P_\theta - z )^{-1} \chi & = 
\chi e^{\eps G^w / h } ( P_{\theta,\eps} - z )^{-1} e^{-\eps G^w / h } \chi \\
& = \chi ( P_{\theta,\eps} - z )^{-1} \chi + 
\Oo_{L^2 \to L^2 } ( h^\infty )\, \| ( P_{\theta,\eps} - z)^{-1} \| \,. 
\end{split}
\]
Hence, 
\be\label{eq:Pth}
\|  \chi ( P_\theta - z )^{-1} \chi \| = 
( 1 + \Oo ( h^\infty ) )  \| ( P_{\theta,\eps}  - z )^{-1} \| \,. 
\ee

For future use, we now consider an auxiliary simpler scattering situation, namely
an operator $ \Ps = \Ps(h) $ satisfying 
the assumptions of \S\ref{ass} and for which the associated classical flow is {\em 
non-trapping} at energy $ E $, that is, $ K_E = \emptyset $.
From a result of Martinez \cite{Mar}, we have 
$$ 
\chi ( \Ps  - z )^{-1} \chi = \Oo ( 1/h ) \,, \quad
z \in D ( E , C h ) \,, 
$$
see \cite[Proposition 3.1]{NSZ}\footnote{The statement
of that Proposition should be corrected
to include a cut-off $ \chi $, or, without a cut-off, a factor $ \log (1/h) $
on the right hand side of \cite[(3.2)]{NSZ}. Lemma~\ref{l:b} gives a correct global
version without the logarithmic loss.}. 
Below we will need an estimate for the resolvent of
$ \Ps _{\theta, \epsilon} $, given in the next lemma.
\begin{lem}
\label{l:b}
Suppose that $ \Ps  = \Ps  ( h ) $ is an operator satisfying 
the assumptions of \S \ref{ass} and that the flow of $ p^\sharp  $ 
is {\em non-trapping} at energy $ E $, that is, $ K_E = \emptyset $. 
Then in the notation of \S \ref{s:mso}, 
\be\label{eq:b}
( \Ps _{\theta,\epsilon} - z )^{-1} = {\mathcal O}_{L^2 \to L^2} ( 1/h ) \,, \quad
z \in D ( E , C h ) \,.
\ee
\end{lem}
\begin{proof}
Since $ \Ps _{\theta,\eps} - z $ is a Fredholm operator 
on $ L^2 (X) $ (as elsewhere we identify $ X_\theta $ with $ X $),
the estimate will follow if we find $ Q ( z )$ such that, 
for $ z \in D ( E , C h ) $,
\be\label{eq:Rz}
 ( \Ps _{\theta,\epsilon} - z ) Q ( z ) = 
I + A ( z ) \,, \ \ Q ( z ) = {\mathcal O}_{L^2 \rightarrow L^2 } 
( 1/h ) \,, \ \ A ( z ) = {\mathcal O}_{L^2 \rightarrow L^2 } ( h) \,. 
\ee
We will solve this problem in two steps, away and near the energy layer $\cE_E$. 
Consider the  two nested energy cutoffs
\be\label{eq:chis}
 \psi_0( x, \xi )  = \psi (( p ( x, \xi) - E )/\delta ) \,, \quad
\psi_1 ( x , \xi ) = \psi (8( p ( x, \xi) - E )/\delta ) \,, 
\ee
where $ \psi \in \CIc ( (-2,2) , [0,1] )$ and
and $ \psi\rest_{[-1,1]} \equiv 1 $. Since $ \Ps_{\theta, \epsilon } $
is elliptic on $ \supp (1 - \psi_1) $ (that is, away from $\cE_E$), 
standard symbolic calculus (as in the proof of Lemma~\ref{l:re1}) 
provides an operator $ Q_0 ( z ) $ such that
\[ 
( \Ps _{\theta,\epsilon} - z ) Q_0 ( z ) = I - \psi^w_1 ( x , h D) + A_0( z ) \,, \quad
Q_0 ( z ) = \Oo_{L^2 \to L^2 } ( 1 ) \,, \quad  
A_0 ( z ) = \Oo_{L^2 \to L^2 } ( h) \,. 
\]
We now treat the problem near the energy layer. 
We want to produce an operator $ Q_1 ( z) $
such that 
\[ ( \Ps _{\theta,\epsilon} - z ) Q_1 ( z ) = 
\psi_1^w ( x, h D) + A_1 ( z ) \,, 
\quad Q_1 ( z ) = \Oo_{L^2 \to L^2 } ( 1/h ) \,, \quad  
A_1 ( z ) = \Oo_{L^2 \to L^2 } ( h) \,. \]
To that aim we  use the tools developed in \S \ref{s:eo} and 
consider
the energy-localized propagator
\[ 
\Us ( t ) \defeq \exp ( -  i t \tPs _{\theta,\eps}/h ),\qquad
\tPs _{\theta,\eps}  \defeq  \psi_0^w ( x , h D )\, \Ps _{\theta,\eps} \,
\psi_0^w ( x ,h D )  \,,
\]
which satisfies $\| \Us ( t ) \| \leq  e^{Ct}$ for any $t>0$.
The non-trapping assumption at energy $ E $ implies that 
\be\label{eq:eT} 
\exists \, T > 0\,, \quad  \forall \rho \in 
\cE_E^{\delta/4}\cap T^*_{B(0,3R_0)}X \,, \qquad | \pi ( \Phi^t (\rho ) | > 3 R_0 \,, \quad  t > T \,.
\ee
We claim that we can take
\be\label{eq:R1}
Q_1 ( z ) \defeq \frac{i }{ h} \int_0^T \Us ( t) \psi^w_1 ( x, h D ) e^{ i t z/h} dt \,. 
\ee
Indeed, 
\begin{gather*} ( \Ps _{\theta,\eps} - z ) Q_1 ( z )  = 
\psi^w_1 ( x, h D ) + A_1 ( z ) \,, \\
A_1 ( z) \defeq
 - \, \Us ( T ) \psi_1 ^w ( x, h D )  
+ \frac{i }{ h} \int_0^T ( \Ps _{ \theta, \epsilon} - 
 \tPs _{\theta , \epsilon}  )\, 
 \Us ( t) \,\psi^w_1 ( x, h D )\, e^{ i t z/h} dt \,. 
\end{gather*}
The escape property \eqref{eq:eT} shows that there exists a time $0<T_{\min}<T$, such that
points in $\cE_E^{\delta/4}\cap T*_{B(0,3R_0)}X$
will have escaped outside $B(0,5R_0/2)$ after $T - T_{\min}$, while
points in $\cE_E^{\delta/4}\cap T*(X\setminus B(0,3R_0))$ cannot penetrate inside $B(0,5R_0/2)$ before 
the time $T_{\min}$. In both cases, Lemma \ref{l:deep} provides the
following estimate:
\[ 
\|\Us ( T ) \psi_1 ^w ( x, h D )\|  = \Oo ( h^{M_1/C_0}) \,,
\]
for some $C_0 = C_0(T-T_{\min})$. 
On the other hand, $M_1$ can be chosen arbitrary large, in particular
we assume that $M_1/C_0 > 1$.

To analyse the second term in the definition of $ A_1 $, we use
the energy cutoff $\psi_{1/2}(\rho) \defeq \psi ( 4 ( p (\rho ) - E ) /\delta )$,
which is nested between $\psi_1$ and $\psi_0$, and write
$$
\Ps_{\theta,\eps} - \tPs_{\theta,\eps} =  
 \Ps_{\theta,\eps} ( 1- \psi^w_0)  +  ( 1 - \psi^w_0 ) \Ps _{\theta ,\eps} 
\psi^w_0 ( 1 - \psi_{1/2}^w ) + \Oo_{L^2 \to L^2 }  ( h^\infty) \,.
$$
From the support properties of the $\psi_j$ and using \eqref{e:kill}, we get
$$ 
( 1 - \psi^w_0 )\, \Us ( t )\, \psi^w_1 \,, \quad 
( 1 - \psi_{1/2}^w )\, \Us ( t)\, \psi^w_1 = \Oo_{L^2 \to L^2 }  ( h^\infty ) \,. 
$$
These estimates show that $ A_1 ( z ) = \Oo_{L^2\to L^2} ( h ) $. 

As a result, the operators $ Q ( z ) = Q_0 (z ) + Q_1 ( z ) $ and 
$ A ( z ) = A_0 ( z ) + A_1 ( z ) $ satisfy \eqref{eq:Rz} completing
the proof.
\end{proof}

\medskip

\noindent
{\em Proof of Theorem \ref{th:4}:} We now return
 to our original operator $P(h)$ with properties 
described in \S\ref{defhyp}.
As is seen from 
\eqref{eq:Pth}, it is sufficient to prove the bound
\[ ( P_{\theta,\epsilon} - E )^{-1} = \Oo_{L^2 \to L^2} ( \log ( 1/h )/ h )\,.  \]
As in the Lemma above, we will construct an approximate inverse 
\[ 
( P_{\theta,\eps} - E ) Q = I + A \,, \quad 
Q = \Oo ( \log ( 1/h ) / h ) \,, \quad  A = {\mathcal O} ( h ) \,. 
\]
We consider the cutoffs \eqref{eq:chis}. 
Once again, the operator can be easily inverted away from the energy shell. We then need
to solve 
\be
\label{eq:Q1} 
 ( P_{\theta, \epsilon} - E ) Q_1 = \psi^w_1 ( x, h D)  + A_1 \,, \ \ 
Q_1 = {\mathcal O} ( \log ( 1/h ) / h ) \,, \ \ A_1 = {\mathcal O} ( h ) \,.
\ee
We will now use our knowledge of the absorbing-potential resolvent, see Proposition~\ref{p:re1}
and Proposition~\ref{c:resolvent}: we will use the fact that the operators $P_{\theta,\eps}$ and $P-iW$ 
are very similar near the trapped set. 

Assume that  $ 1 \ll R_4 < R_3< R_2 <  R_1  < R_0/2  $, where the radius $R_1$ is used to
define the absorbing potential $W$, while $R_0$ is used in the complex deformation of $X$
(see \S\ref{defcs}), and the weight $G$ is supposed to vanish on $ \pi^{-1}B ( 0 , R_0/2 )$.
Consider the spatial cutoffs $ \chi_j \in \CIc ( X, [0,1] ) $, $ j = 1 , 2 $, satisfying
$$
\supp \chi_j \Subset B ( 0 , R_j ) \,, \quad \chi_j \rest_{ B(0,R_{j+1}) }\equiv 1\,,\quad
j=1,2,3 \,.
$$
To solve \eqref{eq:Q1}, we first put 
$$
 Q_2 \defeq \chi_1 ( P - i W - E )^{-1} \chi_2 \psi^w_1 \,.
$$
We can then compute
\be
 ( P_{\theta, \epsilon} - E ) Q_2 = \chi_2 \psi_1^w + 
[ P , \chi_1 ]  ( P - i W - E )^{-1} \chi_2 \psi^w_1 + 
\Oo_{L^2 \to L^2} (h^\infty )\,,
\ee
where the error term is due to the weight $G$, which vanishes near the supports of $\chi_j$:
$$
\chi_j\,e^{\eps G/h} = \chi_j +\Oo_{L^2\to H^k_h}(h^\infty)\,,\quad\forall k\,.
$$
On the other hand, Proposition~\ref{c:resolvent} implies that
$$
 Q_2 = \Oo_{L^2 \to L^2} ( \log ( 1/ h ) / h ) \,, \qquad 
[P , \chi_1 ]  ( P - i W - E )^{-1} \chi_2 \psi^w_1  
= \Oo_{L^2 \to L^2} ( \log ( 1/ h ) )\,. 
$$
To treat the operator on the right, we observe that the differential operator 
$[ P , \chi_1  ]$ vanishes outside $B(0,R_2)$, while $\chi_1$ vanishes outside $B(0,R_1)$.
We are thus in position to apply Lemma~\ref{l:prA}. For any $v\in L^2$, $\|v\|=1$, 
set $f\defeq \chi_2 \psi^w_1 v$. The support of $f$ is contained inside 
$B(0,R_2)$, and its wavefront set lies inside $\cE_E^{\delta}$. 
As a consequence, the state
$u\defeq ( P - i W - E )^{-1} f$ also satisfies $\WFh(u)\subset \cE_E^{\delta}$, 
and the wavefront set of the state 
$[P , \chi_1 ] u$ is contained inside  $\WFh(u)\cap T^*(X\setminus B(0,R_2))$. According to
the Lemma, 
\be\label{e:outgo}
\Phi^t\big(\WFh([P , \chi_1 ] u)\big) \cap T^*_{B(0,3R_0)}X=\emptyset\quad \text{for any }
t\geq T(R_2,R_0,E/2)\,.
\ee
Using $ T= T(R_2,R_0,E/2)$, we put
$$
Q_3 \defeq - \frac{i}{h} \int_0^T U( t ) e^{ i t E/h} 
[ P , \chi_1 ]  ( P - i W - E )^{-1} \chi_2 \psi^w_1 
= \Oo_{L^2 \to L^2 } ( \log ( 1/h ) / h ) \,.
$$
Like in the proof of Lemma \ref{l:b}, the outgoing property \eqref{e:outgo} implies that
$$
( P_{\theta, \epsilon} - E ) Q_3 = - [ P , \chi_1 ]  ( P - i W - E )^{-1} \chi_2 \psi^w_1 +
\Oo_{L^2 \to L^2 }(h^{M_1/C_0})\,.
$$
Hence, assuming $M_1\gg 1$, we have
\[  
( P_{\theta, \epsilon} - E ) ( Q_2 + Q_3 ) = 
\chi_2 \psi_1^w + \Oo_{L^2 \rightarrow L^2 } ( h ) \,. 
\]
There remains to find an approximate solution with the right hand
side given by 
$$ 
( 1 - \chi_2 ) \psi_1^w +  \Oo_{L^2 \to L^2 } ( h ) \,.
$$
Since we chose $ R_3 $ large
enough to contain $ \pi ( K_E ) $, we can choose some $1 \ll R_4<R_3$, and construct 
an operator $ \Ps  $ which is non-trapping
in the sense of Lemma \ref{l:b}, and satisfies 
$$
\Ps \rest_{X\setminus B( 0 , R_4 ) } = P \rest_{X\setminus  B ( 0 , R_4 ) }\,.
$$
From the discussion leading to \eqref{eq:Pth} follows that
\[  
\Ps _{\theta, \epsilon} \rest_{X\setminus B( 0 , R_4 ) } = 
P_{\theta,\eps}  \rest_{X\setminus B ( 0 , R_4 ) } 
+ \Oo_{L^2 \to H^k_h } ( h^\infty ) \,.\]
Using the cutoff $\chi_3$, we put
\[ 
Q_4 \defeq ( 1 - \chi_3 ) ( \Ps _{\theta, \epsilon} - E)^{-1} 
( 1 - \chi_2 ) \psi_1^w \,,
\]
and then check that
\begin{gather*}   
( P_{\theta,\eps} - E )  Q_4 = 
( 1 - \chi_2) \psi_1^w - A_4 + \Oo_{L^2 \to L^2} ( h^\infty ) \,, \quad  
Q_4 = \Oo_{L^2 \to L^2} ( 1/h ) \,,\\ 
A_4 \defeq  [ P , \chi_3  ] 
( \Ps_{\theta,\eps} - E)^{-1} ( 1 - \chi_2 ) \psi_1^w 
\,, 
 \quad  A_4 = \Oo_{L^2 \to L^2} ( 1 ) \,. 
\end{gather*}
The operator $ A_4 = \tilde \chi_2 A_4 $ where $ \tilde \chi_2 $
has the same properties as $ \chi_2 $ (in particular, $\tilde\chi_2\rest_{\supp\chi_3}\equiv 1$). 
For any $v\in L^2$, the state $A_4\,v$ will be supported inside $B(0,R_3)$, and
its wavefront set will be contained in $\cE_E^{\delta}$. 
One can thus adapt the construction of
$ Q_2 + Q_3 $ when replacing $\chi_2\psi_1^w$ by $ A_4$, to
obtain an approximate inverse $ Q_5 $ with the properties
$$  
( P_{\theta,\eps} - E ) Q_5  = A_4 + \Oo_{L^2 \to L^2} ( h ) \,, \quad 
Q_5 = \Oo_{L^2 \to L^2} ( \log ( 1/ h ) / h ) \,. 
$$
We conclude that $ Q_1 \defeq Q_2 + Q_3 + Q_4 + Q_5 $ satisfies
\eqref{eq:Q1}, which proves the Theorem. 

\stopthm

\vspace{0.2cm}
\begin{center}
{\sc Appendix}
\end{center}
\vspace{0.1cm}
\renewcommand{\theequation}{A.\arabic{equation}}
\refstepcounter{section}
\renewcommand{\thesection}{A}
\setcounter{equation}{0}

In this appendix we explain how the methods of \S \ref{qd}
apply to the case in which the deformed operator $P_{\theta,\eps}$ is
replaced by the operator with the absorbing-potential operator, $P-iW$, 
where $W$ is described in
Proposition \ref{p:re1}. The arguments are easier in the
case of $P-iW$ and the only complication comes with the following 
replacement of Lemma \ref{l:new1}:
\begin{lem}
\label{l:new1A}
Let $ W $ satisfy the conditions given in Proposition \ref{p:re1}.
Then any fixed $t>0$, the operator
\be\label{eq:uotA}
V ( t) \defeq  e^{ i t P / h } \, e^{ - i t ( P - i W ) / h} \,, 
\ee
satisfies
\be\label{eq:VtA}
 V ( t ) = ( v( t )) ^w ( x, h D) + {\mathcal O}_{L^2 \rightarrow L^2 }
( h^\infty ) \,, \ \  v(t) \in S_{1/2} ( T^* X )\,.
\ee
\end{lem}
\begin{proof}
We start as in the proof of Lemma~\ref{l:new1}:
differentiating $V(s)$ with respect to $s$ gives
\begin{gather*}
 \partial_s V( s ) = \frac{1}{h}\, a(s)^w ( x , h D ) V ( s) \,, \quad  V( 0) = I \,, \ \ 
a ( s )^w ( x , h D ) \defeq - e^{it P/h} W e^{-it P/h}  \,, 
\end{gather*}
with $ a \in S $. Let 
\[  
A ( t) \defeq \int_0^t  a ( s) ds \,, \qquad
  v_0 ( t) \defeq \exp (  A( t ) / h ) \,. 
\]
We claim that the function $ v_0 \in S_{1/2} $. If fact, by Egorov's theorem,
\[  
A = A_0 + \cO ( h ) \,, \ \ 
A_0 ( t) ( x , \xi ) = - W ( \pi ( \Phi^t ( x , \xi) )) \leq 0 \,, 
\]
hence we only need to check the claim for $ \exp ( A_0 ( t) / h ) $.
The non-negativity and the $C^2$-boundedness of $ (- A_0) $ 
imply the standard estimate
$ | \partial_{(x,\xi)}^{\alpha} A_0 | \leq 
C |A_0|^{1/2} $, $ | \alpha |=1 $, 
from which we see that for any $\beta\in\NN^n$,
\[ \begin{split} \partial^\beta \exp ( A_0 ( t) / h ) & = 
\left( \sum_{ \sum_{\ell=1}^k \beta_\ell = \beta} 
h^{-k} \prod_{\ell=1}^k \partial^{\beta_\ell} A_0 \right)
 \exp ( A_0 ( t) / h ) \\ 
& = 
 \sum_{ \sum_{\ell=1}^k \beta_\ell = \beta} \Oo (h^{-k} )
\prod_{\ell =1}^k  \left( |A_0 ( t)|^{\frac12 \delta_{1, |\beta_\ell| } }
 \exp ( A_0 ( t) /k h ) 
\right)
 \\ 
& \leq C_\beta  \sum_{ \sum_{\ell=1}^k \beta_\ell = \beta} 
 \prod_{\ell =1}^k  h^{ -1 + \frac12 \delta_{1, |\beta_\ell| } }  
 \leq C_\beta' \prod_{\ell=1}^k h^{ - \frac12 |\beta_\ell|} =
{\mathcal O} ( h^{-|\beta|/2} ) \,. 
\end{split}
\]
that is, $ v_0 ( t ) \in S_{1/2} $. 
It follows that 
\[ \begin{split} \partial_s v_0 ( s)^w ( x, h D )  & = 
 \frac1h \, ( a( s) v_0 ( s) )^w ( x , h D) \\ 
& =   
 \frac1h \, a ( s)^w ( x, h D) v_0 ( s ) ^w ( x, h D) - 
r(s)^w ( x , h D) 
 \,, \end{split} \]
where the symbolic calculus shows that $ r( s) \in h^{1/2}\, S_{1/2} $.
By Duhamel's formula,
$$
E ( t) \defeq V( t) - v_0 ( t) ^w ( x , h D) 
= \int_0^t V ( t - s ) 
r(s)^w ( x , h D ) ds  = \cO_{L^2\to L^2} ( h^{1/2} ) \,, 
$$
and
\[ \begin{split} V( t) & = v_0 ( t )^w ( x , h D) + 
\int_0^t ( v_0 ( t - s ) \# r ( s) )^w ( x , h D ) d s \\
& \qquad + 
\int_0^t \int_0^s V ( t - s - s') E ( s - s') r( s')^ w ( x, h D ) ds' ds
\\
& = v_0 ( t )^w ( x , h D) + 
\int_0^t ( v_0 ( t - s ) \# r ( s) )^w ( x , h D ) d s +
{\mathcal O}_{L^2 \rightarrow L^2 }  ( h )\,.
\end{split}
\]
The iteration of this argument gives the full expansion of a symbol $ v ( t )\in S_{1/2} $, the 
quantization of which is equal to 
$ V ( t ) $ modulo an error $ \Oo_{L^2 \to L^2} (h^\infty)$.
\end{proof}
Using this Lemma we obtain the analogues of all the results of 
\S \ref{s:eo}, for $ t \geq 0 $, 
with $ U ( t ) $ replaced by $ \exp ( - i t ( P - i W ) /h ) $,
and errors given by $ \Oo ( h^\infty ) $ instead of $\Oo(h^{M_1/C_0}$. 
The proof of 
the modified Proposition~\ref{p:crucial} is then the same, and 
Proposition~\ref{p:re1} follows from the argument presented in 
\S \ref{s:pressure}. 
For instance, here is a version of the propagation results of Proposition
\ref{p:new1} (see also Proposition \ref{p:3}):
\begin{prop}
\label{p:new1A}
Fix $T>0$. Then for any $ v=v(h) \in L^2 $, $ \| v \| = \Oo(h^{-M} ) $ (in particular,
$v$ is $h$-tempered in the sense of \eqref{eq:tempu}), 
\[ 
\WF_h ( \exp ( - i t ( P - i W )/h ) \, v ) \subset \Phi^{t} ( \WF_h ( v ) ) \,,
\]
where $ \WF_h $ is defined by \eqref{eq:defWF}.
\end{prop}
\begin{proof}
In the notation of Lemma \ref{l:new1A} we write 
\[ \exp ( - i t ( P - i W ) / h )\, v = 
 \exp ( - it P /h )\, V ( t )\, v \,, \]
and observe that the symbolic calculus on $S_{1/2}$ and \eqref{eq:VtA} give 
$ \WF_h ( V ( t ) v ) \subset \WF_h (v ) $. Indeed, if $ a ( x, h D )^w v = 
\Oo_{L^2} ( h^\infty ) $, $ a ( x , \xi ) \equiv 1  $ in a neighbourhood of $ ( x_0 , \xi_0 ) $ (that
is, $ ( x_0, \xi_0 ) \notin \WFh ( v ) $), then for any symbol $ b $ with 
$ \supp b \Subset \{ a = 1 \} $, 
$$
b^w ( x, h D ) V ( t ) = b^w  \, v ( t )^w 
\,  a^w  + \Oo_{L^2 \to L^2} ( h^\infty ) \,. 
$$
Hence $ b^w ( x, h D) V ( t ) v = \Oo_{L^2} ( h^\infty ) $,
and $ ( x_0, \xi_0 ) \notin \WF ( V (t ) v ) $.
It follows that all we need is the inclusion
\[ \WFh ( \exp ( - i t P /h )  V( t) v ) \subset \Phi^t (\WFh ( V ( t ) v ) ) 
\,,\]
and that follows from the $h$-temperedness of $V(t)v$ and Egorov's theorem.
\end{proof}

In \S \ref{re} we also need the following propagation 
result:
\begin{lem}
\label{l:prA}
Let $ P $ satisfy the general assumptions of \S \ref{ass} and $ W$
is as in Proposition \ref{p:re1}, in particular $W\rest_{B(0,R_1)}\equiv 0$.
Suppose that, for some radii $1 \ll   R_2 < R_1$, we have
\begin{gather*} 
 ( P - i W - z ) u = f \,, \quad \Im z = \Oo ( h ) \,,
\\   \|  u \| = \Oo(h^{-M}) \,, \quad  
  \| f\| = \Oo ( 1 ) \,, \quad  \supp f \Subset B ( 0 , R_2 ) \,.
\end{gather*}
Then 
\begin{gather}\label{eq:prA}
\begin{gathered}
\forall \, \eps > 0\,, \ \  \exists \, T=T(R_2,R_0,\eps) > 0, \ \ \text{s.t.} \\  
\forall \,  ( x, \xi ) \in \WF_h ( u ) \setminus T^*_{B( 0 , R_2)}X \quad\text{with}\ \ 
p(x,\xi) \geq \epsilon  \,, \\ 
| \pi ( \Phi^t ( x , \xi )) | >  3 R_0 \,, \quad  \forall t > T \,.
\end{gathered}
\end{gather}
Here $ \pi : T^*X \rightarrow X $ is the natural projection. 
In other words, $ u\rest_{ X\setminus B ( 0 , R_2 ) } $ is outgoing.
\end{lem}
\begin{proof}
The principal symbol satisfies $ \Im ( p - i W - \Re z ) \leq 0 $ 
hence we have backward propagation: 
\begin{equation}
\label{eq:WFA} 
\WF_h ( u ) \subset \Phi^t ( \WF_h ( u )) \cup \bigcup_{ 0 \leq s \leq t }
\Phi^s ( \WF_h ( f )) \,, \quad \forall \, t \geq 0 \,. 
\end{equation}
Indeed, we check that
$$  
( i h \partial_t - ( P - i W ) ) ( U ( t ) u - e^{-i t z / h} u ) 
=  e^{-i t z/h } f \,,
$$
and thus, by Duhamel's formula,
$$
e^{-itz/h} u = U ( t ) u + \frac{i}{h} \int_0^t \exp(-i(t-s)(P - iW)/h)\,e^{-isz/h}\,f\, ds \,, 
$$
from which \eqref{eq:WFA} follows by applying Proposition \ref{p:new1A}.

From ellipticity of $ P - i W - z $ in 
$ X\setminus B ( 0 , R_1 + r_1 ) $, we have 
$ \| u \|_{ L^2(X\setminus B ( 0 , R_1 + r_1 )) } = \Oo ( h^\infty ) $. 
Together with \eqref{eq:WFA}, this implies that
\[  \WFh ( u ) \subset \Gamma_+ \cup  \bigcup_{ s \geq 0 } 
\Phi^s ( \WF_h ( f )) \,, \ \ 
\Gamma_+ \defeq \{ ( x , \xi) \; : \; 
\exp ( t H_p ) ( x, \xi ) \not\rightarrow \infty \,,  \ 
t \rightarrow - \infty \} \,. \]
The assumptions on $ P $ in \S \ref{ass} (essentially the fact
that it is close to the Euclidean Laplacian near infinity) show that
for $ x ( t) \defeq \pi ( \Phi^t ( x_0, \xi_0) )  $, $ p ( x_0, \xi_0) \geq
\epsilon $, 
\be
\label{eq:prA1}
 \frac{d}{dt} |x(t)|^2 \rest_{t=0} \geq 0 \,, \ \
|x_0| > R  \, \, \Longrightarrow \, \, \frac{d}{dt} |x ( t ) 
|^2 > 0 \,, \ t \geq 0 \,,  \ee
if $ R $ is large enough. Indeed,
\[ \begin{split} \frac{d^2}{dt^2} | x( t ) |^2 & = 
2  \frac{d}{dt} \langle x ( t ) , x' ( t ) \rangle = 
2  \frac{d}{dt} \langle x ( t ) , p_\xi' ( x ( t) , \xi ( t ) ) \rangle \\
& = 
2 | p_\xi' |^2 + 2 \langle x ( t) , p''_{\xi x } 
[ p_\xi ' ] - p''_{\xi \xi} [ p_x' ] \rangle \geq  
4 \xi^2 - o(1) \langle \xi \rangle^2 \,,
\end{split} 
\]
where we used \eqref{eq:validC} to obtain 
\[ p''_{\xi x } = o ( \langle \xi \rangle |x|^{-1} ) \,, \ \ 
 p_x' = o ( \langle \xi \rangle^2 |x|^{-1} ) \,, \]
(here $ o ( 1 ) \rightarrow 0 $ as $ x \rightarrow \infty$). 
Hence $ t \mapsto
| x ( t ) |^2 $ is strictly convex and that proves \eqref{eq:prA1}.

Now observe that, for any point $\rho\in \WFh(u)\setminus T^*_{B(0,R_2)}X$, we have 
$\rho\in \Gamma_+\setminus T^*_{B(0,R_2)}X$, or $\rho\in \Phi^s(\WFh( f))$ for some $s>0$.
In both cases, there exists $1 \ll \tR_2<R_2$ and $t>0$ such that $\Phi^{-t}(\rho)\in T^*_{B(0,\tilde R_2)}$. 
Thus, the trajectory $(\Phi^s(\rho))_{s\in [-t,0]}$
has necessarily crossed the sphere $\{|x|=R_2\}$ for some $t_0$, coming from inside.
From the above discussion, the trajectory is then strictly outgoing ($d|x(s)|/ds >0$) for 
$s>t_0$. In particular, there exists a time $T=T(R_2,R_0,\eps)$ (uniform for all such $\rho$)
such that $\Phi^s(\rho)$ will be outside $B(0,3R_0)$ for $s\geq T$.
\end{proof}

\noindent
{\sc Acknowledgements.} 
We would like to thank the National Science Foundation 
for partial support. The first author is grateful to
UC Berkeley for its hospitality in April 2006; he was partially supported by
the Agence Nationale de la Recherche, under the grant ANR-05-JCJC-0107-01.
We would also like to thank Nicolas Burq for a helpful discussion of
resolvent estimates.

\end{document}